\documentclass{LMCS}

\def\doi{9(1:11)2013}
\lmcsheading%
{\doi}
{1--44}
{}
{}
{Mar.~\phantom01, 2012}
{Mar.~13, 2013}
{}

\usepackage[utf8]{inputenc}
\usepackage[english]{babel}
\usepackage[T1]{fontenc}
\usepackage{lmodern}
\usepackage{microtype}
\usepackage{amsmath,amssymb}
\usepackage{hyperref}
\usepackage{ifthen}
\usepackage{wrapfig}
\usepackage{subfigure}
\usepackage{tikz}
\usepackage{enumitem}

\usetikzlibrary{arrows}
\usetikzlibrary{shapes}
\tikzstyle{every picture}=[>=stealth']
\tikzstyle{vertex}=[draw,circle,minimum size=4ex,inner sep=1pt]
\tikzstyle{vertex2}=[draw,rounded corners,inner sep=1pt,minimum size=3ex]
\tikzstyle{vertex3}=[draw,circle,minimum size=3ex,inner sep=1pt]
\tikzstyle{vertex4}=[draw,circle,line width=0.2mm,inner sep=0mm,minimum size=4mm]
\tikzstyle{dotvertex}=[fill,circle,inner sep=3pt]

\newcommand{\alias}[2]{%
  \expandafter\let\csname #1\expandafter\endcsname\csname #2\endcsname
  \expandafter\let\csname end#1\expandafter\endcsname\csname end#2\endcsname
}
\alias{theorem}{thm}
\alias{lemma}{lem}
\alias{proposition}{prop}
\alias{corollary}{cor}
\alias{definition}{defi}
\alias{example}{exa}
\alias{remark}{rem}
\newcounter{claimcnt}[thm]
\newtheorem{sclaim}[claimcnt]{Claim}
\newenvironment{proofofclaim}
  {\begin{trivlist}\item\textit{Proof.}}
  {\end{trivlist}}
\newcommand{\varqedsymbol}{\qedsymbol}
\newcommand{\varqed}{\hfill\varqedsymbol}
\newcommand{\varqedeq}{\tag*{\varqedsymbol}}

\newcommand{\ad}[1]{\par\addvspace{\topsep}\noindent\textit{Ad #1:}~}

\newcommand{\isdef}{\mathrel{\mathop:}=}
\newcommand{\bigmid}{\;\big\vert\;}

\newcommand{\ring}[1]{\mathring{#1}}

\newcommand{\Cc}{\mathcal{C}}

\newcommand{\Ic}{\mathcal{I}}

\newcommand{\Kc}{\mathcal{K}}
\newcommand{\Lc}{\mathcal{L}}

\newcommand{\Qc}{\mathcal{Q}}

\newcommand{\nat}{\mathbb{N}}

\newcommand{\graphG}{\mathtt{G}}
\newcommand{\graphC}{\mathtt{C}}
\newcommand{\graphV}{\mathtt{V}}
\newcommand{\graphE}{\mathtt{E}}

\newcommand{\set}[1]{\{#1\}}
\newcommand{\Set}[1]{\left\{#1\right\}}
\newcommand{\card}[1]{\lvert{#1}\rvert}
\newcommand{\Card}[1]{\left\lvert{#1}\right\rvert}
\newcommand{\union}{\cup}

\newcommand{\disjunion}{\uplus}
\newcommand{\intersect}{\cap}

\newcommand{\tup}[1]{\bar{#1}}
\newcommand{\len}[1]{\lvert{#1}\rvert}

\newcommand{\quot}[2]{#1/_{#2}}
\newcommand{\eclass}[2]{#1/_{#2}}

\newcommand{\isomorphic}{\cong}

\newcommand{\dist}{\operatorname{dist}}

\newcommand{\size}{\operatorname{size}}
\newcommand{\childcount}{\operatorname{\#}}

\renewcommand{\phi}{\varphi}

\newcommand{\limplies}{\rightarrow}
\newcommand{\lequiv}{\leftrightarrow}
\newcommand{\bigland}{\bigwedge}
\newcommand{\biglor}{\bigvee}

\newcommand{\ar}{\operatorname{ar}}

\newcommand{\free}{\operatorname{free}}

\newcommand{\rank}{\operatorname{rk}}

\newcommand{\Dom}[2]{#1^{#2}}

\newcommand{\dtc}{\textsf{dtc}}
\newcommand{\dtcx}[3]{[\dtc_{\,{#1},{#2}}\;{#3}]}

\newcommand{\lrec}{\textsf{lrec}}
\newcommand{\lrecx}[5]{[\lrec_{#1,#2,#3}\;#4,\;#5]}
\newcommand{\lreceq}[6]{[\lrec_{#1,#2,#3}\;#4,\;#5,\;#6]}

\newcommand{\logic}[1]{\textup{\textsf{#1}}}
\newcommand{\FO}{\logic{FO}}
\newcommand{\FOC}{\logic{FO{$+$}C}}
\newcommand{\DTC}{\logic{DTC}}
\newcommand{\DTCC}{\logic{DTC{$+$}C}}
\newcommand{\STC}{\logic{STC}}
\newcommand{\STCC}{\logic{STC{$+$}C}}
\newcommand{\TC}{\logic{TC}}
\newcommand{\TCC}{\logic{TC{$+$}C}}

\newcommand{\LREC}{\logic{LREC}}
\newcommand{\CInfB}{\Lc_{\infty\omega}^*(\textbf{C})}
\newcommand{\CInf}{\Lc_{\infty\omega}(\textbf{C})}

\newcommand{\FP}{\logic{FP}}
\newcommand{\FPC}{\logic{FP{$+$}C}}

\newlength{\tmplenx}
\newlength{\tmpleny}
\newcommand{\problem}[3]{%
  \begin{center}
    \fbox{
      \parbox[t]{0.95\linewidth}{%
        \noindent%
        \ifthenelse{\equal{#1}{}}{}{{#1}\\[1ex]}%
        \settowidth{\tmplenx}{\textit{Eingabe:}}%
        \setlength{\tmpleny}{\linewidth}%
        \addtolength{\tmpleny}{-1.15\tmplenx}%
        \parbox[t]{\tmplenx}{\textit{Eingabe:}}\hfill\parbox[t]{\tmpleny}{#2}
        \\[1ex]
        \parbox[t]{\tmplenx}{\textit{Frage:}}\hfill\parbox[t]{\tmpleny}{#3}
      }%
    }
  \end{center}}

\newcommand{\cclassname}[1]{\textup{\textsf{#1}}}

\newcommand{\PTIME}{\cclassname{PTIME}}
\newcommand{\NP}{\cclassname{NP}}
\newcommand{\LOGSPACE}{\cclassname{LOGSPACE}}
\newcommand{\NL}{\cclassname{NL}}
\newcommand{\NC}[1][]{\ensuremath{\cclassname{NC}^{#1}}}
\newcommand{\AC}[1][]{\ensuremath{\cclassname{AC}^{#1}}}

\newcommand{\num}[2][]{\left\langle#2\right\rangle_{#1}}

\newcommand{\fail}{\textit{fail}}
\newcommand{\pre}{\operatorname{pre}}
\newcommand{\decr}{\operatorname{decr}}
\newcommand{\child}{\operatorname{child}}

\newcommand{\fiso}{\phi_{\isomorphic}}
\newcommand{\ford}{\phi_{\prec}}

\newcommand{\spn}{\operatorname{span}}

\newcommand{\ngh}{\operatorname{N^c}}

\begin{document}

\title[L-Recursion and a new Logic for Logarithmic Space]
  {L-Recursion and a new Logic\\for Logarithmic Space\rsuper*}

\author[M.~Grohe]{Martin Grohe\rsuper a}
\address{{\lsuper a}RWTH Aachen University, Germany}
\email{grohe@informatik.rwth-aachen.de}

\author[B.~Grußien]{Berit Grußien\rsuper b}
\address{{\lsuper{b,c,d}}Humboldt-Universität, Berlin, Germany}
\email{\{grussien,hernich,laubner\}@informatik.hu-berlin.de}
\thanks{{\lsuper{b,d}}Work by B.G.\ and B.L.\ was supported by the Deutsche Forschungsgemeinschaft (DFG)
  within the research training group ‘Methods for Discrete Structures’
  (GrK 1408).}

\author[A.~Hernich]{Andr\'e Hernich\rsuper c}
\address{\vskip-6 pt}

\author[B.~Laubner]{Bastian Laubner\rsuper d}
\address{\vskip-6 pt}

\ACMCCS{[{\bf  Theory of computation}]: Computational complexity and
  cryptography---Complexity classes \& Logic---Finite Model Theory} 
\titlecomment{{\lsuper*}This article is an extended version of \cite{GGH+11}.}

\begin{abstract}
  We extend first-order logic with counting by a new operator that
  allows it to formalise a limited form of recursion which can be
  evaluated in logarithmic space.  The resulting logic $\LREC$ has a
  data complexity in $\LOGSPACE$, and it defines $\LOGSPACE$-complete
  problems like deterministic reachability and Boolean formula
  evaluation.  We prove that $\LREC$ is strictly more expressive than
  deterministic transitive closure logic with counting and
  incomparable in expressive power with symmetric transitive closure
  logic $\STC$ and transitive closure logic (with or without
  counting).  $\LREC$ is strictly contained in fixed-point
  logic with counting $\FPC$. We also study an extension $\LREC_=$ of
  $\LREC$ that has nicer closure properties and is more expressive
  than both $\LREC$ and $\STC$, but is still contained in $\FPC$ and
  has a data complexity in $\LOGSPACE$.

  Our main results are that $\LREC$ captures $\LOGSPACE$ on the class
  of directed trees and that $\LREC_=$ captures $\LOGSPACE$ on the
  class of interval graphs.
\end{abstract}

\maketitle

\section{Introduction}
\label{sec:intro}

Descriptive complexity theory gives logical characterisations for most of the
standard complexity classes. For example, Fagin's Theorem~\cite{fag74} states
that a property of finite structures is decidable in \NP\ if and only
if it is definable in existential second-order logic $\Sigma_1^1$. More
concisely, we say that $\Sigma_1^1$ \emph{captures} \NP. Similarly,
Immerman~\cite{imm86} and Vardi~\cite{var82} proved that fixed-point
logic \FP\ captures \PTIME,\footnote{More precisely, Immerman and Vardi's
  theorem holds for \emph{least fixed-point logic} and the equally expressive
  \emph{inflationary fixed-point logic}. Our indeterminate \FP\ refers to
  either of the two logics. For the counting extension \FPC\ considered below,
  it is most convenient to use an inflationary fixed-point operator. See any
  of the textbooks \cite{ebbflu95,gklmsvvw07,imm99,lib04} for details.}
and Immerman~\cite{imm87} proved
that deterministic transitive closure logic \DTC\ captures
\LOGSPACE. However, these and all other known logical
characterisations of \PTIME\ and \LOGSPACE\ and all other
complexity classes below \NP\ have a serious drawback --- they only
hold on ordered structures. (An \emph{ordered structure} is a structure that
has a distinguished binary relation which is a linear order of the elements of
the structure.) The question of whether there are logical characterisations of
these complexity classes on arbitrary, not necessarily ordered structures, is
viewed as the most important open problem in descriptive complexity
theory. For the class \PTIME\, this problem goes back to Chandra and Harel's
fundamental article \cite{chahar82} on query languages for relational databases.

For \PTIME, at least partial positive results are known. The strongest
of these say that fixed-point logic with counting \FPC\ captures
\PTIME\ on all classes of graphs with excluded minors~\cite{gro10} and
on the class of interval graphs~\cite{lau10}. It is well-known that
fixed-point logic \FP\ (without counting) is too weak to capture
\PTIME\ on any natural class of structures that are not ordered. The
idea that the extension \FPC\ by counting operators might remedy the
weakness of \FP\ goes back to
Immerman~\cite{imm87a}. Together with Lander he proved that \FPC\
captures \PTIME\ on the class of trees~\cite{immlan90}. Later, Cai,
F\"urer, and Immerman~\cite{caifurimm92} proved that \FPC\ does not
capture \PTIME\ on all finite structures.

Much less is known for \LOGSPACE. In view of the results described so far, an
obvious idea is to try to capture \LOGSPACE\ with the extension \DTCC\ of
deterministic transitive closure logic \DTC\ by counting operators. However,
Etessami and Immerman~\cite{eteimm00} proved that (directed) tree isomorphism is not
definable in \DTCC, not even in the stronger transitive closure logic with
counting \TCC. Since Lindell~\cite{Lindell:Tree-Canon} proved that tree
isomorphism is decidable in \LOGSPACE, this shows that \DTCC\ does not capture
\LOGSPACE.

We introduce a new logic $\LREC$ and prove that it captures \LOGSPACE\ on
directed trees. An extension $\LREC_=$ captures $\LOGSPACE$ on the class of
interval graphs (and on the class of undirected trees). The logic
\LREC\ is an extension of first-order logic with counting by a ``limited
recursion operator''. The logic is more complicated than the transitive
closure and fixed-point logics commonly studied in descriptive complexity, and
it may look rather artificial at first sight. To explain the motivation for
this logic, recall that fixed-point logics may be viewed as extensions of
first-order logic by fixed-point operators that allow it to formalise
recursive definitions in the logics. \LREC\ is based on an analysis of the
amount of recursion allowed in logarithmic space computations. The idea of the
limited recursion operator is to control the depth of the recursion by a
``resource term'', thereby making sure that we can evaluate the recursive definition in
logarithmic space. Another way to arrive at the logic is based on an analysis
of the classes of Boolean circuits that can be evaluated in \LOGSPACE. We will
take this route when we introduce the logic in Section~\ref{sec:lrec}.

$\LREC$ is easily seen to be (semantically) contained in \FPC. We show that \LREC\
contains \DTCC, and as \LREC\ captures \LOGSPACE\ on directed trees, this containment
is strict. Moreover, \LREC\ is not contained in \TCC. Then we prove that
undirected graph reachability is not definable in \LREC. Hence \LREC\ does not
contain transitive closure logic \TC, not even in its symmetric variant \STC,
and therefore \LREC\ is strictly contained in \FPC.

It can be argued that our proof of the inability of \LREC\ to express
graph reachability reveals a weakness in our definition of the logic
rather than a weakness of the limited recursion operator underlying
the logic: \LREC\ is not closed under (first-order) logical
reductions. To remedy this weakness, we introduce an extension $\LREC_=$
of \LREC. It turns out that undirected graph reachability is definable
in $\LREC_=$ (this is a convenient side effect of the definition and
not a deep result). Thus $\LREC_=$ strictly contains symmetric transitive
closure logic with counting. We prove that $\LREC_=$
captures \LOGSPACE\ on the class of interval graphs. To complete the
picture, we prove that plain $\LREC$, even if extended by a symmetric
transitive closure operator, does not capture \LOGSPACE\ on the class
of interval graphs.

The paper is organised as follows: After giving the necessary preliminaries in
Section~\ref{sec:basics}, in Section~\ref{sec:lrec} we introduce the logic
\LREC\ and prove that its data complexity is in \LOGSPACE. Then in
Section~\ref{sec:tree-canon}, we prove that directed tree isomorphism and canonisation
are definable in \LREC. As a consequence, \LREC\ captures \LOGSPACE\ on
directed trees. In Section~\ref{sec:reach-nondef}, we study the expressive power of
\LREC\ and prove that undirected graph reachability is not definable in
\LREC. The extension $\LREC_=$ is introduced in
Section~\ref{sec:lrec-eq}. Finally, our results on interval graphs are
presented in Section~\ref{sec:interval}. We close with a few concluding
remarks and open problems.

\section{Basic Definitions}
\label{sec:basics}

$\nat$ denotes the set of all non-negative integers.
For all $m,n \in \nat$,
we let $[m,n] \isdef \{p \in \nat \mid m \leq p \leq n\}$
and $[n] \isdef [1,n]$.
Mappings $f\colon A \to B$
are extended to tuples $\tup{a} = (a_1,\dotsc,a_k)$ over $A$
via $f(\tup{a}) \isdef (f(a_1),\dotsc,f(a_k))$.
Given a tuple $\tup{a} = (a_1,\dotsc,a_k)$,
let $\tilde{a} \isdef \set{a_1,\dotsc,a_k}$.
If $\sim$ is an equivalence relation on a set $A$,
we denote by $a/_\sim$ the equivalence class of an element $a$
with respect to $\sim$,
and by $A/_\sim$ the quotient of $A$ with respect to $\sim$.

A \emph{vocabulary} is a finite set $\tau$ of relation symbols,
where each $R \in \tau$ has a fixed arity $\ar(R)$.
A \emph{$\tau$-structure} $A$ consists of a non-empty finite set $V(A)$,
its \emph{universe},
and for each $R \in \tau$ a relation $R(A) \subseteq V(A)^{\ar(R)}$.
For logics $\logic{L},\logic{L}'$ we write $\logic{L} \leq \logic{L}'$
if $\logic{L}$ is semantically contained in $\logic{L}'$,
and $\logic{L} < \logic{L}'$ if this containment is strict.

All logics considered in this paper are extensions of
\emph{first-order logic with counting ($\FOC$)};
see, e.g., \cite{ebbflu95,gklmsvvw07,imm99,lib04,imm87}
for a detailed discussion of $\FOC$ and its extensions.
$\FOC$ extends first-order logic by a counting operator
that allows for counting the cardinality of $\FOC$-definable relations.
It lives in a two-sorted context,
where structures $A$ are equipped with a \emph{number sort}
$N(A) \isdef [0,\card{V(A)}]$.
$\FOC$-variables are either \emph{structure variables}
that range over the universe of a structure,
or \emph{number variables} that range over the number sort.
For each variable $u$,
let $\Dom{A}{u} \isdef V(A)$ if $u$ is a structure variable,
and $\Dom{A}{u} \isdef N(A)$ if $u$ is a number variable.
Tuples $(u_1,\dotsc,u_k)$ and $(v_1,\dotsc,v_\ell)$ of variables
are \emph{compatible} if $k = \ell$,
and for every $i \in [k]$ the variables $u_i$ and $v_i$ have the same type.
Let
$
  \Dom{A}{(u_1,\dotsc,u_k)} \isdef
  \Dom{A}{u_1} \times \dotsb \times \Dom{A}{u_k}.
$
An \emph{assignment in $A$} is a mapping $\alpha$
from the set of variables to $V(A) \union N(A)$,
where for each variable $u$ we have $\alpha(u) \in \Dom{A}{u}$.
For tuples $\tup{u} = (u_1,\dotsc,u_k)$ of variables
and $\tup{a} = (a_1,\dotsc,a_k) \in \Dom{A}{\tup{u}}$,
the assignment $\alpha[\tup{a}/\tup{u}]$
maps $u_i$ to $a_i$ for each $i \in [k]$,
and each variable $v \not \in \tilde{u}$ to $\alpha(v)$.

$\FOC$ is obtained by extending first-order logic with the following
formula formation rules:
$p \leq q$ is a formula for all number variables $p,q$;
and $\#\tup{u}\,\psi = \tup{p}$ is a formula
for all tuples $\tup{u}$ of variables,
all tuples $\tup{p}$ of number variables,
and all formulae $\psi$.
Free variables are defined in the obvious way,
with
$
  \free(\#\tup{u}\,\psi = \tup{p})
  \isdef (\free(\psi) \setminus \tilde{u}) \union \tilde{p}.
$
Formulae $\#\tup{u}\,\psi = \tup{p}$ hold in a structure $A$
under an assignment $\alpha$ in $A$ if
$
  \card{
    \set{\tup{a} \in \Dom{A}{\tup{u}} \mid
      (A,\alpha[\tup{a}/\tup{u}]) \models \psi}
  }
  =
  \num[A]{\alpha(\tup{p})},
$
where for tuples $\tup{n} = (n_1,\dotsc,n_k) \in N(A)^k$
we let $\num[A]{\tup{n}}$ be the number
\[
  \num[A]{\tup{n}}\ \isdef\
  \sum_{i=1}^k\, n_i \cdot (\card{V(A)}+1)^{i-1}.
\]
If $A$ is understood from the context,
we write $\num{\tup{n}}$ instead of $\num[A]{\tup{n}}$.

We write $\phi(u_1,\dotsc,u_k)$ to denote a formula $\phi$
with $\free(\phi) \subseteq \set{u_1,\dotsc,u_k}$.
Given a formula $\phi(u_1,\dotsc,u_k)$, a structure $A$
and $a_1,\dotsc,a_k \in A^{(u_1,\dotsc,u_k)}$,
we write $A \models \phi[a_1,\dotsc,a_k]$ if $\phi$ holds in $A$
with $u_i$ assigned to the element $a_i$, for each $i \in [k]$.
We use similar notation for substitution:
For a tuple $(v_1,\dotsc,v_k)$ of variables that is compatible
with $(u_1,\dotsc,u_k)$,
we let $\phi(v_1,\dotsc,v_k)$ be the result of substituting $v_i$
for $u_i$ for every $i \in [k]$.
We write $\phi[A,\alpha;\tup{u}]$
for the set of all tuples $\tup{a} \in \Dom{A}{\tup{u}}$
with $(A,\alpha[\tup{a}/\tup{u}]) \models \phi$.

In many places throughout this paper we refer to various transitive
closure and fixed-point logics (all mentioned in the
introduction). Our results and remarks about the relation between
these logics and our new logics \LREC\ and $\LREC_=$ are relevant for
a reader familiar with descriptive complexity theory to put our
results in context, but they are not essential to follow the technical
core of this paper. Therefore, we omit the definitions and refer the
reader to the textbooks \cite{ebbflu95,gklmsvvw07,imm99,lib04}
and the paper \cite{imm87}.

\section{The Logic \texorpdfstring{$\LREC$}{LREC}}
\label{sec:lrec}

In this section,
we introduce $\LREC$ as a first step towards the logic $\LREC_=$,
to be introduced in Section~\ref{sec:lrec-eq}.
$\LREC$ is already expressive enough to capture $\LOGSPACE$
on directed trees,
but still lacks several important properties.
For example, it is unable to capture $\LOGSPACE$ on undirected trees
and interval graphs (cf.\ Remark~\ref{rem:interval-canon-and-LREC}),
and is not closed under first-order reductions (Section~\ref{sec:lrec-eq}).
On the other hand,
although $\LREC_=$ could have been introduced without the detour via $\LREC$,
its definition is much easier to grasp by developing an understanding
of $\LREC$ first.

Let us start our development of $\LREC$
by looking at how certain kinds of Boolean circuits
can be evaluated in $\LOGSPACE$.


\begin{wrapfigure}[9]{r}{4.25cm}
  \centering
  \begin{tikzpicture}[scale=0.75]
    \tikzstyle{every node}=
      [draw,circle,line width=0.2mm,inner sep=0mm,minimum size=4mm]
    \node {$\wedge$} [level distance=12mm,sibling distance=14mm]
    child[->] {node {$\vee$}
      child[->] {node[rectangle] {$1$} }
      child[->] {node[rectangle] {$0$} }
    }
    child[->] {node[rectangle] {$1$}}
    child[->] {node {$\neg$}
      child[->] {node {$\wedge$}[sibling distance=9mm]
        child[->] {node[rectangle] {$1$}}
        child[->] {node[rectangle] {$0$}}
        child[->] {node[rectangle] {$1$}}
        child[->] {node[rectangle] {$1$}}
      }
    };
  \end{tikzpicture}
\end{wrapfigure}
The figure on the right shows a \emph{Boolean formula},
i.e., a Boolean circuit whose underlying graph is a \emph{tree}.
It is easy to evaluate such circuits in $\LOGSPACE$:
Start at the output node,
determine the value of the first child recursively,
then determine the value of the second child,
and so on.
We only have to store the current node and its value
(if it has been determined already),
since the parent node and the next child of the parent (if any)
are uniquely determined by the current node.
It is known that Boolean formula evaluation is complete for $\LOGSPACE$
under $\NC[1]$-reductions \cite{formula-eval}.%
\footnote{Boolean formula evaluation is only complete for $\LOGSPACE$
  if input formulae are represented as graphs
  (e.g., by the list of all edges plus gate types).
  It was however shown in \cite{Buss90anoptimal}
  that the problem is complete for $\NC[1]$ under $\AC[0]$-reductions
  if input formulae are given by their natural string encoding.}
In contrast, Boolean \emph{circuit} evaluation is $\PTIME$-complete.

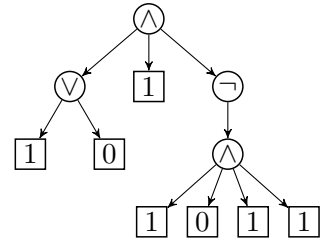
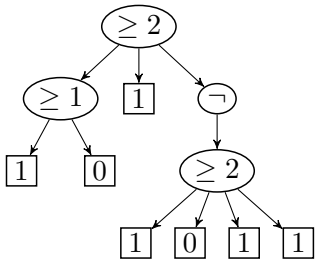
\begin{wrapfigure}[10]{l}{4.5cm}
  \centering
  \begin{tikzpicture}[scale=0.8]
    \tikzstyle{every node}=
      [draw,ellipse,line width=0.2mm,inner sep=0.5mm,minimum size=4mm]
    \node {$\ge 2$} [level distance=12mm,sibling distance=13mm]
      child[->] {node {$\ge 1$}
        child[->] {node[rectangle] {$1$} }
        child[->] {node[rectangle] {$0$} }
      }
      child[->] {node[rectangle] {$1$}}
      child[->] {node {$\neg$}
        child[->] {node {$\ge 2$}[sibling distance=9mm]
          child[->] {node[rectangle] {$1$}}
          child[->] {node[rectangle] {$0$}}
          child[->] {node[rectangle] {$1$}}
          child[->] {node[rectangle] {$1$}}
        }
      };
  \end{tikzpicture}
\end{wrapfigure}
Let us now turn to formulae with \emph{threshold gates}, which,
in addition to Boolean gates,
may contain gates of the form ``$\geq i$'' for a number $i$;
such a gate outputs 1 if, and only if, at least $i$ input gates are set to 1.
An example is shown on the left.
To evaluate such formulae in $\LOGSPACE$,
we again start at the root and evaluate the values of the children recursively.
For each node we count how many 1-values we have seen already.
To this end, when evaluating the values of the children of a node $v$,
we begin with the child with the largest subtree and proceed to children
with smaller subtrees.
Note that the $i$th child of $v$ in this order
has a subtree of size at most $s/i$,
where $s$ is the size of the subtree of $v$.
So, we can store a counter of up to $\log_2 i$ bits
for the number of 1-values seen so far.
It is easy to extend the algorithm to formulae
with other \emph{arithmetic gates} such as \emph{modulo-gates}.

\begin{wrapfigure}[14]{r}{4cm}
  \centering
  \begin{tikzpicture}[scale=0.75]
    \tikzstyle{every node}=
      [draw,ellipse,line width=0.2mm,inner sep=0.5mm,minimum size=4mm]
    \tikzstyle{every edge}=
      [draw,->,line width=0.2mm]
          \draw (1.8,6.5) node(1) {$\wedge$}
            (2.4,5.5) node(2) {$\vee$}
            (0.6,4.5) node(3) {$\vee$}
            (3,4.5) node(4) {$\wedge$}
            (0,3) node(5) {$\ge 2$}
            (1.2,3) node(6) {$\wedge$}
            (2.4,3) node(7) {$\neg$}
            (3.6,3) node(8) {$\vee$}
            (0,1.5) node(9) {$\neg$}
            (1.2,1.5) node(10) {$\wedge$}
            (2.4,1.5) node(11) {$\ge 2$}
            (3.6,1.5) node(12) {$\wedge$}
            (0,0) node(13) [rectangle] {$0$}
            (1.3,0) node(14) [rectangle] {$1$}
            (2.4,0) node(15) [rectangle] {$1$}
            (3.6,0) node(16) [rectangle] {$1$};

      \path (1) edge (2) edge (3)
            (2) edge (3) edge (4)
            (3) edge (5) edge (6)
            (4) edge (7) edge (8)
            (5) edge (9) edge (10) edge (11)
            (6) edge (11) edge (12)
            (7) edge (11)
            (8) edge (11) edge (12)
            (9) edge (13)
            (10) edge (13) edge (14)
            (11) edge (14) edge (15) edge (16)
            (12) edge (15) edge (16);
  \end{tikzpicture}
\end{wrapfigure}
As a more complicated example, let us consider the following type of circuit.
A circuit $C$ has the \emph{$m$-path property} if for all paths $P$ in $C$
the product of the in-degrees of all but the first node on $P$ is at most $m$.
For example, formulae have the 1-path property,
whereas the circuit on the right has the 16-path property.
It is not hard to see that for every $k \geq 1$,
circuits $C$ having the $\card{C}^k$-path property
can be evaluated in $\LOGSPACE$.
The idea here is very similar to the one for evaluating circuits
with threshold gates.
We start at the root node and evaluate the children recursively.
After ``entering'' a node $v$ from one of its parent nodes, say $p(v)$,
we check whether $v$ evaluates to 1
by counting the number of children that evaluate to one
using the above-mentioned strategy,
and return with this information to $p(v)$.
In order to return to $p(v)$, we need to remember $p(v)$,
which we do by storing the index of $p(v)$ among all the in-neighbours of $v$.
This requires only $\log_2 d^-(v)$ bits of storage,
where $d^-(v)$ denotes the in-degree of $v$.
The space for writing down the index of the predecessor $p(v)$
for each vertex $v$ on the path from the root to the currently visited vertex
is thus bounded by the sum of the logarithms of the in-degrees
of the vertices $v$ on that path.
Since $C$ has the $\card{C}^k$-path property,
this sum is bounded by $\log_2 \card{C}^k$,
and thus logarithmic in the size of $C$.
Another way of evaluating the circuit is to first ``unravel'' the circuit
to a tree (i.e., a formula)
which can be done in $\LOGSPACE$ due to the $\card{C}^k$-path property,
and then to evaluate the formula as above.

The logic $\LREC$ allows it to recursively define sets $X$ of tuples
based on graphs $G$
that have the $\card{G}^k$-path property for some $k \geq 1$.

We turn to the formal definition of the logic $\LREC$.
To define the syntax, let $\tau$ be a vocabulary.
The set of all $\LREC[\tau]$-formulae
is obtained by extending the formula formation rules of $\FOC[\tau]$
by the following rule:
If $\tup{u},\tup{v},\tup{w}$ are compatible tuples of variables,
$\tup{p},\tup{r}$ are non-empty tuples of number variables,
and $\phi_{\graphE}$ and $\phi_{\graphC}$
are $\LREC[\tau]$-formulae,
then
\begin{equation}
  \label{eq:lrec}
  \phi \,\isdef\,
  \lrecx{\tup{u}}{\tup{v}}{\tup{p}}{\phi_{\graphE}}{\phi_{\graphC}}(\tup{w},\tup{r})
\end{equation}
is an $\LREC[\tau]$-formula,
and we let
$
  \free(\phi) \isdef (\free(\phi_{\graphE}) \setminus (\tilde{u} \union \tilde{v}))
  \union (\free(\phi_{\graphC}) \setminus (\tilde{u} \union \tilde{p}))
  \union \tilde{w} \union \tilde{r}.
$

To define the semantics of $\LREC[\tau]$-formulae,
let $A$ be a $\tau$-structure
and $\alpha$ an assignment in $A$.
The semantics of $\LREC[\tau]$-formulae that are not of the form
\eqref{eq:lrec} is defined as usual.

Let $\phi$ be an $\LREC[\tau]$-formula of the form \eqref{eq:lrec}.
We define a set $X \subseteq \Dom{A}{\tup{u}} \times \nat$
recursively as follows.
We consider
$
  \graphE \isdef \phi_{\graphE}[A,\alpha;\tup{u},\tup{v}]
$
as the edge relation of a directed graph $\graphG$
with vertex set $\graphV \isdef \Dom{A}{\tup{u}}$.
Moreover, for each vertex $\tup{a} \in \graphV$ we think of the set
$
  \graphC({\tup{a}}) \isdef
  \set{\num{\tup{n}} \mid \tup{n} \in \phi_{\graphC}[A,\alpha[\tup{a}/\tup{u}];\tup{p}]}
$
of integers as the label of $\tup{a}$.
Let
$
  \tup{a} \graphE \isdef \set{\tup{b} \in \graphV \mid \tup{a}\tup{b} \in \graphE}
$
and
$
  \graphE \tup{b} \isdef \set{\tup{a} \in \graphV \mid \tup{a}\tup{b} \in \graphE}.
$
Then, for all $\tup{a} \in \graphV$ and $\ell \in \nat$,
\begin{align*}
  (\tup{a},\ell) \in X
  \ :\Longleftrightarrow\
  \ell > 0
  \ \, \text{and}\ \,
  \Card{
    \Set{
      \tup{b} \in \tup{a} \graphE
      \ \bigg\vert\
      \left(
        \tup{b},
        \left\lfloor\frac{\ell-1}{\card{\graphE \tup{b}}}\right\rfloor
      \right)
      \in X
    }
  } \in \graphC(\tup{a}).
\end{align*}
Notice that $X$ contains only elements $(\tup{a},\ell)$ with $\ell > 0$.
Hence, the recursion eventually stops at $\ell = 0$.
We call $X$ the \emph{relation defined by $\phi$ in $(A,\alpha)$}.
Finally, we let
\[
  (A,\alpha) \models \phi
  \ :\Longleftrightarrow\
  \bigl(\alpha(\tup{w}),\num{\alpha(\tup{r})}\bigr) \in X.
\]

\begin{example}[Boolean circuit evaluation]
  \label{ex:circuits}
  Let $\sigma \isdef \set{E,P_\land,P_\lor,P_\lnot,P_0,P_1}$.
  A Boolean circuit $C$ may be viewed as a $\sigma$-structure,
  where $E(C)$ is the edge relation of $C$,
  and $P_\star(C)$ contains all $\star$-gates
  for $\star \in \set{\land,\lor,\lnot,0,1}$.
  Suppose $C$ has the $\card{C}$-path-property.
  Then,
  \[
    \phi(z) :=
    \exists r_1,r_2\, (\lrecx{x}{y}{p}{\phi_{\graphE}}{\phi_{\graphC}}(z,(r_1,r_2)) \land
    \forall r (r \leq r_1 \land r \leq r_2))
  \]
  with $\phi_{\graphE}(x,y) \isdef E(x,y)$ and
  \[
    \phi_{\graphC}(x,p)
    \isdef
    (P_\land(x) \land \# y\, E(x,y) = p)
    \lor
    (P_\lor(x) \land \text{``$p > 0$''})
    \lor
    (P_\lnot(x) \land \text{``$p = 0$''})
    \lor
    P_1(x)
  \]
  states that gate $z$ evaluates to 1.

  For example,
  let $C$ be the first circuit at the beginning of this section,
  and let $\alpha$ be the assignment in $C$ mapping $z$ to the root of $C$,
  $r_1$ to $4$, and $r_2$ to 0.
  Figure~\ref{fig:circuits-example-graph} shows the graph $\graphG = (\graphV,\graphE)$
  with $\graphV := \Dom{C}{x}$, $\graphE := \phi_{\graphE}[C,\alpha;x,y]$,
  and labels defined by $\phi_{\graphC}$.
  \begin{figure}
    \centering
    \begin{tikzpicture}[yscale=0.7,xscale=1.25]
      \tikzstyle{vertex2}=
        [draw,circle,line width=0.2mm,inner sep=0mm,minimum size=4mm]
      \node[vertex2,label={right:$\set{3}$}] {$a$}
        [level distance=12mm,sibling distance=14mm]
        child[->] {
          node[vertex2,label={left:$[11]$}] {$b$}
            child[->] {node[vertex2] (a) {$e$} }
            child[->] {node[vertex2] (b) {$f$} }
        }
        child[->] {
          node[vertex2] (c) {$c$}
        }
      child[->] {
        node[vertex2,label={right:$\set{0}$}] {$d$}
        child[->] {
          node[vertex2,label={right:$\set{4}$}] {$g$} [sibling distance=9mm]
          child[->] {node[vertex2] (d) {$h$}}
          child[->] {node[vertex2] (e) {$i$}}
          child[->] {node[vertex2] (f) {$j$}}
          child[->] {node[vertex2] (g) {$k$}}
        }
      };
      \foreach \name in {a,c,d,f,g}
        \node[below=5pt] at (\name) {$[0,11]$};
      \foreach \name in {b,e}
        \node[below=5pt] at (\name) {$\emptyset$};
    \end{tikzpicture}
    \caption{The graph $\graphG$ from Example~\ref{ex:circuits}.
      Each vertex is labelled with a subset of $[0,11]$.}
    \label{fig:circuits-example-graph}
  \end{figure}
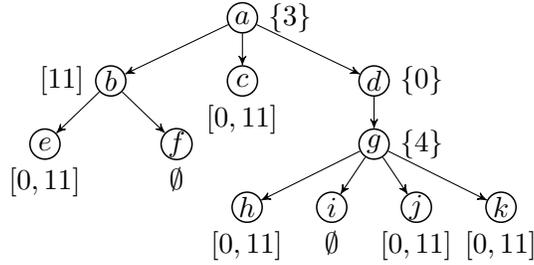
  The vertices $a$--$k$ of $\graphG$ are precisely the vertices of $C$,
  and each vertex is labelled with a subset of $N(C) = [0,11]$.
  Let $X$ be the relation
  defined by $\lrecx{x}{y}{p}{\phi_{\graphE}}{\phi_{\graphC}}(z,(r_1,r_2))$ in $(C,\alpha)$.
  For a leaf $v$ of $\graphG$,
  we have $(v,1) \in X$ (and, in fact, $(v,\ell) \in X$ for any $\ell > 0$)
  if and only if $0$ occurs in the label of $v$.
  Hence, $(v,1) \in X$ for $v \in \set{c,e,h,j,k}$,
  but $(f,1) \notin X$ and $(i,1) \notin X$.
  Since $(e,1) \in X$ and $1$ occurs in the label of $b$,
  we also have $(b,2) \in X$;
  as for the leaves,
  we also have $(b,\ell) \in X$ for any $\ell \geq 2$.
  However, note that $(g,2) \notin X$
  (and, in fact, $(g,\ell) \notin X$ for all $\ell > 0$),
  because there are only three children $v$ of $g$ with $(v,1) \in X$,
  but 3 does not appear in the label of $g$.
  Consequently, $(d,3) \in X$.
  Since we now have $(b,3) \in X$, $(c,3) \in X$, and $(d,3) \in X$,
  we have $(a,4) \in X$,
  and therefore $(C,\alpha) \models \phi$.

  While for the circuit $C$ above,
  we could have replaced the tuple $(r_1,r_2)$ in the formula $\phi$
  by a single number variable $r$,
  it is not hard to construct circuits $C$
  which have the $\card{C}$-path property,
  but the single number variable $r$ does not suffice.
  \varqed
\end{example}

\begin{example}[Deterministic transitive closure]
  \label{ex:dtc}
  Let $G=(V,E)$ be a directed graph and $a,b\in V$.
  Then there is a \emph{deterministic path} from $a$ to $b$ in $G$
  if there exists a path $v_1,\dotsc,v_n$ from $a=v_1$ to $b=v_n$ in $G$
  such that for every $i \in [n-1]$,
  $v_{i+1}$ is the unique out-neighbour of $v_i$.
  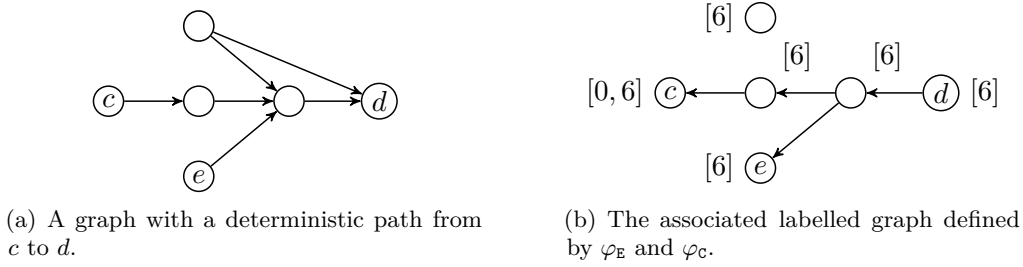
\begin{figure}
    \subfigure[A graph with a deterministic path from $c$ to $d$.]
    {
      \label{fig:det-path}
      \hspace{2.25em}
      \begin{tikzpicture}[xscale=1.2,rotate=90]
        \tikzstyle{place}=
          [draw,circle,line width=0.2mm,inner sep=0.5mm,minimum size=4.0mm]
        \tikzstyle{every edge}=[draw,->,line width=0.2mm]
        \tikzstyle{every label}=[blue]
        \draw(1,2) [color=black!20] node(6)[place] {};
        \draw (-1,2) node(5)[place] {$e$}
              (0,3) node(4) [place]{$c$}
              (0,2) node(3) [place]{}
              (0,1) node(2)[place] {}
              (0,0) node(1) [place]{$d$}
              (1,2) node(6)[place] {};
        \path (4) edge (3)
              (3) edge (2)
              (2) edge (1)
              (5) edge (2)
              (6) edge (2) edge (1);
      \end{tikzpicture}
      \hspace{2.25em}
    }
    \hspace{2.25em}
    \subfigure[The associated labelled graph defined by $\phi_{\graphE}$ and $\phi_{\graphC}$.]
    {
      \label{fig:det-path-transformed}
      \begin{tikzpicture}[xscale=1.2,rotate=90]
        \tikzstyle{place}=
          [draw,circle,line width=0.2mm,inner sep=0.5mm,minimum size=4.0mm]
        \tikzstyle{every edge}=[draw,->,line width=0.2mm]

        \draw (1,2) node(6)[place,label={left:$[6]$}] {};
        \draw (-1,2) node(5)[place,label={left:$[6]$}] {$e$}
              (0,3) node(4) [place,label={left:$[0,6]$}]{$c$}
              (0,2) node(3) [place,label={above right:$[6]$}]{}
              (0,1) node(2)[place,label={above right:$[6]$}] {}
              (0,0) node(1) [place,label={right:$[6]$}]{$d$};

	    \path (3) edge (4)
              (2) edge (3)
              (1) edge (2);
	    \path (2) edge (5);
      \end{tikzpicture}
    }
    \caption{A graph with a deterministic path,
      and the labelled graph defined by the formulae $\phi_{\graphE}$ and $\phi_{\graphC}$
      in Example~\ref{ex:dtc} from that graph.}
  \end{figure}
  Figure~\ref{fig:det-path} shows a directed graph
  with a deterministic path from $c$ to $d$.

  Let $\psi(\tup{u},\tup{v})$ be an $\LREC[\tau]$-formula,
  and let $\tup{s},\tup{t}$ be tuples of variables
  such that $\tup{u},\tup{v},\tup{s},\tup{t}$ are pairwise compatible.
  We devise a formula $\phi(\tup{s},\tup{t})$ such that
  for any $\tau$-structure $A$ and assignment $\alpha$ in $A$,
  we have $(A,\alpha) \models \phi(\tup{s},\tup{t})$
  iff in the graph $G = (V,E)$ defined by $V \isdef \Dom{A}{\tup{u}}$
  and $E \isdef \psi[A,\alpha;\tup{u},\tup{v}]$
  there is a deterministic path
  from $\alpha(\tup{s})$ to $\alpha(\tup{t})$.
  Note that there is such a path
  precisely if, in the graph obtained from $G$ by reversing the edges,
  there is a path $v_n,\dotsc,v_1$ from $\alpha(\tup{t})$ to $\alpha(\tup{s})$
  such that for every $i \in [n-1]$,
  $v_{i+1}$ is the unique in-neighbour of $v_i$.
  Therefore, we can choose $\phi$ like this:
  \begin{align}
    \label{eq:dtc}
    \phi \ \isdef\
    \exists \tup{r}\,
    \lrecx{\tup{v}}{\tup{u}}{\tup{p}}
      {\phi_{\graphE}(\tup{v},\tup{u})}{\phi_{\graphC}(\tup{v},\tup{p})}(\tup{t},\tup{r}),
  \end{align}
  where $\tup{p}$ and $\tup{r}$ are $\len{\tup{u}}$-tuples
  of number variables, and
  \begin{align*}
    \phi_{\graphE}(\tup{v},\tup{u})
    & \,\isdef\,
      \psi(\tup{u},\tup{v}) \land
      \forall \tup{v}' (\psi(\tup{u},\tup{v}') \limplies \tup{v}' = \tup{v}),
    &
    \phi_{\graphC}(\tup{v},\tup{p})
    & \,\isdef\,
      \tup{v}=\tup{s} \lor (\tup{v} \neq \tup{s} \land \tup{p} \neq \tup{0}).
  \end{align*}
  Informally, $\phi_{\graphE}(\tup{v},\tup{u})$ removes all edges $\tup{a}\tup{b}$
  of $G$, where $\tup{a}$ has more than one out-neighbour,
  and reverses the remaining edges.
  All that remains is to check whether there is a path from $\alpha(\tup{t})$
  to $\alpha(\tup{s})$ in the graph defined by $\phi_{\graphE}$.
  The node labelling formula $\phi_{\graphC}$ is chosen in such a way
  that the latter is true iff $(\alpha(\tup{t}),\ell)$,
  for an $\ell \leq \card{V}$,
  appears in the relation $X$ defined by $\phi$ in $(A,\alpha)$.
  If, for example, $G$ is the graph in Figure~\ref{fig:det-path},
  and if $\alpha(\tup{s}) = c$ and $\alpha(\tup{t}) = d$,
  then the labelled graph defined by $\phi_{\graphE}$ and $\phi_{\graphC}$
  is as shown in Figure~\ref{fig:det-path-transformed},
  and it is easy to see that $(d,4) \in X$,
  while, for example, $(e,\ell) \notin X$ for all $\ell > 0$.
  \varqed
\end{example}

As from now, we use
\begin{align}
  \label{eq:dtc-op}
  \dtcx{\tup{u}}{\tup{v}}{\psi}(\tup{s},\tup{t})
\end{align}
as an abbreviation for the $\LREC$-formula in \eqref{eq:dtc}.

\begin{remark}
  In the preceding two examples,
  the set $X$ turned out to possess a certain monotonicity property:
  If $(\tup{a},\ell) \in X$ for some $\ell$,
  then $(\tup{a},\ell') \in X$ for all $\ell' \geq \ell$.
  In general, however, the relation $X$ defined by an $\lrec$ operator
  does not possess this property.
  For example, consider the formula
  $
    \phi \isdef \lrecx{u}{v}{p}{E(u,v)}{\text{``$p = 0$''}}(u,p).
  $
  Now let $G$ be the graph consisting of a single edge $(a,b)$,
  and let $\alpha$ be the assignment mapping $u$ to $a$ and $p$ to 2.
  Then the relation $X$ defined by $\phi$ in $(G,\alpha)$
  contains $(a,1)$, but not $(a,2)$.
\end{remark}

The following theorem
shows that the data complexity of $\LREC$ is in $\LOGSPACE$.

\begin{theorem}
  \label{theo:lrec-complexity}
  For every vocabulary $\tau$, and every $\LREC[\tau]$-formula $\phi$
  there is a deterministic logspace Turing machine that,
  given a $\tau$-structure $A$ and an assignment $\alpha$ in $A$,
  decides whether $(A,\alpha) \models \phi$.
\end{theorem}

\begin{proof}
  We proceed by induction on the structure of $\phi$.
  The case where $\phi$ is \emph{not} of the form \eqref{eq:lrec} is easy.
  Let $\phi$ be of the form \eqref{eq:lrec}, i.e., let
  \[
    \phi \,=\,
    \lrecx{\tup{u}}{\tup{v}}{\tup{p}}{\phi_{\graphE}}{\phi_{\graphC}}(\tup{w},\tup{r}).
  \]
  Let $\graphG = (\graphV,{\graphE})$ be the graph with $\graphV = \Dom{A}{\tup{u}}$
  and $\graphE = \phi_{\graphE}[A,\alpha;\tup{u},\tup{v}]$,
  let
  $
    \graphC({\tup{a}}) \isdef
    \set{\num{\tup{n}} \mid
      \tup{n} \in \phi_{\graphC}[A,\alpha[\tup{u}/\tup{a}];\tup{p}]}
  $
  for all $\tup{a} \in \graphV$,
  and let $X \subseteq \graphV \times \nat$ be the relation defined by $\phi$
  in $(A,\alpha)$.
  We construct a deterministic logspace Turing machine
  that decides whether $(\alpha(\tup{w}),\num{\alpha(\tup{r})}) \in X$.

  The machine is constructed in two steps.
  The first step consists of
  constructing a deterministic logspace Turing machine $M_1$ that,
  given $A$ and $\alpha$ as input,
  computes a labelled directed tree $T$
  that is obtained basically from ``unravelling'' $\graphG$
  starting at $\alpha(\tup{w})$ with ``resource'' $\num{\alpha(\tup{r})}$.
  The second step is to devise a deterministic logspace Turing machine $M_2$
  that takes $T$ as input and decides
  whether its root, $(\alpha(\tup{w}),\num{\alpha(\tup{r})})$, belongs to $X$.
  The composition of $M_1$ and $M_2$ finally yields the desired machine.

  Let $k \isdef \len{\tup{r}}$.
  We define a labelled directed tree $T$
  whose set $W$ of vertices consists of all the sequences
  $
    ((\tup{a}_0,\ell_0),\dotsc,(\tup{a}_m,\ell_m))
  $
  of pairs from $\graphV \times \nat$ for some $m \in \nat$
  such that
  \begin{enumerate}[leftmargin=*]
  \item
    $(\tup{a}_0,\ell_0) = (\alpha(\tup{w}),\num{\alpha(\tup{r})})$,
  \item
    $\tup{a}_{i+1} \in \tup{a}_i \graphE$ for all $i < m$, and
  \item
    $
      \ell_{i+1} =
      \left\lfloor
        \frac{\ell_i-1}{\card{\graphE \tup{a}_{i+1}}}
      \right\rfloor
    $
    for all $i < m$.
  \end{enumerate}
  There is an edge from $((\tup{a}_0,\ell_0),\dotsc,(\tup{a}_m,\ell_m))$
  to $((\tup{a}'_0,\ell'_0),\dotsc,(\tup{a}'_{m'},\ell'_{m'}))$ in $T$
  if $m' = m+1$,
  and $(\tup{a}'_i,\ell'_i) = (\tup{a}_i,\ell_i)$ for all $i \leq m$.
  We label each vertex
  $v = ((\tup{a}_0,\ell_0),\dotsc,(\tup{a}_m,\ell_m)) \in W$
  with the set $\graphC(v) \isdef \graphC(\tup{a}_m)$,
  and with the number $\fail(v) \in \set{0,1}$
  such that $\fail(v) = 1$ iff $\ell_m = 0$.
  Note that $\fail(v) = 1$ only if $v$ is a leaf in $T$.
  Clearly, $T$ is a labelled directed tree
  rooted at $(\alpha(\tup{w}),\num{\alpha(\tup{r})})$.

  Define $Y \subseteq W$ such that
  \begin{align*}
    v \in Y
    \iff
    \card{\set{w \in Y \mid \text{$w$ is a child of $v$}}} \in \graphC(v)
    \ \ \text{and}\ \
    \fail(v) = 0
    \quad
    \text{(for every $v \in W$).}
  \end{align*}

  \begin{sclaim}
    \label{cl:lrec-complexity/X-prop}
    For every $v = ((\tup{a}_0,\ell_0),\dotsc,(\tup{a}_m,\ell_m)) \in W$
    we have $v \in Y$ if and only if $(\tup{a}_m,\ell_m) \in X$.
    In particular, $(\alpha(\tup{w}),\num{\alpha(\tup{r})}) \in X$
    if and only if $(\alpha(\tup{w}),\num{\alpha(\tup{r})}) \in Y$.
  \end{sclaim}

  \begin{proofofclaim}
    The proof is by induction on the \emph{rank} $r_v$ of $v$ in $T$:
    if $v$ is a leaf in $T$, then $r_v = 0$;
    and if $v$ is not a leaf in $T$, then $r_v$ is one more than the maximum
    of the ranks of $v$'s children.
    For every $v = ((\tup{a}_0,\ell_0),\dotsc,(\tup{a}_m,\ell_m)) \in W$,
    let $\lambda(v) \isdef (\tup{a}_m,\ell_m)$.

    Suppose that $r_v = 0$, that is, $v$ is a leaf in $T$.
    Consider $(\tup{a},\ell) = \lambda(v)$.
    Then $\tup{a} \graphE$ is the empty set or $\ell = 0$.
    First consider the case that $\ell = 0$.
    In this case, $(\tup{a},\ell) \notin X$ by the definition of $X$.
    But we also have $\fail(v) = 1$, which implies $v \notin Y$.
    Next consider the case that $\tup{a} \graphE$ is the empty set and $\ell > 0$.
    In this case,
    \[
      v \in Y
      \iff
      0 \in \graphC(v) = \graphC(\tup{a})
      \iff
      (\tup{a},\ell) \in X,
    \]
    as desired.

    Suppose now that $r_v = r+1$,
    and that the claim is true for vertices $w$ with $r_w \leq r$.
    In particular, since $v$ is not a leaf we must have $\fail(v) = 0$.
    This implies $\ell > 0$, and
    \begin{align}
      \notag
      v \in Y
      & \iff
        \card{\set{w \in Y \mid \text{$w$ is a child of $v$}}} \in \graphC(v) \\
      \label{eq:lrec-complexity/X-prop/1}
      & \iff
        \card{
          \set{\lambda(w) \in X \mid
            \text{$w$ is a child of $v$}}}
        \in \graphC(v)
        \quad
        \text{by the induction hypothesis.}
    \end{align}
    Let $W'$ be the set of all children $w$ of $v$
    such that $\lambda(w) \in X$,
    and let $f\colon W' \to \Dom{A}{\tup{u}}$ be such that for all $w \in W'$,
    $f(w)$ is the first component of $\lambda(w)$.
    Then $f$ is a bijection from $W'$
    to the set of all tuples $\tup{b} \in \tup{a} \graphE$ with
    \begin{align}
      \label{eq:lrec-complexity/X-prop/2}
      \left(
        \tup{b},
        \left\lfloor
          \frac{\ell-1}{\card{\graphE \tup{b}}}
        \right\rfloor
      \right)
      \in X.
    \end{align}
    As a consequence,
    the number of all tuples $\tup{b} \in \tup{a} \graphE$
    with \eqref{eq:lrec-complexity/X-prop/2}
    is precisely $\card{W'}$.
    Hence, by \eqref{eq:lrec-complexity/X-prop/1} and $\ell > 0$,
    \[
      v \in Y
      \iff
      \card{W'} \in \graphC(v) = \graphC(\tup{a})
      \iff
      \lambda(v) = (\tup{a},\ell) \in X.
      \varqedeq
    \]
  \end{proofofclaim}

  By Claim~\ref{cl:lrec-complexity/X-prop},
  it suffices to compute $T$,
  and use $T$ to decide whether its root,
  $(\alpha(\tup{w}),\num{\alpha(\tup{r})})$, belongs to $Y$.
  This is precisely what the two machines $M_1$ and $M_2$
  mentioned at the beginning of this proof do.
  We now prove the existence of such machines.

  \begin{sclaim}
    \label{cl:lrec-complexity/M1}
    There is a deterministic logspace Turing machine
    that takes $A$ and $\alpha$ as input and outputs $T$.
  \end{sclaim}

  \begin{proofofclaim}
    We first construct a deterministic logspace Turing machine $M$
    that takes $A$ and $\alpha$ as input and outputs the vertices of $T$
    (represented as sequences $((\tup{a}_0,\ell_0),\dotsc,(\tup{a}_m,\ell_m))$
     as above).
    This machine makes use of a deterministic logspace Turing machine $M_{\graphE}$
    that takes $A$, $\alpha$ and a pair $(\tup{a},\tup{b}) \in \graphV^2$ as input
    and decides whether $\tup{a}\tup{b} \in \graphE$.
    Such a machine exists by the induction hypothesis.
    Once $M$ is constructed,
    we can easily compute the edges and the labels of $T$,
    using a deterministic logspace Turing machine
    for computing the labels $\graphC(\tup{a})$ for each $\tup{a} \in \graphV$
    as guaranteed by the induction hypothesis.

    In what follows, we describe how $M$ computes the vertices of $T$
    from $A$ and $\alpha$.
    We basically do a depth-first search in $\graphG$
    starting in $\alpha(\tup{w})$ with ``resources'' $\num{\alpha(\tup{r})}$.
    In each step, we visit some vertex $\tup{a} \in \graphV$.
    We also maintain a number $\ell < \card{N(A)}^k$,
    the length $m$ of the path $P = (\tup{a}_0,\dotsc,\tup{a}_m)$
    on which $\tup{a}$ was reached from $\alpha(\tup{w})$,
    and for each $i \in [m]$ a number $e_i \in [0,\card{\graphE \tup{a}_i}-1]$
    with the following property.
    For each $\tup{b} \in \Dom{A}{\tup{u}}$
    let $\tup{b}_0,\dotsc,\tup{b}_p$ be the elements of $\graphE \tup{b}$
    ordered lexicographically according to their representation
    in the input string;
    let $\pre(\tup{b},i) \isdef \tup{b}_i$.
    Then the number $e_i$ will have the property that
    $\tup{a}_{i-1} = \pre(\tup{a}_i,e_i)$.
    When we move from $\tup{a}$ to some vertex $\tup{b} \in \tup{a} \graphE$
    we update $\ell$ to be
    \begin{align*}
      \decr(\ell,\tup{b}) \,\isdef\,
      \left\lfloor
        \frac{\ell-1}{\card{\graphE \tup{b}}}
      \right\rfloor.
    \end{align*}
    This ensures that the space needed to store the numbers $e_1,\dotsc,e_m$
    is logarithmic in $\card{W}$ (which we shall prove later).
    Finally, upon visiting $\tup{a}$ for the first time,
    we write the sequence $(\tup{a}_0,\ell_0),\dotsc,(\tup{a}_m,\ell_m)$
    to the output tape,
    where the $\ell_i$ are the values for $\ell$ maintained along the path $P$.

    More precisely, we proceed as follows.
    In the first step, we let $\tup{a} \isdef \alpha(\tup{w})$,
    $\ell \isdef \num{\alpha(\tup{r})}$ and $m \isdef 0$.
    Let $\tup{a} \in \graphV$, $\ell < \card{N(A)}^k$, $m \in \nat$
    and numbers $e_1,\dotsc,e_m$ be given.
    Furthermore,
    let $\tup{a}_0,\dotsc,\tup{a}_m$ be such that $\tup{a}_m = \tup{a}$,
    and for each $i \in [m]$, $\tup{a}_{i-1} = \pre(\tup{a}_i,e_i)$;
    and let $\ell_0,\dotsc,\ell_m$
    be such that $\ell_0 = \num{\alpha(\tup{r})}$ and for each $i \in [m]$,
    $\ell_i = \decr(\ell_{i-1},\tup{a}_i)$.
    Notice that each of the $\tup{a}_i$ and $\ell_i$
    can be computed in logarithmic space
    given $\tup{a}$, $m$, $e_1,\dotsc,e_m$ and $i$ as input.
    Let $\preceq$ be some fixed ordering on $\tup{a} \graphE$.
    There are now two possible cases:
    \begin{enumerate}[leftmargin=*]
    \item\label{case:lrec-complexity/M1/first-visit}
      \emph{$m$ was increased in the last move,
        or there was no last move.}
      This corresponds to a first visit of the vertex $\tup{a}$ with $\ell$
      on the current path.
      Therefore we write the sequence
      $(\tup{a}_0,\ell_0),\dotsc,(\tup{a}_m,\ell_m)$ to the output tape.
      We then let $j \isdef 0$ be the index of the child of $\tup{a}$
      to be visited next.
    \item\label{case:lrec-complexity/M1/return}
      \emph{$m$ was decreased in the last move.}
      This corresponds to a return from a child $\tup{b}$ of $\tup{a}$.
      Therefore, we do not write anything to the output tape.
      Let $\tup{b}$ be the vertex visited in the last step,
      let $j'$ be its rank in $\tup{a} \graphE$ with respect to $\preceq$
      (i.e., the number of elements in $\tup{a} \graphE$ that precede $\tup{b}$
       with respect to $\preceq$),
      and let $j \isdef j'+1$.
    \end{enumerate}
    If $\ell > 0$ and $j \leq \card{\tup{a} \graphE} - 1$,
    we update $\tup{a}$ to be the element of rank $j$
    in $\tup{a} \graphE$ with respect to $\preceq$;
    we also update $\ell$ to be $\decr(\ell,\tup{a})$,
    increase $m$ by one,
    and let $e_m$ be such that $\tup{a}_m = \pre(\tup{a},e_m)$.
    Otherwise, if $\ell = 0$ or $j = \card{\tup{a} \graphE}$,
    we do the following.
    If $m = 0$, we stop;
    and if $m > 0$ we update $\tup{a}$ to be $\tup{a}_{m-1}$,
    set $\ell$ to $\ell_{m-1}$,
    and decrease $m$ by one.
    It is not hard to see that this procedure outputs all the vertices of $T$.

    Maintaining the vertex $\tup{a} \in \graphV$
    and the vertex from the respective last step
    needs space $O(\log\, \card{V(A)})$.
    Notice that
    \begin{align*}
      \ell_0
      \,=\,
      \num{\alpha(\tup{r})}
      \,\leq\,
      (\card{V(A)}+1)^k-1.
    \end{align*}
    Since $\ell_i = \decr(\ell_{i-1},\tup{a}_i)$ for every $i \in [m]$,
    this implies
    \begin{align}
      \label{eq:lrec-complexity/M1/est}
      m \,\leq\, \ell_0 \,<\, (\card{V(A)}+1)^k
      \qquad\text{and}\qquad
      \prod_{i=1}^m\, \card{\graphE \tup{a}_i}
      \,<\, \frac{\ell_0}{\ell_m}
      \,<\, (\card{V(A)}+1)^k.
    \end{align}
    In particular, $m$ together with a bit indicating
    whether $m$ was increased or decreased in the last move
    can be maintained in space $O(\log\, \card{V(A)})$.
    Furthermore, each of the numbers $e_i$ needs space
    $\eta_i \isdef \lceil \log_2 \card{\graphE \tup{a}_i} \rceil$.
    Let $\Ic$ be the set of all $i \in [m]$ with $\card{\graphE \tup{a}_i} \geq 2$.
    By \eqref{eq:lrec-complexity/M1/est}
    we have $\card{\Ic} \leq \log_2 (\card{V(A)}+1)^k$.
    Hence,
    \begin{align*}
      \sum_{i=1}^m \eta_i
      \,=\,
      \sum_{i \in \Ic} \lceil \log_2 \card{\graphE \tup{a}_i} \rceil
      \,\leq\,
      \card{\Ic} + \log_2 \prod_{i \in \Ic} \card{\graphE \tup{a}_i}
      \,\stackrel{\eqref{eq:lrec-complexity/M1/est}}{\leq}\,
      2 \log_2 (\card{V(A)}+1)^k.
    \end{align*}
    In particular, we can store $e_1,\dotsc,e_m$ as a single number $e$
    with $\eta \isdef 2 \log_2 (\card{V(A)}+1)^k-1)$ bits,
    reserving $\eta_i$ bits in $e$ for the number $e_i$.
    To extract $e_i$ from $e$,
    we start by computing $\eta_m$ from $\tup{a} = \tup{a}_m$,
    let $e_m$ be the number represented by the last $\eta_m$ bits of $e$,
    and let $\tup{a}_{m-1} \isdef \pre(\tup{a}_m,e_m)$.
    We then compute $\eta_{m-1}$ from $\tup{a}_{m-1}$,
    let $e_{m-1}$ be the number corresponding to bit $\eta_{m-1}$
    to $\eta-\eta_m$ of $e$,
    and let $\tup{a}_{m-2} \isdef \pre(\tup{a}_{m-1},e_{m-1})$.
    We continue this way until $e_i$ is found.
    \varqed
  \end{proofofclaim}

  \begin{sclaim}
    \label{cl:lrec-complexity/M2}
    There is a deterministic logspace Turing machine
    that takes $T$ as input and decides
    whether the root $(\alpha(\tup{w}),\num{\alpha(\tup{r})})$ of $T$
    belongs to $Y$.
  \end{sclaim}

  \begin{proofofclaim}
    Let $v_0 \isdef (\alpha(\tup{w}),\num{\alpha(\tup{r})})$.
    On input $T$, a deterministic logspace Turing machine
    can decide whether $v_0 \in Y$ as follows.
    The idea is to visit the vertices in a depth-first fashion,
    starting in $v_0$,
    and count, for each node that is visited,
    the number of children that belong to $Y$.
    To implement this in logarithmic space,
    we proceed in steps as follows.

    In each step, we are in a vertex $v$ of $T$,
    which is $v_0$ in the first step.
    With each vertex $v_i$ on the path $v_0,v_1,\dotsc,v_m$ from $v_0$ to $v$
    we associate $2 \cdot \ell_v(i)$ bits of memory
    for counters $t(i),c(i)$ from $0$ to $2^{\ell_v(i)}-1$,
    where $\ell_v(i)$ will be specified below.
    The counter $t(i)$ simply counts the number of children of $v_i$
    that have already been processed
    (excluding the vertex in whose subtree we are currently in),
    while $c(i)$ counts the number of children of $v_i$
    that have already been processed and belong to $Y$.
    We guarantee that the sum of the numbers $2\cdot\ell_v(i)$
    over $i \in [0,m]$ is bounded by $6 \cdot \log_2 \card{W}$.
    Moreover, it will be easy to determine $\ell_v(i)$ from $v$ and $i$
    in logspace;
    so we can store the counters in a bit string of length
    at most $6 \cdot \log_2 \card{W}$,
    and identify the bits that belong to $t(i)$ and $c(i)$
    from that bit string in logspace, given $v$ and $i$.
    By visiting the children of each vertex
    in decreasing order of the number of vertices in the children's subtrees,
    we ensure that there is always enough space to keep the counters
    in memory until all children have been processed.

    We now give a more detailed description of a single step.
    In the initial step,
    we set $v \isdef v_0$ and $t(0) \isdef c(0) \isdef 0$.
    For the other steps, we need the following definitions:
    \begin{itemize}[leftmargin=*]
    \item
      The \emph{size $s(v)$} of a vertex $v \in W$ is the number of vertices
      in the subtree of $T$ rooted at $v$.
      It is easy to compute this number in logarithmic space:
      all we need to do is to initialise a counter,
      iterate over all vertices of $T$,
      and for each such vertex move upwards and increment the counter by 1
      if $v$ is reached.
    \item
      Let $v \in W$,
      and let $w_1,\dotsc,w_p$ be the children of $v$
      such that $s(w_1) \geq s(w_2) \geq \dotsb \geq s(w_p)$;
      children of the same size are ordered in lexicographic order
      based on their representation in the input string.
      For every $j \in [p]$, let $\child(v,j) \isdef w_j$.
      The vertex $\child(v,j)$ is easy to compute in logarithmic space,
      given $v$ and $j$.
    \item
      Let $v \in W$, let $v_0,v_1,\dotsc,v_m$ be the path from $v_0$ to $v$,
      and let $i \in [0,m]$.
      Then
      \[
        \ell_v(i) \,\isdef\,
        \begin{cases}
          \lceil \log_2 j\rceil,
          & \text{if $i < m$ and $\child(v_i,j) = v_{i+1}$}, \\
          \lceil \log_2 \card{W}\rceil,
          & \text{if $i = m$.}
        \end{cases}
      \]
      This number is easy to compute in logspace given $v$ and $i$ as input.
    \end{itemize}
    Suppose that $v$ is the current vertex,
    and that $v_0,v_1,\dotsc,v_m$ is the path from $v_0$ to $v$.
    If $t(m)$ is smaller than the number of children of $v$,
    then we set $v \isdef \child(v,t(m)+1)$
    and $t(m+1) \isdef c(m+1) \isdef 0$,
    and continue with the next step.
    Otherwise, we check whether $c(m) \in \graphC(v)$ and $\fail(v) = 0$.
    If this is the case, we say that $v$ \emph{succeeds}.
    In any case, whether $v$ succeeds or not, we do the following:
    \begin{enumerate}[leftmargin=*]
    \item
      If $m = 0$, then we accept $T$ iff $v$ succeeds.
    \item
      If $m > 0$, then we increase $t(m-1)$ by one,
      and if $v$ succeeds we also increase $c(m-1)$ by one.
      Afterwards, we let $v$ be the parent of $v$,
      and continue with the next step.
      Note that with the updated $v$,
      $2 \ell_v(m-1)$ bits suffice to store $t(m-1)$ and $c(m-1)$.
    \end{enumerate}
    It should be clear that this procedure correctly decides
    whether $v_0 \in Y$.

    Concerning the space for the counters,
    let $j_0,j_1,\dotsc,j_{m-1}$ be such that $\child(v_i,j_i) = v_{i+1}$
    for every $i < m$.
    Then
    \begin{align}
      \label{eq:lrec-complexity/sum}
      \sum_{i<m} \ell_v(i)
      & \,=\,    \sum_{\substack{i<m\\j_i \geq 2}} \lceil \log_2 j_i\rceil
        \,\leq\, \sum_{\substack{i<m\\j_i \geq 2}} \bigl(1 + \log_2 j_i\bigr)
        \,=\,    \card{\set{i<m \mid j_i \geq 2}}
                          + \log_2 \prod_{i<m} j_i.
    \end{align}
    Now observe that
    \begin{align}
      \label{eq:lrec-complexity/s}
      s(v_{i+1}) < \frac{s(v_i)}{j_i} \qquad \text{for every $i \in [0,m-1]$}.
    \end{align}
    To see this, consider $w_j \isdef \child(v_i,j)$ for every $j \leq j_i$.
    By the choice of $\child(\cdot,\cdot)$,
    we have $s(w_1) \geq \dotsb \geq s(w_{j_i})$.
    Hence, if $s(w_{j_i}) = s(v_{i+1}) \geq s(v_i)/j_i$,
    then $s(w_1) + \dotsb + s(w_{j_i}) \geq s(v_i)$, which is impossible.
    As a consequence of \eqref{eq:lrec-complexity/s}, we have
    \begin{align}
      \label{eq:lrec-complexity/est}
      \card{\set{i<m \mid j_i \geq 2}} \,<\, \log_2 \card{W}
      \qquad
      \text{and}
      \qquad
      \prod_{i<m} j_i
      \,\stackrel{\eqref{eq:lrec-complexity/s}}{<}\,
      \prod_{i<m} \frac{s(v_i)}{s(v_{i+1})}
      \,=\, \frac{s(v_0)}{s(v_m)}
      \,\leq\, \card{W}.
    \end{align}
    Altogether, this yields
    \begin{align*}
      \sum_{i \leq m} \ell_v(i)
      \,\stackrel{\eqref{eq:lrec-complexity/sum}}{\leq}\,
      \card{\set{i<m \mid j_i \geq 2}} + \log_2 \prod_{i<m} j_i
      + \log_2 \card{W} + 1
      \,\stackrel{\eqref{eq:lrec-complexity/est}}{<}\,
        3 \log_2 \card{W} + 1,
    \end{align*}
    which implies
    $\sum_{i \leq m} \ell_v(i) \leq 3 \log_2 \card{W}$,
    and therefore $\sum_{i \leq m} 2 \ell_v(i) \leq 6 \log_2 \card{W}$.
    \varqed
  \end{proofofclaim}

  Altogether, this concludes the proof of Theorem~\ref{theo:lrec-complexity}.
\end{proof}

\begin{remark}
  It follows from Example~\ref{ex:dtc} that $\DTCC \leq \LREC$.
  This containment is strict as directed tree isomorphism
  is definable in $\LREC$ (we will show this in the next section),
  but not in $\DTCC$.
  On the other hand,
  it is easy to see that
  the relation $X$ defined by an $\LREC$-formula
  of the form \eqref{eq:lrec} in an interpretation $(A,\alpha)$
  can be defined in fixed point logic with counting $\FPC$.
  Hence,
  $\LREC \leq \FPC$,
  and this containment is strict
  since we show in Section~\ref{sec:reach-nondef} that
  undirected graph reachability is not $\LREC$-definable.
\end{remark}

\section{Capturing Logspace on Directed Trees}
\label{sec:tree-canon}

In this section we show that $\LREC$ captures $\LOGSPACE$
on the class of all directed trees.
Our construction is based on Lindell's
\LOGSPACE\ tree canonisation algorithm \cite{Lindell:Tree-Canon}.
Note, however, that Lindell's algorithm makes essential use
of a linear order on the tree's vertices
that is given implicitly by the encoding of the tree.
Here we do not have such a linear order,
so we cannot directly translate Lindell's algorithm to an $\LREC$-formula.
We show that we can circumvent using the linear order
if we have a formula for directed tree isomorphism.
Hence, our first task is to construct such a formula.

\subsection{Directed Tree Isomorphism}
\label{sec:tree-iso}

Let $T$ be a directed tree.
For every $v \in V(T)$
let $T_v$ be the subtree of $T$ rooted at $v$,
let $\size(v) \isdef \card{V(T_v)}$ be the \emph{size} of $v$,
and let $\childcount_s(v)$ be the number of children of $v$ of size $s$.
We construct an $\LREC[\set{E}]$-formula $\fiso(x,y)$
that is true in a directed tree $T$
with interpretations $v,w \in V(T)$ for $x,y$
if and only if $T_v \isomorphic T_w$.
We assume that $\card{V(T)} \geq 4$,
but it is easy to adapt the construction to directed trees
with less than 4 vertices.

We implement the following recursive procedure
to check whether $T_v \isomorphic T_w$:
\begin{enumerate}[leftmargin=*]
\item
  If $\size(v) \neq \size(w)$ or if $\childcount_s(v) \neq \childcount_s(w)$
  for some $s \in [0,\card{V(T_v)}-1]$,
  then return ``$T_v \not\isomorphic T_w$''.
\item
  If for all children $\hat{v}$ of $v$
  there is a child $\hat{w}$ of $w$ and a number $k$ such that
  \begin{enumerate}[leftmargin=*]
  \item
    $T_{\hat{v}} \isomorphic T_{\hat{w}}$,
  \item
    there are exactly $k$ children $\ring{w}$ of $w$
    with $T_{\hat{v}} \isomorphic T_{\ring{w}}$,
    and
  \item
    there are exactly $k$ children $\ring{v}$ of $v$
    with $T_{\ring{v}} \isomorphic T_{\hat{w}}$,
  \end{enumerate}
  then return ``$T_v \isomorphic T_w$''.
\item
  Return ``$T_v \not\isomorphic T_w$''.
\end{enumerate}
Clearly, this procedure outputs ``$T_v \isomorphic T_w$'' if and only if
$T_v \isomorphic T_w$.

To simplify the presentation we fix a directed tree $T$
and an assignment $\alpha$ in $T$,
but the construction will be uniform in $T$ and $\alpha$.

We construct a directed graph $\graphG = (\graphV,\graphE)$
with labels $\graphC(v) \subseteq \nat$ for each $v \in \graphV$ as follows.
Let $\graphV \isdef N(T) \times V(T)^4 \times N(T)$.
The first component of each vertex is its \emph{type};
the meaning of the other components will become clear soon.
Although $\graphG$ will not be a tree,
it is helpful to think of it as a \emph{decision tree}
for deciding $T_v \isomorphic T_w$.
For each pair $(v,w) \in V(T)^2$,
we designate the vertex $\tup{a}_{v,w} = (0,v,w,v,w,0)$
to stand for ``$T_v \isomorphic T_w$''.
Let us call $(v,w)$ \emph{easy}
if $v,w$ satisfy the condition in line 1 of the procedure
(i.e., $\size(v) \neq \size(w)$, or $\childcount_s(v) \neq \childcount_s(w)$
 for some $s \in [0,\card{V(T_v)}-1]$).
Note that the set of all such easy pairs is $\LREC$-definable.%
\footnote{Using the $\dtc$-operator \eqref{eq:dtc-op} from Example~\ref{ex:dtc}
  we can construct an $\LREC[\set{E}]$-formula
  defining the descendant relation between vertices in a directed tree,
  and using this formula it is easy to determine the size
  and the number of children of size $s$ of a vertex.}
If $(v,w)$ is easy,
then $\tup{a}_{v,w}$ has no outgoing edges and $\graphC(\tup{a}_{v,w}) = \emptyset$.
On the other hand, if $(v,w)$ is not easy,
then $\graphG$ contains the following edges and labels
(see Figure~\ref{fig:iso} for an illustration):%
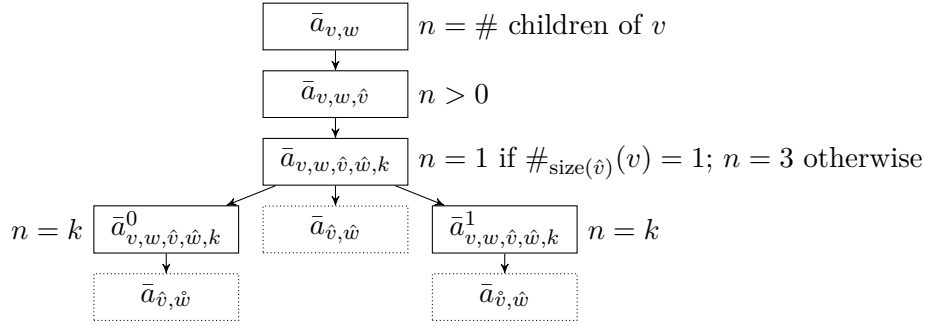
\begin{figure}
  \centering
  \begin{tikzpicture}[scale=0.9]
    \tikzstyle{vertex}=
      [draw,inner sep=2pt,minimum width=5em,minimum height=3.75ex]


    \node[vertex,label=right:{$n = \#$ children of $v$}] (type0)
      at (0,0) {$\tup{a}_{v,w}$};
    \node[vertex,label=right:{$n > 0$}] (type1)
      at (0,-1) {$\tup{a}_{v,w,\hat{v}}$};
    \node[vertex] (type2)
      at (0,-2) {$\tup{a}_{v,w,\hat{v},\hat{w},k}$};
    \node[right] at (type2.east)
      {$n = 1$ if $\childcount_{\size(\hat{v})}(v) = 1$;
       $n = 3$ otherwise};
    \node[vertex,label=left:{$n = k$}] (type3)
      at (-2.5,-3) {$\tup{a}_{v,w,\hat{v},\hat{w},k}^0$};
    \node[vertex,label=right:{$n = k$}] (type4)
      at (2.5,-3) {$\tup{a}_{v,w,\hat{v},\hat{w},k}^1$};

    \node[vertex,densely dotted] (type2-child)
      at (0,-3) {$\tup{a}_{\hat{v},\hat{w}}$};
    \node[vertex,densely dotted] (type3-child)
      at (-2.5,-4) {$\tup{a}_{\hat{v},\ring{w}}$};
    \node[vertex,densely dotted] (type4-child)
      at (2.5,-4) {$\tup{a}_{\ring{v},\hat{w}}$};


    \path[->] (type0) edge (type1);
    \path[->] (type1) edge (type2);
    \path[->] (type2) edge (type2-child) edge (type3) edge (type4);
    \path[->] (type3) edge (type3-child);
    \path[->] (type4) edge (type4-child);
  \end{tikzpicture}
  \caption{Sketch of ``decision tree'' for deciding $T_v \isomorphic T_w$.
    Here, $\hat{v},\ring{v}$ range over the children of $v$;
    $\hat{w},\ring{w}$ range over the children of $w$;
    and $k \in [\childcount_{\size(\hat{v})}(v)]$.
    Moreover, $\hat{v},\ring{v},\hat{w},\ring{w}$ all have the same size.
    Labels indicate which integers $n$ belong to the set $\graphC(\tup{a})$
    labelling each vertex $\tup{a}$.
    If $\hat{v}$ is the only child of $v$ of size $\size(\hat{v})$,
    then $\tup{a}_{\hat{v},\hat{w}}$ is the only child
    of $\tup{a}_{v,w,\hat{v},\hat{w},k}$.}
  \label{fig:iso}
\end{figure}
\begin{itemize}[leftmargin=*]
\item
  The vertex $\tup{a}_{v,w}$ has an outgoing edge
  to $\tup{a}_{v,w,\hat{v}} \isdef (1,v,w,\hat{v},w,0)$,
  for each child $\hat{v}$ of $v$.
  Furthermore, $\graphC(\tup{a}_{v,w}) = \set{\text{$\#$ of children of $v$}}$.
  This corresponds to ``for all children $\hat{v}$ of $v$\ldots''
  in the above procedure's step 2.
\item
  The vertex $\tup{a}_{v,w,\hat{v}}$ has an outgoing edge
  to $\tup{a}_{v,w,\hat{v},\hat{w},k} \isdef (2,v,w,\hat{v},\hat{w},k)$,
  for each child $\hat{w}$ of $w$ with $\size(\hat{w}) = \size(\hat{v})$
  and each $k \in [\childcount_{\size(\hat{v})}(v)]$.
  Furthermore, $\graphC(\tup{a}_{v,w,\hat{v}}) = N(T) \setminus \set{0}$.
  This branching corresponds to
  ``\ldots there is a child $\hat{w}$ of $w$ and a number $k$ such that\ldots''.
\item
  The vertex $\tup{a}_{v,w,\hat{v},\hat{w},k}$ has an outgoing edge
  to $\tup{a}_{\hat{v},\hat{w}}$.
  If $\hat{v}$ is the only child of $v$ of size $\size(\hat{v})$,
  then this is the only outgoing edge,
  and we let $\graphC(\tup{a}_{v,w,\hat{v},\hat{w},k}) = \set{1}$.
  Otherwise, there are additional outgoing edges to
  $\tup{a}_{v,w,\hat{v},\hat{w},k}^i = (3+i,v,w,\hat{v},\hat{w},k)$
  for $i \in \set{0,1}$,
  and we let $\graphC(\tup{a}_{v,w,\hat{v},\hat{w},k}) = \set{3}$.
  This corresponds to conditions 2a--2c.
\item
  The vertex $\tup{a}_{v,w,\hat{v},\hat{w},k}^0$
  has outgoing edges to $\tup{a}_{\hat{v},\ring{w}}$ for each child $\ring{w}$
  of $w$ of size $\size(\hat{v})$,
  and $\tup{a}_{v,w,\hat{v},\hat{w},k}^1$
  has outgoing edges to $\tup{a}_{\ring{v},\hat{w}}$ for each child $\ring{v}$
  of $v$ of size $\size(\hat{w}) = \size(\hat{v})$.
  Furthermore, $\graphC(\tup{a}_{v,w,\hat{v},\hat{w},k}^i) = \set{k}$.
  The vertex $\tup{a}_{v,w,\hat{v},\hat{w},k}^i$
  corresponds to condition 2b for $i = 0$,
  and to 2c for $i = 1$.
\end{itemize}
From the above description it should be easy
to construct $\LREC[\set{E}]$-formulae $\phi_{\graphE}(\tup{u},\tup{u}')$
and $\phi_{\graphC}(\tup{u},p)$,
where $\tup{u} = (q_t,x,y,\hat{x},\hat{y},q_k)$
and $\tup{u}' = (q_t',x',y',\hat{x}',\hat{y}',q_k')$,
such that $\phi_{\graphE}[T,\alpha;\tup{u},\tup{u}'] = \graphE$,
and $\set{\num{n} \mid n \in \phi_{\graphC}[T,\alpha[\tup{a}/\tup{u}];p]} = \graphC(\tup{a})$
for each $\tup{a} \in \graphV$.

Let
\[
  \fiso(x,y)\, \isdef\,
  \exists \tup{r}\,
  \lrecx{\tup{u}}{\tup{u}'}{p}{\phi_{\graphE}}{\phi_{\graphC}}((0,x,y,x,y,0),\tup{r}),
\]
where $\tup{r}$ is a 5-tuple of number variables.%
\footnote{We use 0 as a constant,
  but clearly we can modify $\fiso$ to a formula that does not use
  the constant 0.}
Let $X$ be the relation defined by $\fiso$ in $(T,\alpha)$.
Then:

\begin{lemma}
  \label{lem:tree-iso}
  Let $v,w \in V(T)$.
  \begin{enumerate}[label=\emph{(\arabic*)},leftmargin=*]
  \item
    \label{lem:tree-iso-correctness}
    If $(\tup{a}_{v,w},\ell) \in X$ for some $\ell \in \nat$,
    then $T_v \isomorphic T_w$.
  \item
    \label{lem:tree-iso-completeness}
    If $T_v \isomorphic T_w$,
    then for all $\ell \geq \size(v)^5$ we have $(\tup{a}_{v,w},\ell) \in X$.
  \end{enumerate}
\end{lemma}

\begin{proof}
  \ad{\ref{lem:tree-iso-correctness}}
  The proof is by induction on $\size(v)$.
  If $\size(v) = 1$ and $(\tup{a}_{v,w},\ell) \in X$,
  then $(v,w)$ is not easy,
  which implies $\size(w) = 1$ and hence $T_v \isomorphic T_w$.

  Now let $\size(v) = s+1$ for some $s \geq 1$.
  If $(\tup{a}_{v,w},\ell) \in X$,
  then $(v,w)$ is not easy, implying $\size(w) = s+1$
  and $\childcount_t(v) = \childcount_t(w)$ for all $t \in \nat$.
  It is then easy to see that for all children $\hat{v}$ of $v$ in $T$
  there is a child $\hat{w}$ of $w$ in $T$
  and a number $k \in [1,\#_{\size(\hat{v})} v]$ such that
  \begin{itemize}[leftmargin=*]
  \item
    $(\tup{a}_{\hat{v},\hat{w}},\ell') \in X$
    for some $\ell' \in \nat$,
  \item
    there are exactly $k$ children $\ring{w}$ of $w$
    such that $(\tup{a}_{\hat{v},\ring{w}},\ell') \in X$
    for some $\ell' \in \nat$,
    and
  \item
    there are exactly $k$ children $\ring{v}$ of $v$
    such that $(\tup{a}_{\ring{v},\hat{w}},\ell') \in X$
    for some $\ell' \in \nat$.
  \end{itemize}
  By the induction hypothesis,
  this corresponds to step~2 of the procedure
  given at the beginning of Section~\ref{sec:tree-iso},
  and therefore implies $T_v \isomorphic T_w$.

  \ad{\ref{lem:tree-iso-completeness}}
  The proof is by induction on $\size(v)$.
  Suppose that $\size(v) = 1$ and $T_v \isomorphic T_w$.
  Then $\size(w) = 1$ which implies that $(v,w)$ is not easy.
  Furthermore, as $v$ has no children in $T$,
  we know that $\tup{a}_{v,w}$ has no children in $\graphG$
  and $\graphC(\tup{a}_{v,w}) = \set{0}$.
  Hence, $(\tup{a}_{v,w},\ell) \in X$ for all $\ell \geq 1 = \size(v)^5$.

  Now suppose that $\size(v) = s+1$ for some $s \geq 1$,
  and $T_v \isomorphic T_w$.
  First note that $(v,w)$ is not easy.
  Let $\ell \geq (s+1)^5$.
  We show that $(\tup{a}_{v,w,\hat{v}},\ell-1) \in X$
  for all children $\hat{v}$ of $v$,
  which implies $(\tup{a}_{v,w},\ell) \in X$.
  Let $\hat{v}$ be a child of $v$ in $T$.
  Since $T_v \isomorphic T_w$,
  there is a child $\hat{w}$ of $w$ of size $s' \isdef \size(\hat{v})$
  and a number $k \in [\childcount_{s'}(v)]$ such that
  \begin{itemize}[leftmargin=*]
  \item
    $T_{\hat{v}} \isomorphic T_{\hat{w}}$,
  \item
    there are exactly $k$ children $\ring{w}$ of $w$ of size $s'$
    such that $T_{\hat{v}} \isomorphic T_{\ring{w}}$, and
  \item
    there are exactly $k$ children $\ring{v}$ of $v$ of size $s'$
    such that $T_{\ring{v}} \isomorphic T_{\hat{w}}$.
  \end{itemize}
  Pick such $\hat{w}$ and $k$.

  Let us deal with the case $\childcount_{s'}(v) = 1$ first.
  In this case,
  $\tup{a}_{\hat{v},\hat{w}}$ is the only child of
  $\tup{a}_{v,w,\hat{v},\hat{w},k}$;
  moreover, $\tup{a}_{v,w,\hat{v},\hat{w},k}$ and $\tup{a}_{\hat{v},\hat{w}}$
  have exactly one incoming edge each.
  Since $T_{\hat{v}} \isomorphic T_{\hat{w}}$ and $\ell-3 \geq (s')^5$,
  the induction hypothesis implies $(\tup{a}_{\hat{v},\hat{w}},\ell-3) \in X$.
  Consequently $(\tup{a}_{v,w,\hat{v}},\ell-1) \in X$.

  In the following we assume $\childcount_{s'}(v) \geq 2$.
  Let $d \isdef 3 \cdot \childcount_{s'}(v)^2$.
  Note that all vertices in Figure~\ref{fig:iso}
  except the type 0-vertices have exactly one incoming edge,
  and that the in-degree $d'$ of a type 0-vertex $\tup{a}_{v',w'}$,
  where $v',w'$ are children of $v$ and $w$, respectively, of size $s'$
  is at most $d$,
  because it has incoming edges from
  \begin{itemize}[leftmargin=*]
  \item
    vertices $\tup{a}_{v,w,v',w',k}$,
    where $v$ and $w$ are the (unique) parents of $v'$ and $w'$, respectively,
    and $k \in [\childcount_{s'}(v)]$;
  \item
    vertices $\tup{a}_{v,w,v',\hat{w},k}^0$,
    where $v,w,k$ are as above and $\hat{w}$ is a child of $w$ of size $s'$;
    and
  \item
    vertices $\tup{a}_{v,w,\hat{v},w',k}^1$,
    where $v,w,k$ are as above and $\hat{v}$ is a child of $v$ of size $s'$.
  \end{itemize}
  Let $\ell' \isdef \lfloor(\ell-4)/d\rfloor$.
  Then
  \[
    \ell'
    \,\geq\, \frac{\ell-d-3}{d}
    \,\geq\, \frac{s^5}{d} + \frac{s^4}{d} - 2
    \,\geq\, \frac{\childcount_{s'}(v)^5 \cdot (s')^5}
                  {3 \cdot \childcount_{s'}(v)^2}
             + \frac{\childcount_{s'}(v)^4}
                  {3 \cdot \childcount_{s'}(v)^2} - 2
    \,\geq\, 2 (s')^5 - 1
    \,\geq\, (s')^5,
  \]
  where for the second inequality we use $(s+1)^5 \geq s^5 + s^4$,
  for the third one we use
  $\childcount_{s'}(v) \cdot s' \leq s$,
  and for the fourth one we use $\childcount_{s'}(v) \geq 2$.
  Hence, by the induction hypothesis we have:
  \begin{itemize}[leftmargin=*]
  \item
    $(\tup{a}_{\hat{v},\hat{w}},\lfloor (\ell-3)/d' \rfloor) \in X$
    (note that $\lfloor (\ell-3)/d' \rfloor \geq \ell'$).
  \item
    There are exactly $k$ children $\ring{w}$ of $w$ of size $s'$
    with $(\tup{a}_{\hat{v},\ring{w}},\lfloor (\ell-4)/d' \rfloor) \in X$
    (note that $\lfloor (\ell-4)/d' \rfloor \geq \ell'$),
    which implies $(\tup{a}_{v,w,\hat{v},\hat{w},k}^0,\ell-3) \in X$.
  \item
    There are exactly $k$ children $\ring{v}$ of $v$ of size $s'$
    with $(\tup{a}_{\ring{v},\hat{w}},\lfloor (\ell-4)/d' \rfloor) \in X$,
    which implies that $(\tup{a}_{v,w,\hat{v},\hat{w},k}^1,\ell-3) \in X$.
  \end{itemize}
  It follows immediately that $(\tup{a}_{v,w,\hat{v},\hat{w},k},\ell-2) \in X$,
  and therefore $(\tup{a}_{v,w,\hat{v}},\ell-1) \in X$.
\end{proof}

\begin{corollary}
  Let $v,w \in V(T)^2$.
  Then, $T \models \fiso[v,w]$ if and only if $T_v \isomorphic T_w$.
\end{corollary}

\begin{proof}
  $T \models \fiso[v,w]$ holds precisely
  when $(\tup{a}_{v,w},\card{N(T)}^{\len{\tup{r}}}-1) \in X$.
  Furthermore,
  $\card{N(T)}^{\len{\tup{r}}}-1 \geq \card{V(T)}^5 \geq \size(v)^5$.
  Therefore, by the preceding lemma,
  $(\tup{a}_{v,w},\card{N(T)}^{\len{\tup{r}}}-1) \in X$
  is equivalent to $T_v \isomorphic T_w$,
  and the claim follows.
\end{proof}

\subsection{Defining an Order on Directed Trees}
\label{sec:tree-order}

Lindell's tree canonisation algorithm is based on a logspace-computable
linear order on isomorphism classes of directed trees.
We show that a slightly refined version of this order is $\LREC$-definable.

Let $T$ be a directed tree.
For each $v \in V(T)$ let
$
  \pi(v) \isdef
  \bigl(\size(v),\childcount_1(v),\dotsc,\childcount_{\size(v)-1}(v)\bigr)
$
be the \emph{profile} of $v$.%
\footnote{Lindell's order 
  can be obtained by replacing $\pi(v)$
  with $\pi'(v) \isdef \bigl(\size(v),\#\text{children of $v$}\bigr)$.}
Let $\preceq$ be the total preorder on $V(T)$,%
\footnote{That is, $\preceq$ is a preorder on $V(T)$
  such that for all $v,w \in V(T)$ we have $v \preceq w$ or $w \preceq v$.}
where $v \prec w$ whenever
\begin{enumerate}[leftmargin=*]
\item
  $\pi(v) < \pi(w)$ lexicographically, or
\item
  $\pi(v) = \pi(w)$ and the following is true:
  Let $v_1,\dotsc,v_k$ and $w_1,\dotsc,w_k$ be the children of $v$ and $w$,
  respectively,
  ordered such that $v_1 \preceq \dotsb \preceq v_k$
  and $w_1 \preceq \dotsb \preceq w_k$.
  Then there is an $i \in [k]$ with $v_i \prec w_i$,
  and for all $j < i$ we have $v_j \preceq w_j$ and $w_j \preceq v_j$.
\end{enumerate}
Note that $v \preceq w$ and $w \preceq v$ imply $T_v \isomorphic T_w$.
We show that $\preceq$ is $\LREC$-definable.

To simplify the presentation, we again fix a directed tree $T$
and an assignment $\alpha$,
and we assume that $\card{V(T)} \geq 4$.

We apply the $\lrec$-operator to the following graph $\graphG = (\graphV,\graphE)$
with labels $\graphC(v) \subseteq \nat$ for each $v \in \graphV$.
Let $\graphV \isdef N(T) \times V(T)^4 \times N(T)$.
For each $(v,w) \in V(T)^2$,
the vertex $\tup{a}_{v,w} = (0,v,w,v,w,0)$ represents ``$v \prec w$''.
If $\pi(v) < \pi(w)$,
then $\tup{a}_{v,w}$ has no outgoing edges and $\graphC(\tup{a}_{v,w}) = \set{0}$.
If $\pi(v) > \pi(w)$,
then $\tup{a}_{v,w}$ has no outgoing edges and $\graphC(\tup{a}_{v,w}) = \emptyset$.
Note that the relation ``$\pi(v) \leq \pi(w)$'' is $\LREC$-definable.

Suppose that $\pi(v) = \pi(w)$.
For all $t,u \in V(T)$
let $\theta_u(t)$ be the number of children $u'$ of $u$
with $T_{u'} \isomorphic T_t$.
Call a child $\hat{v}$ of $v$ \emph{good}
if $\theta_v(\hat{v}) > \theta_w(\hat{v})$
and for all children $v'$ of $v$ with $\size(v') < \size(\hat{v})$
we have $\theta_v(v') = \theta_w(v')$.
Then it is not hard to see that $v \prec w$
precisely if there is a good child $\hat{v}$ of $v$,
a child $\hat{w}$ of $w$ of size $s \isdef \size(\hat{v})$
and a $k \in [\childcount_s(v)]$ such that:
\begin{itemize}[leftmargin=*]
\item
  $\hat{v} \prec \hat{w}$;
\item
  there are exactly $k$ children $\ring{w}$ of $w$ of size $s$
  with $\ring{w} \prec \hat{v}$;
\item
  there are exactly $k$ children $\ring{v}$ of $v$ of size $s$
  with $\ring{v} \prec \hat{w}$ and $T_{\ring{v}} \not\isomorphic T_{\hat{v}}$;
\item
  and for all $k$ children $w'$ of $w$ of size $s$ with $w' \prec \hat{v}$
  we have $\theta_v(w') = \theta_w(w')$.
\end{itemize}
The ``decision tree'' in Figure~\ref{fig:ord} checks precisely these
conditions.%
\begin{figure}
  \centering
  \begin{tikzpicture}[xscale=0.8,yscale=0.9]
    \tikzstyle{vertex}=
      [draw,inner sep=2pt,minimum width=5em,minimum height=3.75ex]

    \node[vertex,label=right:{$n > 0$}] (root)
      at (0,0) {$\tup{a}_{v,w}$};

    \node[vertex] (find)
      at (0,-1) {$(1,v,w,\hat{v},\hat{w},k)$};
    \node[right] at (find.east)
      {$n = 1$ if $\childcount_{\size(\hat{v})}(v) = 1$;
       $n = 4$ otherwise};

    \node[vertex,densely dotted] (check1)
      at (-7,-2) {$\tup{a}_{\hat{v},\hat{w}}$};

    \node[vertex,label=right:{$n = k$}] (count_a)
      at (-2.5,-2) {$(2,v,w,\hat{v},\hat{w},k)$};
    \node[vertex,densely dotted] (check2)
      at (-2.5,-3) {$\tup{a}_{\ring{w},\hat{v}}$};

    \node[vertex,label=right:{$n = k$}] (count_b)
      at (2.5,-2) {$(3,v,w,\hat{v},\hat{w},k)$};
    \node[vertex,densely dotted] (check3)
      at (2.5,-3) {$\tup{a}_{\ring{v},\hat{w}}$};

    \node[vertex,label=right:{$n = k$}] (count_c)
      at (7.5,-2) {$(4,v,w,\hat{v},\hat{w},k)$};
    \node[vertex,densely dotted] (check4)
      at (7.5,-3) {$\tup{a}_{w',\hat{v}}$};

    \path[->] (root) edge (find);
    \path[->] (find) edge (check1) edge (count_a) edge (count_b) edge (count_c);
    \path[->] (count_a) edge (check2);
    \path[->] (count_b) edge (check3);
    \path[->] (count_c) edge (check4);
  \end{tikzpicture}
  \caption{Gadget for deciding $v \prec w$ when $\pi(v) = \pi(w)$.
    Here, $\hat{v}$ ranges over good children of $v$;
    $\ring{v}$ ranges over children of $v$ of size $s \isdef \size(v)$
    and $T_{\ring{v}} \not\isomorphic T_{\hat{v}}$;
    $\hat{w},\ring{w}$ range over children of $w$ of size $s$;
    $w'$ ranges over children of $w$ of size $s$
    with $\theta_v(w') = \theta_w(w')$;
    and $k \in [\childcount_s(v)]$.
    The edges from $(2,v,w,\hat{v},\hat{w},k)$ to $(t,\dotsc)$
    for $t \in \set{2,3,4}$
    exist only if $\childcount_s(v) > 1$.
    Labels indicate which integers $n$ belong to the set $\graphC(\tup{a})$
    labelling each vertex $\tup{a}$.}
  \label{fig:ord}
\end{figure}
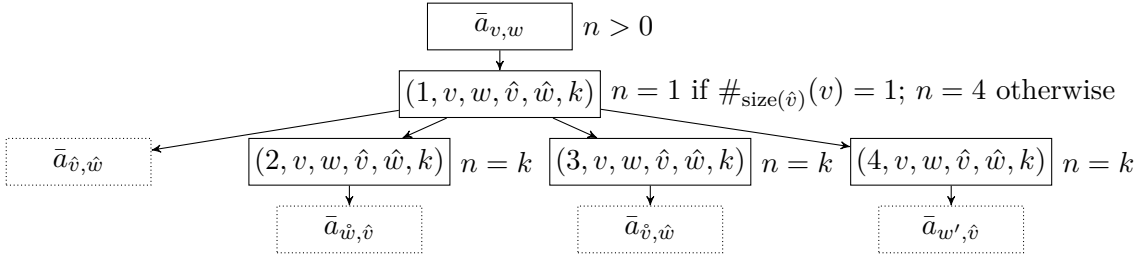

Using the formula $\fiso$ from the previous section
it is now straightforward to construct $\LREC[\set{E}]$-formulae
$\phi_{\graphE}(\tup{u},\tup{u}')$ and $\phi_{\graphC}(\tup{u},p)$
that define the edge relation $\graphE$ of $\graphG$
and the sets $\graphC(\tup{a})$ for each $\tup{a} \in \graphV$,
where $\tup{u}$ and $\tup{u}'$ are as in the definition of $\fiso$.
Let
\[
  \ford(x,y)
  \, \isdef\,
  \exists \tup{r}\,
  \lrecx{\tup{u}}{\tup{u}'}{p}{\phi_{\graphE}}{\phi_{\graphC}}((0,x,y,x,y,0),\tup{r}),
\]
where $\tup{r}$ is a 5-tuple of number variables.
Let $X$ be the relation defined by $\ford$ in $(T,\alpha)$.
We then have:

\begin{lemma}
  \label{lem:tree-ord}
  Let $v,w \in V(T)$.
  \begin{enumerate}[label=\emph{(\arabic*)},leftmargin=*]
  \item
    \label{lem:tree-ord-correctness}
    If $(\tup{a}_{v,w},\ell) \in X$ for some $\ell \in \nat$, then $v \prec w$.
  \item
    \label{lem:tree-ord-completeness}
    If $v \prec w$,
    then for all $\ell \geq \size(v)^5$ we have $(\tup{a}_{v,w},\ell) \in X$.
  \end{enumerate}
\end{lemma}

\begin{proof}
  The proof is similar to the proof of Lemma~\ref{lem:tree-iso}.

  \ad{\ref{lem:tree-ord-correctness}}
  The proof is by induction on $\size(v)$.
  Suppose $\size(v) = 1$.
  If $(\tup{a}_{v,w},\ell) \in X$, then $\pi(v) \leq \pi(w)$.
  We cannot have $\pi(v) = \pi(w)$,
  since otherwise $0 \notin \graphC(\tup{a}_{v,w})$ (see Figure~\ref{fig:ord}),
  so that $X$ would contain at least one tuple of the form
  $\bigl((1,v,w,\hat{v},\cdot,\cdot),\ell-1)$ with $\hat{v}$ a child of $v$.
  But such a tuple does not exist, since $v$ has no children.
  It follows that $\pi(v) < \pi(w)$ which implies $v \prec w$.

  Now let $\size(v) = s+1$ for some $s \geq 1$.
  If $(\tup{a}_{v,w},\ell) \in X$,
  then as above we have $\pi(v) \leq \pi(w)$.
  If $\pi(v) < \pi(w)$, we have $v \prec w$.
  So, suppose that $\pi(v) = \pi(w)$,
  that is, $\size(w) = s+1$
  and $\childcount_t(v) = \childcount_t(w)$ for all $t \in \nat$.
  It is then easy to see that there is a good child $\hat{v}$ of $v$,
  a child $\hat{w}$ of $w$ of size $s \isdef \size(\hat{v})$,
  and a $k \in [\childcount_s(v)]$ such that
  \begin{itemize}[leftmargin=*]
  \item
    $(\tup{a}_{\hat{v},\hat{w}},\ell') \in X$
    for some $\ell' \in \nat$,
  \item
    there are exactly $k$ children $\ring{w}$ of $w$ of size $s$
    such that $(\tup{a}_{\ring{w},\hat{v}},\ell') \in X$
    for some $\ell' \in \nat$,
  \item
    there are exactly $k$ children $\ring{v}$ of $v$ of size $s$
    with $T_{\ring{v}} \not\isomorphic T_{\hat{v}}$
    such that $(\tup{a}_{\ring{v},\hat{w}},\ell') \in X$
    for some $\ell' \in \nat$,
    and
  \item
    all $k$ children $w'$ of $w$ of size $s$
    with $(\tup{a}_{\ring{w},\hat{v}},\ell') \in X$
    for some $\ell' \in \nat$
    satisfy $\theta_v(w') = \theta_w(w')$.
  \end{itemize}
  By the induction hypothesis,
  this means that
  \begin{itemize}[leftmargin=*]
  \item
    $\hat{v} \prec \hat{w}$,
  \item
    there are exactly $k$ children $\ring{w}$ of $w$ of size $s$
    such that $\ring{w} \prec \hat{v}$,
  \item
    there are exactly $k$ children $\ring{v}$ of $v$ of size $s$
    with $T_{\ring{v}} \not\isomorphic T_{\hat{v}}$
    such that $\ring{v} \prec \hat{w}$,
    and
  \item
    all $k$ children $w'$ of $w$ of size $s$ with $w' \prec \hat{v}$
    satisfy $\theta_v(w') = \theta_w(w')$.
  \end{itemize}
  As pointed out in Section~\ref{sec:tree-order},
  this implies $v \prec w$.

  \ad{\ref{lem:tree-ord-completeness}}
  The proof is by induction on $\size(v)$.
  If $\size(v) = 1$ and $v \prec w$, then $\pi(v) < \pi(w)$.
  By the construction of $\graphG$ this immediately implies
  $(\tup{a}_{v,w},\ell) \in X$ for all $\ell \geq 1 = \size(v)^5$.

  Now suppose that $\size(v) = s+1$ for some $s \geq 1$,
  and $v \prec w$.
  First note that $\pi(v) \leq \pi(w)$.
  If $\pi(v) < \pi(w)$,
  then $(\tup{a}_{v,w},\ell) \in X$ for all $\ell \geq 1$,
  and in particular, for all $\ell \geq \size(v)^5$.
  So, assume that $\pi(v) = \pi(w)$.

  Since $v \prec w$,
  there is a good child $\hat{v}$ of $v$,
  a child $\hat{w}$ of $w$ of size $s' \isdef \size(\hat{v})$
  and a $k \in [\childcount_{s'}(v)]$ such that
  \begin{itemize}[leftmargin=*]
  \item
    $\hat{v} \prec \hat{w}$,
  \item
    there are exactly $k$ children $\ring{w}$ of $w$ of size $s'$
    with $\ring{w} \prec \hat{v}$,
  \item
    there are exactly $k$ children $\ring{v}$ of $v$ of size $s'$
    with $\ring{v} \prec \hat{w}$ and $T_{\ring{v}} \not\isomorphic T_{\hat{v}}$,
    and
  \item
    for all $k$ children $w'$ of $w$ of size $s'$ with $w' \prec \hat{v}$
    we have $\theta_v(w') = \theta_w(w')$.
  \end{itemize}
  Pick such $\hat{v}$, $\hat{w}$ and $k$.

  If $\childcount_{s'}(v) = 1$,
  then $\tup{a}_{\hat{v},\hat{w}}$ is the only child of
  $(1,v,w,\hat{v},\hat{w},k)$,
  and $(1,v,w,\hat{v},\hat{w},k)$ and $\tup{a}_{\hat{v},\hat{w}}$
  each have exactly one incoming edge.
  Since $\hat{v} \prec \hat{w}$ and $\ell-2 \geq (s')^5$,
  the induction hypothesis implies $(\tup{a}_{\hat{v},\hat{w}},\ell-2) \in X$,
  and consequently, $(\tup{a}_{v,w},\ell) \in X$.

  In the following we assume $\childcount_{s'}(v) \geq 2$.
  Let $d \isdef 4\childcount_{s'}(v)^2$.
  Note that all vertices in Figure~\ref{fig:ord}
  except the type 0-vertices have exactly one incoming edge.
  The type 0-vertices $\tup{a}_{v',w'}$,
  where $v',w'$ are children of $v$ and $w$, respectively, of size $s'$,
  have incoming edges from
  \begin{itemize}[leftmargin=*]
  \item
    vertices $(1,v,w,v',w',k)$, where $k \in [\childcount_{s'}(v)]$;
  \item
    vertices $(2,w,v,w',v'',k)$,
    where $k$ is as above and $v''$ is a child of $v$ of size $s'$;
  \item
    vertices $(3,v,w,v'',w',k)$,
    where $k$ and $v''$ is a good child of $v$; and
  \item
    vertices $(4,w,v,w',v'',k)$,
    where $k$ and $v''$ is a child of $v$ of size $s'$.
  \end{itemize}
  Hence, the in-degree of $\tup{a}_{v',w'}$ is at most $d$.
  For the type 0-vertices $\tup{a}_{w',v'}$,
  where $v',w'$ are children of $v$ and $w$, respectively, of size $s'$,
  this is symmetric.

  Let $\ell \geq (s+1)^5$ and $\ell' \isdef \lfloor(\ell-3)/d\rfloor$.
  Then
  \[
    \ell'
    \,\geq\, \frac{\ell-d-2}{d}
    \,\geq\, \frac{s^5}{d} + \frac{s^4}{d} - 2
    \,\geq\, \frac{\childcount_{s'}(v)^5 \cdot (s')^5}
                  {4 \cdot \childcount_{s'}(v)^2}
             + \frac{\childcount_{s'}(v)^4}
                  {4 \cdot \childcount_{s'}(v)^2} - 2
    \,\geq\, 2 (s')^5 - 1
    \,\geq\, (s')^5,
  \]
  where for the second inequality we use $(s+1)^5 \geq s^5 + s^4$,
  for the third one we use
  $\childcount_{s'}(v) \cdot s' \leq s$,
  and for the fourth one we use $\childcount_{s'}(v) \geq 2$
  Hence, by the induction hypothesis we have:
  \begin{itemize}[leftmargin=*]
  \item
    $(\tup{a}_{\hat{v},\hat{w}},\lfloor(\ell-2)/d_1\rfloor) \in X$,
    where $d_1 \leq d$ is the in-degree of $\tup{a}_{\hat{v},\hat{w}}$,
  \item
    there are exactly $k$ children $\ring{w}$ of $w$ of size $s'$
    with $(\tup{a}_{\ring{w},\hat{v}},\lfloor(\ell-3)/d_2\rfloor) \in X$,
    where $d_2 \leq d$ is the in-degree of the vertices
    $\tup{a}_{\ring{w},\hat{v}}$,
  \item
    there are exactly $k$ children $\ring{v}$ of $v$ of size $s'$
    with $T_{\ring{v}} \not\isomorphic T_{\hat{v}}$
    and $(\tup{a}_{\ring{v},\hat{w}},\lfloor(\ell-3)/d_3\rfloor) \in X$,
    where $d_3 \leq d$ is the in-degree of the vertices
    $\tup{a}_{\ring{v},\hat{w}}$,
    and
  \item
    for all $k$ children $w'$ of $w$ of size $s'$
    with $(\tup{a}_{w',\hat{v}},\lfloor(\ell-3)/d_2\rfloor) \in X$
    we have $\theta_v(w') = \theta_w(w')$.
  \end{itemize}
  It follows immediately that $(\tup{a}_{v,w},\ell) \in X$.
\end{proof}

\begin{corollary}
  Let $v,w \in V(T)$.
  Then, $T \models \ford[v,w]$ if and only if $v \prec w$.
\end{corollary}

\subsection{Canonising Directed Trees}

We now construct an $\LREC$-formula $\gamma(p,q)$
such that for every directed tree $T$
we have $T \isomorphic ([\card{V(T)}],\gamma[T;p,q])$.
Since $\DTC$ captures $\LOGSPACE$ on ordered structures \cite{imm87}
and a linear order is available on the number sort,
we immediately obtain:

\begin{theorem}\label{thm:LREC-and-dtrees}
  $\LREC$ captures $\LOGSPACE$ on the class of directed trees.
\end{theorem}

Since directed tree isomorphism is in $\LOGSPACE$
by Lindell's tree canonisation algorithm,
but not $\TCC$-definable \cite{eteimm00},
we obtain:

\begin{corollary}\label{cor:LREC-vs-TCC}
  $\LREC \not\leq \TCC$ on the class of all directed trees.
\end{corollary}

We use l-recursion to define a set $X \subseteq V(T) \times N(T)^2$
(for simplicity, we omit the ``resources'' in the description)
such that for every $v \in V(T)$
the set $X_v \isdef \set{(m,n) \in N(T)^2 \mid (v,m,n) \in X}$
is the edge relation of an isomorphic copy $([\card{V(T_v)}],X_v)$ of $T_v$.
Each vertex of $T$ is numbered by its position
in the preorder traversal sequence,
e.g., the root is numbered 1, its first child $v_1$ is numbered 2,
its second child $v_2$ is numbered $2+\size(v_1)$, and so on.

To apply the $\lrec$ operator,
we define a graph $\graphG = (\graphV,\graphE)$ with labels $\graphC(v) \subseteq \nat$
for each $v \in \graphV$ as follows.
Let $\graphV \isdef V(T) \times N(T)^2$,
where $(v,m,n) \in \graphV$ stands for ``$(m,n) \in X_v$?''.
If $v$ is a leaf, then $X_v$ should be empty,
so for all $m,n \in N(T)$ we let $(v,m,n)$ have no outgoing edges
and define $\graphC((v,m,n)) \isdef \emptyset$.
Suppose that $v$ is not a leaf and $w$ is a child of $v$.
Let $D_{w}$ be the set of all children $w'$ of $v$ with $w' \prec w$,
and let $e_w$ be the number of children $w'$ of $v$
with $T_w \isomorphic T_{w'}$.
For each $i \in [0,e_w-1]$, the set $X_v$ will contain an edge from $1$
to $p_{w,i} \isdef 2 + \sum_{w' \in D_{w}} \size(w') + i \cdot \size(w)$,
and the edges in $\{(p_{w,i}-1+m,p_{w,i}-1+n) \mid (m,n) \in X_w\}$.
Hence we let $(v,1,p_{w,i})$ have no outgoing edges and define
$\graphC((v,1,p_{w,i})) \isdef \set{0}$.
Furthermore, for all $m,n \in N(T)$ and all $i < e_w$,
we let $\tup{a} \isdef (v,p_{w,i}-1+m,p_{w,i}-1+n)$ have an edge to $(w,m,n)$
and define $\graphC(\tup{a}) \isdef \set{e_w}$.

It is now easy to construct $\LREC$-formulae
$\phi_{\graphE}(x_1,p_1,p'_1,x_2,p_2,p'_2)$ and $\phi_{\graphC}(x_1,p_1,p_1',q)$
that define the graph $\graphG$ and the labels $\graphC(\cdot)$.
Let
\[
  \gamma(p_1,p_2)\, \isdef\,
  \exists x \exists r \bigl(
    \text{``$x$ is the root''}
    \land
    \lrecx{(x_1,p_1,p'_1)}{(x_2,p_2,p'_2)}{q}
          {\phi_{\graphE}}{\phi_{\graphC}}((x,p_1,p_2),r)
  \bigr).
\]
Noting that the in-degree of each vertex $(v,m,n)$ is at most $e_v$,
it is straightforward to show that $\gamma$ defines an isomorphic copy
of a directed tree:

\begin{lemma}
  \label{lem:tree-canon}
  Let $X$ be the relation defined by $\gamma$ in $T$,
  let $v \in V(T)$
  and let
  $
    X_v \isdef
    \set{(m,n) \mid ((v,m,n),\ell) \in X\ \text{for some $\ell \geq \size(v)$}}.
  $
  Then $T_v \isomorphic ([\card{V(T_v)}],X_v)$.
\end{lemma}

\begin{proof}
  The proof is by induction on $\size(v)$.
  Clearly, the lemma is true if $\size(v) = 1$.
  Suppose that $\size(v) = s+1$.
  By the induction hypothesis,
  for each child $w$ of $v$ we have $T_w \isomorphic ([\card{V(T_w)}],X_w)$.

  Let $\ell \geq \size(v)$.
  Since for all children $w$ of $v$ and all $m,n \in N(T)$,
  the in-degree of $(w,m,n)$ in $\graphG$ is at most $e_w$
  and $e_w \cdot \size(w) < \size(v)$
  (which implies
   $\lfloor (\ell-1)/e_w \rfloor \geq
    \lfloor (\size(v)-1)/e_w \rfloor \geq \size(w)$),
  \[
    \set{(p_{w,i}-1+m,p_{w,i}-1+m) \mid (m,n) \in X_w} \subseteq X_v
    \quad
    \text{for each child $w$ of $v$ and $i < e_w$}.
  \]
  Furthermore, by construction,
  we have $(1,p_{w,i}) \in X_v$ for each child $w$ of $v$ and $i < e_w$,
  and there are no more edges.
  It is easy to see that $T_v \isomorphic ([\card{V(T_v)}],X_v)$.
\end{proof}

\begin{remark}
  The results of this section extend to \emph{coloured directed trees}
  with a linear order on the colours.
  To be precise, consider a directed tree $T$
  and a total preorder $\trianglelefteq$ on $V(T)$.
  Let $\preceq$ be as in Section~\ref{sec:tree-order}.
  We define a refinement $\preceq'$ of $\preceq$
  by letting $v \prec' w$ whenever $v \vartriangleleft w$,
  or: $v \trianglelefteq w$ and $w \trianglelefteq v$ and $v \prec w$.
  It should be obvious how to modify $\ford(x,y)$
  to an $\LREC[\set{E,\trianglelefteq}]$-formula $\ford'(x,y)$
  defining $\prec'$.
  Using this formula, we then obtain a formula $\gamma'(p,q)$
  such that
  $
    (V(T),E(T),\trianglelefteq) \isomorphic
    ([\card{V(T)}],\gamma'[T;p,q],\trianglelefteq'),
  $
  where $m \trianglelefteq' n$
  iff for the vertices $v,w$ that correspond to $m,n$
  we have $v \trianglelefteq w$.
\end{remark}

\section{Inexpressibility of Reachability in Undirected Graphs}
\label{sec:reach-nondef}

While $\LREC$ captures $\LOGSPACE$ on directed trees,
its expressive power still lacks the ability
to define certain important problems on undirected graphs
that can be defined easily in other logics such as $\STC$
with $\LOGSPACE$ data complexity.
As an example,
we show in this section that $\LREC$ cannot define reachability
in undirected graphs:

\begin{theorem}
  \label{thm:reach-nondef}
  There is no $\LREC[\set{E}]$-formula $\phi(x,y)$
  such that for all undirected graphs $G$ and all $v,w \in V(G)$,
  $G \models \phi[v,w]$ iff there is a path from $v$ to $w$ in $G$.
\end{theorem}

As an immediate corollary we obtain:

\begin{corollary}
  $\STC \not\leq \LREC$
\end{corollary}

To prove Theorem~\ref{thm:reach-nondef},
we show that reachability is not $\LREC$-definable
on a certain class of directed graphs.
This class, called $\Cc$ throughout this section,
is defined in terms of the following family of graphs $G_n$,
for $n \geq 1$.
Here, each graph $G_n$ consists of $2 \cdot n^2$ vertices,
which are partitioned into \emph{layers} $V_1^1,\dotsc,V_n^1,V_1^2,\dotsc,V_n^2$
with $\card{V_i^j} = n$.
Any two vertices in consecutive layers $V_i^j$ and $V_{i+1}^j$
are connected by an edge.
That is, the set $E(G_n)$ of edges of $G_n$ is
$\{(v,w) \in V_i^j \times V_{i+1}^j \mid i \in [n-1],\, j \in [2]\}$.
For example, the graph $G_3$ is shown in Figure~\ref{fig:G_n}.
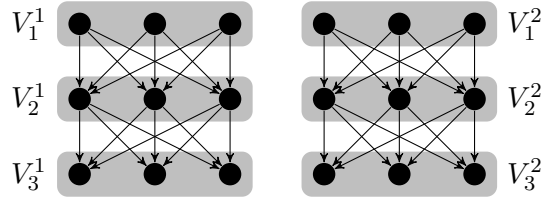
\begin{figure}
  \centering
  \begin{tikzpicture}
    \begin{scope}
      \foreach \i in {1,2,3} {
        \fill[rounded corners,lightgray] (0.7,-\i-.3) rectangle (3.3,-\i+0.3);
        \node[left] at (0.7,-\i) {$V_\i^1$};
      }
      \foreach \i in {1,2,3}
        \foreach \j in {1,2,3}
          \node[dotvertex] (\i\j) at (\i,-\j) {};
      \foreach \i in {1,2,3}
        \foreach \j in {1,2,3} {
          \path[->] (\i1) edge (\j2) (\i2) edge (\j3);
      }
    \end{scope}
    \begin{scope}[xshift=3.25cm]
      \foreach \i in {1,2,3} {
        \fill[rounded corners,lightgray] (0.7,-\i-.3) rectangle (3.3,-\i+0.3);
        \node[right] at (3.3,-\i) {$V_\i^2$};
      }
      \foreach \i in {1,2,3}
        \foreach \j in {1,2,3}
          \node[dotvertex] (\i\j) at (\i,-\j) {};
      \foreach \i in {1,2,3}
        \foreach \j in {1,2,3} {
          \path[->] (\i1) edge (\j2) (\i2) edge (\j3);
      }
    \end{scope}
  \end{tikzpicture}
  \caption{The graph $G_3$.
    The gray areas highlight the different layers of $G_3$.}
  \label{fig:G_n}
\end{figure}
Now, the class $\Cc$ is defined as:
\[
  \Cc :=
  \set{G \mid \text{$G$ is a graph such that $G \isomorphic G_n$
      for some $n \geq 1$}}.
\]

The key property of the graphs in $\Cc$
that enables us to show that reachability on $\Cc$ is not $\LREC$-definable
is that they are rich in a certain kind of automorphisms.
Indeed, let $v$ and $w$ be nodes occurring in the same layer of $G_n$.
Then there is an automorphism of $G_n$ swapping $v$ and $w$,
and fixing the remaining vertices point-wise.
To see why this could be useful at all,
consider an $\LREC$-formula $\phi$ of the form
$
  \lrecx{\tup{u}_1}{\tup{u}_2}{\tup{p}}{\phi_{\graphE}}{\phi_{\graphC}}(\tup{w},\tup{r}),
$
and suppose we want to decide membership of a tuple $(\tup{a}_0,\ell_0)$
in the relation $X$ defined by $\phi$ in $(G_n,\alpha)$,
for an assignment $\alpha$.
First, we would compute the graph $\graphG$ with vertex set $\Dom{G_n}{\tup{u}_1}$
and edge set $\graphE$ defined by $\phi_{\graphE}$,
and then we would recurse to decide which of the tuples $(\tup{a}_1,\ell_1)$,
for successor nodes $\tup{a}_1$ of $\tup{a}_0$ in $\graphG$
and $\ell_1 = \lfloor (\ell_0-1)/\card{\graphE\tup{a}_1} \rfloor$,
belong to $X$.
To decide membership of each of the tuples $(\tup{a}_1,\ell_1)$ in $X$,
we again have to recurse to decide which of the tuples $(\tup{a}_2,\ell_2)$,
for successor nodes $\tup{a}_2$ of $\tup{a}_1$ in $\graphG$
and $\ell_2 = \lfloor (\ell_1-1)/\card{\graphE\tup{a}_2} \rfloor$,
belong to $X$, and so on.
Exploiting the above-mentioned automorphisms
enables us to show that along each branch
$(\tup{a}_0,\ell_0),(\tup{a}_1,\ell_1),(\tup{a}_2,\ell_2),\dotsc$
of the ``recursion tree'',
we see only a constant number of tuples $(\tup{a}_{i+1},\ell_{i+1})$,
where $\tup{a}_{i+1}$ does not contain all the vertices of $G_n$
that occur in $\tup{a}_i$, or vice versa.
Thus, we are left with finitely many sub-branches ``in between'' those tuples
that contain the same vertices of $G_n$.
If all those sub-branches had constant length,
then the whole ``recursion tree'' would have constant depth,
so that we could easily find an $\FOC$-formula
that is equivalent to $\phi$ on $\Cc$
(provided $\phi_{\graphE}$ and $\phi_{\graphC}$ are equivalent to $\FOC$-formulae).
Since reachability is not $\FOC$ definable on $\Cc$,
this would immediately imply Theorem~\ref{thm:reach-nondef}.
In general, the sub-branches do not have constant length
(due to number variables that may occur in $\tup{u}_1$ and $\tup{u}_2$),
so that we move to a logic that is more expressive than $\FOC$,
but still lacks the ability to define reachability on $\Cc$.

More precisely, we show that on $\Cc$,
every $\LREC[\set{E}]$-formula is equivalent to a formula
in the \emph{infinitary counting logic} $\CInfB$,
introduced in \cite{Libkin:TOCL00} (see also \cite[Section~8.2]{lib04}).
The fact that $\CInfB$-formulae without free number variables
are Gaifman-local \cite{Libkin:TOCL00}
then yields that reachability is not $\CInfB$-definable,
and hence not $\LREC$-definable, on $\Cc$.

\subsection{The Logic \texorpdfstring{$\CInfB$}{L*\_inf(C)}}

Before delving into the details of translating $\LREC$-formulae
into $\CInfB$-formulae,
we give here a brief review of the logic $\CInfB$.
For a detailed account,
we refer the reader to \cite{Libkin:TOCL00}, or \cite[Section~8.2]{lib04}.

$\CInfB$ on the one hand extends $\FOC$
by allowing for infinite disjunctions and conjunctions,
and on the other hand imposes restrictions so as to make the resulting logic
not too powerful.
While in the context of $\FOC$,
we equipped structures $A$ with a counting sort $N(A) = [0,\card{V(A)}]$,
in the context of $\CInfB$ we extend this counting sort
to the set of all natural numbers.
Furthermore, $\CInfB$-formulae may use any natural number $n \in \nat$
as a constant, which is always interpreted as $n$.

$\CInfB$ is a restriction of the extremely powerful logic $\CInf$,
which is defined as follows.
A \emph{term} $t$ is a structure variable, a number variable,
or a non-negative integer;
if $t$ is a structure variable, we call $t$ \emph{structure term},
and otherwise \emph{number term}.
The atomic formulae of $\CInf[\tau]$ have the form $R(x_1,\dotsc,x_r)$,
where $R \in \tau$, $r$ is the arity of $R$,
and $x_1,\dotsc,x_r$ are structure variables;
or $t = u$,
where $t$ and $u$ are either structure terms or number terms;
or $t \leq u$,
where $t$ and $u$ are number terms.
The set of all $\CInf[\tau]$-formulae is the smallest set
that contains all atomic formulae,
and is closed under the following formula formation rules:
\begin{enumerate}[leftmargin=*]
\item\label{rule:CInf-neg}
  If $\phi \in \CInf[\tau]$, then $\lnot \phi \in \CInf[\tau]$.
\item\label{rule:CInf-inf}
  If $\Phi \subseteq \CInf[\tau]$,
  then $\biglor \Phi$ and $\bigland \Phi$ belong to $\CInf[\tau]$.
\item\label{rule:CInf-sq}
  If $\phi \in \CInf[\tau]$ and $x$ is a variable,
  then $\exists x \phi$ and $\forall x \phi$ belong to $\CInf[\tau]$.
\item\label{rule:CInf-cq}
  If $\phi \in \CInf[\tau]$, $x$ is a structure variable, and $n \in \nat$,
  then $\exists^{\geq n} x \phi\in \CInf[\tau]$.
\item\label{rule:CInf-cf}
  If $\phi \in \CInf[\tau]$, $\tup{x}$ is a tuple of structure variables,
  and $\tup{p}$ is a tuple of number terms,
  then $\#\tup{x}\,\phi = \tup{p}$ belongs to $\CInf[\tau]$.
\end{enumerate}
Note that, in contrast to $\FOC$,
$\CInf$ restricts us to tuples of structure variables in counting formulae
$\#\tup{x}\,\phi = \tup{p}$.
The semantics of $\CInf[\tau]$-formulae constructed as in \ref{rule:CInf-neg},
\ref{rule:CInf-sq}, and \ref{rule:CInf-cf} is as usual.
The semantics of formulae of the form $\biglor \Phi$ or $\bigland \Phi$
is ``at least one $\phi \in \Phi$ is satisfied'' and
``all $\phi \in \Phi$ are satisfied'', respectively.
Formulae of the form $\exists^{\geq n} x \phi$ have the meaning
``there are at least $n$ assignments to $x$ for which $\phi$ is satisfied''.

$\CInfB[\tau]$-formulae are those $\CInf[\tau]$-formulae
whose \emph{rank} is bounded.
Here, the \emph{rank} $\rank(\phi)$ of a $\CInf[\tau]$-formula $\phi$
is defined as follows.
For atomic formulae $\phi$ we have $\rank(\phi) = 0$.
Furthermore,
$
  \rank(\lnot \phi) = \rank(\phi),
$
$
  \rank(\biglor \Phi) = \rank(\bigland \Phi) =
  \sup_{\phi \in \Phi} \rank(\phi),
$
$
  \rank(\exists x \phi) = \rank(\forall x \phi) =
  \rank(\exists^{\geq n} x \phi) = 1 + \rank(\phi)
$
if $x$ is a structure variable,
$
  \rank(\exists x \phi) = \rank(\forall x \phi) = \rank(\phi)
$
if $x$ is a number variable,
and
$
  \rank(\#\tup{x}\,\phi = \tup{p}) = \len{\tup{x}} + \rank(\phi).
$
Now, a $\CInf[\tau]$-formula $\phi$ belongs to $\CInfB[\tau]$
if there is a number $n \in \nat$ with $\rank(\phi) \leq n$.

As shown in \cite{Libkin:TOCL00}, every $\CInfB$ formula without free
number variables is Gaifman local.  To make this precise, we need some
more notation.  Given a graph $G$ and vertices $v,w \in V(G)$, let
$\dist^G(v,w)$ denote the length of a shortest path from $v$ to $w$ in
the undirected graph obtained from $G$ by adding edges $(w',v')$ for
every edge $(v',w') \in E(G)$, or $\infty$ if there is no such path.
For all $k \geq 1$, all tuples $\tup{v} = (v_1,\dotsc,v_k) \in V(G)^k$
and all $r \in \nat$, let $ B_r^G(\tup{v}) \isdef \set{w \in V(G) \mid
  \exists i \in [k]\colon \dist^G(v_i,w) \leq r}, $ and define
$N_r^G(\tup{v})$ to be the subgraph of $G$ induced by
$B_r^G(\tup{v})$.  The following theorem is stated in
\cite{Libkin:TOCL00} for arbitrary vocabularies:

\begin{theorem}[Restricted form of a theorem in \cite{Libkin:TOCL00}]
  \label{thm:locality}
  For every $\CInfB[\set{E}]$-formula $\phi(\tup{x})$
  without free number variables,
  there is an $r \in \nat$ such that for all graphs $G$
  and all $\tup{a},\tup{b} \in V(G)^{\len{\tup{x}}}$
  with $(N_r^G(\tup{a}),\tup{a}) \isomorphic (N_r^G(\tup{b}),\tup{b})$
  we have: $G \models \phi[\tup{a}] \iff G \models \phi[\tup{b}]$.
\end{theorem}

Using Theorem~\ref{thm:locality}, it is straightforward to show that:

\begin{corollary}
  \label{cor:locality}
  There is no $\CInfB[\set{E}]$-formula $\phi(x,y)$
  such that for all $G \in \Cc$ and all $v,w \in V(G)$
  we have $G \models \phi[v,w]$ iff there is a path from $v$ to $w$ in $G$.
\end{corollary}

\begin{proof}
  For a contradiction, suppose that $\phi(x,y)$ is an $\CInfB[\set{E}]$-formula
  such that for all $G \in \Cc$ and all $v,w \in V(G)$
  we have $G \models \phi[v,w]$ iff there is a path from $v$ to $w$ in $G$.
  Let $r \in \nat$ be as guaranteed by Theorem~\ref{thm:locality}.
  We can now pick vertices $v,w_1,w_2 \in G_{r+2}$
  with $N_r^{G_{r+2}}(v,w_1) \isomorphic N_r^{G_{r+2}}(v,w_2)$
  such that $w_1$ is reachable from $v$,
  but $w_2$ is not reachable from $v$.
  Since $G_{r+2} \models \phi[v,w_1]$,
  we then have $G_{r+2} \models \phi[v,w_2]$,
  a contradiction.
\end{proof}

\subsection{Translation of \texorpdfstring{$\LREC$}{LREC}-Formulae Into \texorpdfstring{$\CInfB$}{L*\_inf(C)}-Formulae}

We now describe the translation of an $\LREC$-formula $\phi$
into an $\CInfB$-formula $\tilde{\phi}$ that is equivalent to $\phi$ on $\Cc$.
The translation proceeds by induction on the structure of $\phi$,
where the only interesting case is that of $\LREC$-formulae $\phi$ of the form
\[
  \lrecx{\tup{u}_1}{\tup{u}_2}{\tup{p}}{\phi_{\graphE}}{\phi_{\graphC}}(\tup{w},\tup{r}).
\]
To decide whether $\phi$ holds in a given graph $G_n$
under an assignment $\alpha$,
$\tilde{\phi}$ needs to check whether the tuple $(\tup{a}_0,\ell_0)$,
for $\tup{a}_0 := \alpha(\tup{w})$ and $\ell_0 := \num{\alpha(\tup{r})}$,
belongs to the relation $X$ defined by $\phi$ in $(G_n,\alpha)$.
To this end,
it looks at the graph $\graphG$ with vertex set $\Dom{G_n}{\tup{u}_1}$
and edge set $\phi_{\graphE}[G_n,\alpha;\tup{u}_1,\tup{u}_2]$,
or rather at its \emph{$\ell_0$-unravelling} $\graphG^{(\tup{a}_0,\ell_0)}$
at $\tup{a}_0$:

\begin{definition}
  The \emph{$\ell$-unravelling} of a graph $\graphG = (\graphV,\graphE)$ at a vertex $v \in \graphV$
  is the tree $\graphG^{(v,\ell)}$ defined as follows:
  \begin{enumerate}[leftmargin=*]
  \item
    The nodes of $\graphG^{(v,\ell)}$ are all finite sequences
    $((v_0,\ell_0),\dotsc,(v_n,\ell_n))$,
    where $(v_0,\ell_0) = (v,\ell)$, $(v_0,\dotsc,v_n)$ is a path in $\graphG$,
    and $\ell_i = \lfloor (\ell_{i-1}-1)/\card{\graphE v_i} \rfloor$
    for every $i \in [n]$.
  \item
    There is an edge from a node $((v_0,\ell_0),\dotsc,(v_m,\ell_m))$
    to a node $((v'_0,\ell'_0),\dotsc,(v'_n,\ell'_n))$
    whenever $n = m+1$, and $(v'_i,\ell'_i) = (v_i,\ell_i)$
    for every $i \leq m$.
  \item
    Each node $((v_0,\ell_0),\dotsc,(v_m,\ell_m))$ is labelled
    with $(v_m,\ell_m)$.
  \end{enumerate}
\end{definition}

\noindent
For each node of $\graphG^{(\tup{a}_0,\ell_0)}$,
$\tilde{\phi}$ checks whether its label belongs to $X$.
Clearly, this suffices to decide whether $(\tup{a}_0,\ell_0) \in X$.

Our construction is based on the following property of $\graphG^{(\tup{a}_0,\ell_0)}$:

\begin{lemma}
  \label{lem:elem-depth}
  Let $\phi_{\graphE}(\tup{x},\tup{y},\tup{z})$ be a formula,
  where $\tup{x},\tup{y}$ are compatible,
  let $n > \len{\tup{x}} + \len{\tup{z}} + 2$,
  let $\alpha$ be an assignment for $\phi_{\graphE}$ in $G_n$,
  and let $\graphG = (\graphV,\graphE)$ be the graph with
  $\graphV := G_n^{\tup{x}}$ and $\graphE := \phi_{\graphE}[G_n,\alpha;\tup{x},\tup{y}]$.
  Consider a node $((\tup{a}_0,\ell_0),\dotsc,(\tup{a}_m,\ell_m))$
  in $\graphG^{(\tup{a},\ell)}$,
  where $\ell \leq \card{N(G_n)}^r-1$.
  Then, the size of
  \[
    \Ic := \set{i \in [m] \mid
      (\tilde{a}_{i-1} \union \alpha(\tilde{z})) \intersect V(G_n)
      \neq
      (\tilde{a}_i \union \alpha(\tilde{z})) \intersect V(G_n)
    }
  \]
  is bounded by a constant that depends only on $\phi_{\graphE}$ and $r$.
\end{lemma}

\begin{proof}
  We first show that the size of
  \[
    \Kc \isdef
    \set{i \in \Ic \mid
      \tilde{a}_{i-1} \intersect V(G_n)
      \nsubseteq
      (\tilde{a}_i \union \alpha(\tilde{z})) \intersect V(G_n)}
  \]
  is bounded by a constant that only depends on $\phi_{\graphE}$ and $r$.
  To this end,
  consider an $i \in \Kc$ and a $b \in \tilde{a}_{i-1} \intersect V(G_n)$
  such that $b \notin \tilde{a}_i \union \alpha(\tilde{z})$.
  Let us call an element $b' \in V(G_n)$ a \emph{sibling} of $b$
  if $b$ and $b'$ belong to the same layer in $G_n$.
  There are at least
  \[
    n - \card{\tilde{a}_i \union \alpha(\tilde{z})} - 1
    \geq
    n - (\len{\tup{x}} + \len{\tup{z}} + 1)
  \]
  siblings of $b$ in $G_n$
  that do not occur in $\tilde{a}_i \union \alpha(\tilde{z}) \union \set{b}$.
  Each such sibling $b'$ gives rise to an \emph{automorphism}
  $f_{b'}\colon V(G_n) \to V(G_n)$ of $G_n$
  that fixes all the vertices in $V(G_n) \setminus \set{b,b'}$ point-wise,
  maps $b$ to $b'$, and maps $b'$ to $b$.
  As a consequence, for each such sibling $b'$
  we have $f_{b'}(\tup{a}_{i-1}) \tup{a}_i \in \graphE$,
  where $f_{b'}(\tup{a}_{i-1})$ is the tuple obtained from $\tup{a}_{i-1}$
  by replacing each element $b''$ in $\tup{a}_{i-1}$ that belongs to $V(G_n)$
  with $f_{b'}(b'')$.
  This implies
  \[
    \card{\graphE \tup{a}_i} \geq n - d_1,
  \]
  where $d_1 \isdef \len{\tup{x}} + \len{\tup{z}} + 1$ depends only on $\phi_{\graphE}$.

  Observe that, by the definition of $\graphG^{(\tup{a},\ell)}$,
  we have $\ell_0 = \ell \leq \card{N(G_n)}^r-1 \leq (2n)^{2r}$
  and $\ell_0 \geq \prod_{i=1}^m\, \card{\graphE \tup{a}_i}$.
  Hence,
  \[
    (2n)^{2r}
    \ \geq\
    \prod_{i=1}^m\, \card{\graphE \tup{a}_i}
    \ \geq\
    \prod_{i \in \Kc}\, \card{\graphE \tup{a}_i}
    \ \geq\
    \prod_{i \in \Kc}\, (n - d_1)
    \ =\
    (n - d_1)^{\card{\Kc}}.
  \]
  For $n > d_1 + 1$ this implies
  $
    \card{\Kc}
    \leq \log_{n-d_1} (2n)^{2r}
    \leq 2 r (1 + \log_{n-d_1} n),
  $
  which is bounded by a constant $d_2$
  that only depends on $\phi_{\graphE}$ and $r$.

  To conclude the proof,
  consider a maximal set $\Ic' \subseteq \Ic$
  such that there are no $i,i' \in \Ic'$ and $k \in \Kc$
  with $i \leq k \leq i'$.
  We show that $\card{\Ic'}$ is bounded by a constant $d_3$
  that depends only on $\phi_{\graphE}$.
  This then implies the lemma as
  \[
    \card{\Ic} \leq (\card{\Kc}+1) \cdot (d_3+1) \leq (d_2+1) \cdot (d_3+1).
  \]
  Let $i_{\min} \isdef \min \Ic'$ and $i_{\max} \isdef \max \Ic'$,
  and notice that
  \[
    \bigl(\tilde{a}_{i_{\min}-1} \union \alpha(\tilde{z})\bigr) \intersect V(G_n)
    \ \subseteq\
    \bigl(\tilde{a}_{i_{\min}} \union \alpha(\tilde{z})\bigr) \intersect V(G_n)
    \ \subseteq\
    \dotsb
    \ \subseteq\
    \bigl(\tilde{a}_{i_{\max}} \union \alpha(\tilde{z})\bigr) \intersect V(G_n).
  \]
  Since $(\tilde{a}_{i_{\max}} \union \alpha(\tilde{z})) \intersect V(G_n)$
  contains at most $d_3 \isdef \len{\tup{x}}$ elements that do not belong to
  $(\tilde{a}_{i_{\min}-1} \union \alpha(\tilde{z})) \intersect V(G_n)$,
  there are at most $d_3$ indices $i \in [i_{\min},i_{\max}]$ with
  $
    \bigl(\tilde{a}_{i-1} \union \alpha(\tilde{z})\bigr) \intersect V(G_n)
    \subsetneq
    \bigl(\tilde{a}_{i} \union \alpha(\tilde{z})\bigr) \intersect V(G_n).
  $
  Hence, $\card{\Ic'} \leq d_3$, as desired.
\end{proof}

We are now ready to prove that on $\Cc$,
every $\LREC[\set{E}]$-formula is equivalent to a $\CInfB[\set{E}]$-formula.

\begin{lemma}
  \label{lem:LREC-Inf-translation}
  For every $\LREC[\set{E}]$-formula $\phi(\tup{x})$,
  there is a $\CInfB[\set{E}]$-formula $\tilde{\phi}(\tup{x})$
  such that for all $G \in \Cc$ and all $\tup{a} \in \Dom{G}{\tup{x}}$,
  we have: $G \models \phi[\tup{a}] \iff G \models \tilde{\phi}[\tup{a}]$.
\end{lemma}

\begin{proof}
  As mentioned above, we proceed by induction on the structure of $\phi$.
  The only interesting case is that of an $\LREC[\set{E}]$-formula
  of the form
  \[
    \phi =
    \lrecx{\tup{u}_1}{\tup{u}_2}{\tup{p}}{\phi_{\graphE}}{\phi_{\graphC}}(\tup{w},\tup{r}).
  \]
  Let $\tup{v}_{\graphE}$ be an enumeration of all variables in $\free(\phi_{\graphE})$
  that are not listed in $\tup{u}_1\tup{u}_2$,
  and let $\tup{v}_{\graphC}$ be an enumeration of all variables in $\free(\phi_{\graphC})$
  that are not listed in $\tup{u}_1\tup{p}$.

  We aim to construct,
  for all integers $n \geq 1$ and $\ell \leq \card{N(G_n)}^{\len{\tup{r}}}-1$,
  a $\CInfB[\set{E}]$-formula $\psi_{n,\ell}(\tup{u}_1,\tup{v}_{\graphE},\tup{v}_{\graphC})$
  such that for all assignments $\alpha$ in $G_n$,
  and all $\tup{a} \in \Dom{G_n}{\tup{u}_1}$,
  \[
    G_n \models \psi_{n,\ell}[\tup{a},\alpha(\tup{v}_{\graphE}),\alpha(\tup{v}_{\graphC})]
    \iff
    (\tup{a},\ell) \in X,
  \]
  where $X$ is the relation defined by $\phi$ in $(G_n,\alpha)$.
  Furthermore, the rank of each $\psi_{n,\ell}$ will be bounded by a constant
  that depends only on $\phi$,
  so that
  \[
    \tilde{\phi} :=
    \!\!\!\!
    \biglor_{\substack{n \geq 1\\\ell < (2n^2+1)^{\len{\tup{r}}}}}
    \!\!\!\!\!\!\!\!
    \left(
      \text{``the universe has size $2n^2$''}
      \land
      \text{``$\tup{r}$ represents the number $\ell$''}
      \land
      \psi_{n,\ell}(\tup{w},\tup{v}_{\graphE},\tup{v}_{\graphC})
    \right)
  \]
  is a $\CInfB[\set{E}]$-formula that is equivalent to $\phi$ on $\Cc$.

  \bigskip\noindent
  \textit{Construction of $\psi_{n,\ell}(\tup{u}_1,\tup{v}_{\graphE},\tup{v}_{\graphC})$:}
  Fix $n \geq 1$ and $\ell \leq \card{N(G_n)}^{\len{\tup{r}}}-1$.
  To simplify the presentation,
  we also fix an assignment $\alpha$ in $G_n$,
  and the graph $\graphG = (\graphV,\graphE)$ with $\graphV \isdef \Dom{G_n}{\tup{u}_1}$
  and $\graphE \isdef \phi_{\graphE}[G_n,\alpha;\tup{u}_1,\tup{u}_2]$;
  the formula $\psi_{n,\ell}(\tup{u}_1,\tup{v}_{\graphE},\tup{v}_{\graphC})$
  we are going to construct will however not depend on $\alpha$.
  For every $\tup{a} \in \graphV$,
  let
  \begin{align*}
    t_{n,\ell}(\tup{a}) :=
    \max\, \{t \in \nat \mid\
    & \text{there is a node $(\tup{a}_0,\ell_0),\dotsc,(\tup{a}_m,\ell_m)$
      in $\graphG^{(\tup{a},\ell)}$ such that $t$ equals} \\
    & \card{
        \set{
          i \in [m] \mid
          (\tilde{a}_{i-1} \union \alpha(\tilde{v}_{\graphE})) \intersect V(G_n)
          \neq
          (\tilde{a}_i \union \alpha(\tilde{v}_{\graphE})) \intersect V(G_n)
        }
      }\}.
  \end{align*}
  By Lemma~\ref{lem:elem-depth} there is a constant $t^*$
  that only depends on $\phi$
  such that
  \[
    t_{n,\ell}(\tup{a}) < t^* \quad \text{for all $\tup{a} \in \graphV$}.
  \]
  In what follows,
  we construct, for all $t \leq t^*$,
  a $\CInfB[\set{E}]$-formula $\psi_{n,\ell}^t(\tup{u}_1,\tup{v}_{\graphE},\tup{v}_{\graphC})$
  such that for all $\tup{a} \in \graphV$ with $t_{n,\ell}(\tup{a}) < t$,
  we have:
  \[
    G_n \models \psi_{n,\ell}^t[\tup{a},\alpha(\tup{v}_{\graphE}),\alpha(\tup{v}_{\graphC})]
    \iff
    (\tup{a},\ell) \in X,
  \]
  where $X$ is the relation defined by $\phi$ in $(G_n,\alpha)$.
  Furthermore, the rank of $\psi_{n,\ell}^t$ will not depend on $n$ or $\ell$.
  The desired formula $\psi_{n,\ell}$ can then be defined as:
  \[
    \psi_{n,\ell} := \psi_{n,\ell}^{t^*}.
  \]

  \bigskip\noindent
  \textit{Construction of $\psi_{n,\ell}^t(\tup{u}_1,\tup{v}_{\graphE},\tup{v}_{\graphC})$:}
  We construct the formulae $\psi_{n,\ell}^t(\tup{u}_1,\tup{v}_{\graphE},\tup{v}_{\graphC})$
  by induction on $t$.
  For $t = 0$,
  we define $\psi_{n,\ell}^0(\tup{u}_1,\tup{v}_{\graphE},\tup{v}_{\graphC})$ to be
  an arbitrary unsatisfiable formula.
  The idea for the construction
  of $\psi_{n,\ell}^{t+1}(\tup{u}_1,\tup{v}_{\graphE},\tup{v}_{\graphC})$ is as follows.
  Let $\tup{a} \in \graphV$, and
  \begin{align*}
    \Qc(\tup{a}) :=
    \{(\tup{a}_m,\ell_m) \mid\
    & ((\tup{a}_0,\ell_0),\dotsc,(\tup{a}_m,\ell_m)) \in V(\graphG^{(\tup{a},\ell)}),\
      \text{and for all $i \in [m]$ we have:} \\
    & (\tilde{a}_{i-1} \union \alpha(\tilde{v}_{\graphE})) \intersect V(G_n)
      =
      (\tilde{a}_{i} \union \alpha(\tilde{v}_{\graphE})) \intersect V(G_n)
    \}.
  \end{align*}
  To check whether $(\tup{a},\ell) \in X$,
  we ``guess'' the set $\hat{X} = \Qc(\tup{a}) \intersect X$,
  and then simply check whether $(\tup{a},\ell) \in \hat{X}$.
  To guess $\hat{X}$,
  we can use an infinite disjunction over all subsets $R$ of $\Qc(\tup{a})$.
  Then we only need to verify for each $R$
  whether $R$ indeed corresponds to $\hat{X}$.
  For the latter, we count, for each pair $(\tup{a}',\ell') \in \Qc(\tup{a})$,
  the number of pairs $(\tup{a}'',\ell'')$ such that $\tup{a}'\tup{a}'' \in \graphE$,
  $\ell'' = \lfloor (\ell'-1)/\card{\graphE \tup{a}''} \rfloor$
  and $(\tup{a}'',\ell'') \in X$,
  and check that $(\tup{a}',\ell') \in R$
  whenever this number belongs to the label of $\tup{a}'$ defined by $\phi_{\graphC}$.
  How do we check whether $(\tup{a}'',\ell'') \in X$?
  If
  $
    (\tilde{a}' \union \alpha(\tilde{v}_{\graphE})) \intersect V(G_n)
    =
    (\tilde{a}'' \union \alpha(\tilde{v}_{\graphE})) \intersect V(G_n),
  $
  that is, if $(\tup{a}'',\ell'') \in \Qc(\tup{a})$,
  then we simply check whether $(\tup{a}'',\ell'') \in R$.
  Otherwise, we use the formula $\psi_{n,\ell''}^t$.

  Let $\phi_{\graphE}'$ and $\phi_{\graphC}'$ be $\CInfB[\set{E}]$-formulae
  that are equivalent to $\phi_{\graphE}$ and $\phi_{\graphC}$, respectively.
  Such formulae exist by the induction hypothesis.
  Using $\phi_{\graphE}'$ it is easy to construct, for each $\ell' \in [0,\ell]$,
  an $\CInfB[\set{E}]$-formula $\chi_{\ell'}(\tup{u}_1,\tup{u}_1',\tup{v}_{\graphE})$
  such that for all $\tup{a},\tup{a}' \in \Dom{G_n}{\tup{u}_1}$,
  \[
    G_n \models \chi_{\ell'}[\tup{a},\tup{a}',\alpha(\tup{v}_{\graphE})]
    \iff
    (\tup{a}',\ell') \in \Qc(\tup{a}).
  \]
  Here, $\tup{u}_1'$ is a tuple of variables that is compatible with,
  but disjoint from $\tup{u}_1$.

  Let $\Qc'$ be the set of all pairs $(\tup{u}',\ell')$,
  where $\ell' \in [0,\ell]$,
  and $\tup{u}'$ is obtained from $\tup{u}_1'$
  by replacing each structure variable with a structure variable
  from $\tup{u}_1$
  and each number variable with an integer from $N(G_n) = [0,2n^2]$.
  Intuitively, each $R \subseteq \Qc'$ corresponds to a guess
  of $\Qc(\tup{a}) \intersect X$ as described above.
  For each $R \subseteq \Qc'$,
  let $\psi_{n,\ell,R}^{t+1}(\tup{u}_1,\tup{v}_{\graphE},\tup{v}_{\graphC})$ be
  \begin{align*}
    & \bigland_{(\tup{u}',\ell') \in R}
      \!\!\biggl(
        \chi_{\ell'}(\tup{u}_1,\tup{u}',\tup{v}_{\graphE})
        \limplies
        \exists \tup{p} \Bigl(
          \phi_{\graphC}'(\tup{u}',\tup{p},\tup{v}_{\graphC})
          \land
          \# \tup{u}'' \bigl(
            \phi_{\graphE}'(\tup{u}',\tup{u}'',\tup{v}_{\graphE})
            \land
            \vartheta_{R,\tup{u}',\ell'}(\tup{u}'')
          \bigr) = \tup{p}
        \Bigr)
      \biggr) \\
    \land
    & \bigland_{\substack{%
        (\tup{u}',\ell') \in \Qc' \\ (\tup{u}',\ell') \notin R}}
      \!\!\biggl(
        \chi_{\ell'}(\tup{u}_1,\tup{u}',\tup{v}_{\graphE})
        \limplies
        \exists \tup{p} \Bigl(
          \lnot \phi_{\graphC}'(\tup{u}',\tup{p},\tup{v}_{\graphC})
          \land
          \# \tup{u}'' \bigl(
            \phi_{\graphE}'(\tup{u}',\tup{u}'',\tup{v}_{\graphE})
            \land
            \vartheta_{R,\tup{u}',\ell'}(\tup{u}'')
          \bigr) = \tup{p}
        \Bigr)
      \biggr)
  \end{align*}
  where
  \begin{multline*}
    \vartheta_{R,\tup{u}',\ell'}(\tup{u}'')\ \isdef\ \\
    \biglor_{\ell'' \in [0,\ell']} \biggl(
      \text{``$\ell'' =
        \left\lfloor\frac{\ell'-1}{\card{\graphE \tup{u}''}}\right\rfloor$''}
      \land
      \Bigl(
        \bigr(
          \chi_{\ell''}(\tup{u}',\tup{u}'',\tup{v}_{\graphE})
          \land
          \text{``$(\tup{u}'',\ell'') \in R$''}
        \bigr)
        \lor
        \psi_{n,\ell''}^t(\tup{u}'',\tup{v}_{\graphE},\tup{v}_{\graphC})
      \Bigr)
    \biggr)
  \end{multline*}
  and ``$(\tup{u}'',\ell'') \in R$'' stands for
  $\biglor_{(\tup{u}^\star,\ell^\star) \in R,\, \ell^\star = \ell''}
   \tup{u}^\star = \tup{u}''$.
  Then it is not hard to see that the formula
  \[
    \psi_{n,\ell}^{t+1}(\tup{u}_1,\tup{v}_{\graphE},\tup{v}_{\graphC})
    \ \isdef\
    \biglor_{\substack{R \subseteq \Qc'\\(\tup{u}_1,\ell) \in R}}
    \psi_{n,\ell,R}^{t+1}(\tup{u}_1,\tup{v}_{\graphE},\tup{v}_{\graphC})
  \]
  is as desired.
  Clearly, the rank of $\psi_{n,\ell}^{t+1}$ does not depend on $n$ or $\ell$.
\end{proof}

To conclude this section,
note that Theorem~\ref{thm:reach-nondef} follows immediately
from Lemma~\ref{lem:LREC-Inf-translation} and Corollary~\ref{cor:locality}.

\section{An Extension of \texorpdfstring{$\LREC$}{LREC}}
\label{sec:lrec-eq}

The proof of the previous section's Theorem~\ref{thm:reach-nondef}
indicates that $\LREC$ is not closed under logical reductions,
not even under very simple first-order reductions.%
\footnote{We defer the definition of logical reductions and what it means to be
  closed under logical reductions to Definition~\ref{def:interpr}
  and Lemma~\ref{lem:transduction-lemma}.
  For first-order reduction, see also \cite{ebbflu95}.}
Indeed, it is easy to see that there is a first-order reduction
that maps a graph $G_n$, for $n \geq 3$,
as defined in Section~\ref{sec:reach-nondef}
to a disjoint union $\hat{G}_n$ of two directed paths on $n$ vertices each,
by identifying vertices in the same layer.
Reachability on the class of all graphs isomorphic to $\hat{G}_n$
for an $n \geq 3$ is easily seen to be $\LREC$-definable.
Hence, if $\LREC$ was closed under first-order reductions,
then reachability on the class of all graphs isomorphic to $G_n$ for some $n$
would be \LREC-definable, contradicting the previous section's results.

In this section, we introduce an extension $\LREC_=$ of $\LREC$
whose data complexity is still in $\LOGSPACE$,
and thus captures $\LOGSPACE$ on directed trees,
while being closed under logical reductions.
The idea is to admit a third formula $\phi_=$ in the $\lrec$-operator
that generates an equivalence relation on the vertices of the graph
defined by $\phi_{\graphE}$.

Let $\tau$ be a vocabulary.
The set of all $\LREC_=[\tau]$-formulae is obtained from $\LREC[\tau]$
by replacing the rule for the $\lrec$-operator from Section~\ref{sec:lrec}
as follows:
If $\bar{u},\bar{v},\bar{w}$ are compatible tuples of variables,
$\bar{p},\bar{r}$ are non-empty tuples of number variables, and
$\phi_{=}$, $\phi_{\graphE}$ and $\phi_{\graphC}$ are $\LREC_=$-formulae,
then the following is an $\LREC_=[\tau]$-formula:
\begin{equation}
  \label{eq:lrec-eq}
  \phi \,\isdef\,
  \lreceq{\tup{u}}{\tup{v}}{\tup{p}}{\phi_{=}}{\phi_{\graphE}}{\phi_{\graphC}}(\tup{w},\tup{r}).
\end{equation}
We let
$
  \free(\phi) \isdef
  \bigl(\free(\phi_{=}) \setminus (\tilde{u} \union \tilde{v})\bigr)
  \union
  \bigl(\free(\phi_{\graphE}) \setminus (\tilde{u} \union \tilde{v})\bigr)
  \union
  \bigl(\free(\phi_{\graphC}) \setminus (\tilde{u} \union \tilde{p})\bigr)
  \union
  \tilde{w} \union \tilde{r}.
$

To define the semantics of $\LREC_=[\tau]$-formulae $\phi$
of the form \eqref{eq:lrec-eq},
let $A$ be a $\tau$-structure and $\alpha$ an assignment in $A$.
Let $\graphV_0 \isdef \Dom{A}{\tup{u}}$
and $\graphE_0 \isdef \phi_{\graphE}[A,\alpha;\tup{u},\tup{v}]$.
We define $\sim$ to be the reflexive, symmetric, transitive closure
of the binary relation $\phi_{=}[A,\alpha;\tup{u},\tup{v}]$ over $\graphV_0$.
Now consider the graph $\graphG = (\graphV,\graphE)$ with
\begin{align*}
  \graphV\, :=\, \graphV_0/_\sim
  \quad \text{and} \quad
  \graphE\, :=\,
  \{
    (\tup{a}/_\sim,\tup{b}/_\sim) \in \graphV^2 \mid \tup{a}\tup{b} \in \graphE_0
  \}.
\end{align*}
To every $\tup{a}/_\sim \in \graphV$ we assign the set
\[
  \graphC(\tup{a}/_\sim)\, \isdef\,
  \set{\num{\tup{n}} \mid \text{there is an $\tup{a}' \in \tup{a}/_\sim$
      with $\tup{n} \in \phi_{\graphC}[A,\alpha[\tup{a}'/\tup{u}];\tup{p}]$}}
\]
of labels.
Then the definition of $X$ can be taken verbatim from Section~\ref{sec:lrec}.
We let $(A,\alpha) \models \phi$ if and only if
$\bigl(\alpha(\tup{w})/_\sim,\num{\alpha(\tup{r})}\bigr) \in X$.
As for $\LREC$, we have:

\begin{theorem}
  \label{theo:lrec=-complexity}
  For every vocabulary $\tau$, and every $\LREC_=[\tau]$-formula $\phi$
  there is a deterministic logspace Turing machine that,
  given a $\tau$-structure $A$ and an assignment $\alpha$ in $A$,
  decides whether $(A,\alpha) \models \phi$.
\end{theorem}

\begin{proof}[Sketch]
  The proof is a straightforward modification
  of the proof of Theorem~\ref{theo:lrec-complexity}.
  The only difference is that, when we deal with $\LREC_=$-formulae
  of form \eqref{eq:lrec-eq},
  we use the vertex set $\graphV$, the edge set $\graphE$, and the labels $\graphC(\cdot)$
  as defined above to compute the set $X$.
  It is easy to compute these sets by first computing the relation $\sim$
  from $\phi_=[A,\alpha;\tup{u},\tup{v}]$ using Reingold's logspace algorithm
  for undirected reachability \cite{rei05}.
  Note that once $\sim$ has been obtained,
  the equivalence class of every element $\tup{a} \in \Dom{A}{\tup{u}}$
  can be determined.
\end{proof}

The following example shows that undirected graph reachability is
definable in $\LREC_=$. This does not involve an implementation of
Reingold's algorithm in our logic, but just uses the observation that
the computation of the equivalence relation $\sim$ boils down to the
computation of undirected reachability.

\begin{example}[Undirected reachability]
  The following $\LREC_=$-formula defines undirected graph reachability:
  \begin{align*}
    \phi(s,t) \ \isdef\
    \lreceq{x}{y}{p}{\phi_{=}(x,y)} {\phi_{\graphE}(x,y)}{\phi_{\graphC}(x,p)}(s,1),
  \end{align*}
  where $\phi_=(x,y) \isdef E(x,y)$, $\phi_{\graphE}(x,y) \isdef \lnot x = x$
  and $\phi_{\graphC}(x,p) \isdef x = t$.
  To see this, let $G$ be an undirected graph and $\alpha$ an assignment in $G$.
  Define $\sim$, $\graphV$, $\graphE$, $\graphC$ and the set $X$ as above.
  Clearly, the set $\graphV$ consists of the connected components of $G$.
  Furthermore, the set $\graphE$ is empty since $\phi_E$ is unsatisfiable.
  Therefore, for all $v \in V(G)$ we have $(v/_\sim,1) \in X$
  iff $0 \in \graphC(v/_\sim)$.
  The latter is true precisely if $\alpha(t) \in v/_\sim$,
  i.e., if $v$ and $\alpha(t)$ are in the same connected component of $G$.
  It follows that for all $v,w \in V(G)$ we have $G \models \phi[v,w]$
  if and only if $v$ and $w$ are in the same connected component of $G$,
  that is, if there is a path from $v$ to $w$ in $G$.
  \varqed
\end{example}

\begin{remark}
  It follows immediately from the previous example that
  $\STCC\leq\LREC_=$.
  Actually, the containment is strict,
  because $\LREC\not\leq\STCC$ by Corollary~\ref{cor:LREC-vs-TCC}.
  Since in $\STCC$ (and actually in $\STC$)
  it is possible to transform trees into directed trees,
  the results from Section~\ref{sec:tree-canon} imply
  that $\LREC_=$ captures \LOGSPACE\ on the class of all trees,
  directed as well as undirected.
  Note also that $\LREC_= \leq \FPC$.
\end{remark}

To conclude this section,
we show that $\LREC_=$ is closed under logical reductions.
We first introduce \emph{$\logic{L}$-transductions}
(also known as \emph{$\logic{L}$-interpretations} \cite{ebbflu95}):

\begin{definition}[Transduction]
  \label{def:interpr}
  Let $\logic{L}$ be a logic, let $\tau_1,\tau_2$ be vocabularies
  and let $\ell \geq 1$.
  \begin{enumerate}[leftmargin=*]
  \item
    An \emph{$\ell$-ary $\logic{L}[\tau_1,\tau_2]$-transduction} is a tuple
    \[
      \Theta \,=\,
      \Bigl(
        \theta_V(\tup{u}),
        \theta_\approx(\tup{u},\tup{v}),
        \bigl(
          \theta_R(\tup{u}_{R,1},\dotsc,\tup{u}_{R,{\ar(R)}})
        \bigr)_{R \in \tau_2}
      \Bigr)
    \]
    of $\logic{L}[\tau_1]$-formulae,
    where $\tup{u},\tup{v}$ are compatible $\ell$-tuples of variables
    and for every $R \in \tau_2$ and $i \in [\ar(R)]$,
    $\tup{u}_{R,i}$ is an $\ell$-tuple of variables
    that is compatible to $\tup{u}$.
  \item
    Let $A$ be a $\tau_1$-structure
    such that $\theta_V[A;\tup{u}]$ is non-empty.
    We define a $\tau_2$-structure $\Theta[A]$ as follows.
    We let $\approx$ be the reflexive, symmetric, transitive closure
    of the binary relation  $\theta_\approx[A;\tup{u},\tup{v}]$,
    and call $\approx$ the equivalence relation \emph{generated}
    by $\theta_\approx[A;\tup{u},\tup{v}]$.
    Let
    \[
      V(\Theta[A]) \,\isdef\, \quot{\theta_V[A;\tup{u}]}{\approx},
    \]
    and for each $R \in \tau_2$, let
    \begin{align*}
      \hspace{4em}
      & R(\Theta[A]) \,\isdef\, \{
        (\eclass{\tup{a}_1}{\approx},\dotsc,\eclass{\tup{a}_{\ar(R)}}{\approx})
        \mid \\
      & \qquad\qquad\qquad\qquad\qquad
       \tup{a}_1,\dotsc,\tup{a}_{\ar(R)} \in \theta_V[A;\tup{u}],\,
       A \models \theta_R[\tup{a}_1,\dotsc,\tup{a}_{\ar(R)}]
      \}.
    \end{align*}
  \end{enumerate}
\end{definition}

So, informally, a $\logic{L}[\tau_1,\tau_2]$-transduction defines a mapping
from structures over the first vocabulary, $\tau_1$,
into structures over the second vocabulary, $\tau_2$,
via $\logic{L}[\tau_1]$-formulae.

\begin{example}
  \label{ex:transduct}
  Consider the $\FO[\set{E},\set{E}]$-transduction
  $\Theta = (\theta_V(x),\theta_\approx(x,y),\theta_E(x,y))$
  with $\theta_V(x) \isdef x = x$,
  $\theta_\approx(x,y) \isdef \forall z \bigl(E(x,z) \lequiv E(y,z)\bigr)$
  and $\theta_E(x,y) \isdef E(x,y)$.
  Recall the definition of the graphs $G_n$ from Section~\ref{sec:reach-nondef}.
  For $n > 3$,
  the equivalence relation $\approx$ generated by $\theta_\approx[G_n;x,y]$
  is $\theta_\approx[G_n;x,y]$ itself.
  It relates any two vertices that occur in the same layer of $G_n$.
  Hence, for $n > 3$,
  $\Theta[G_n]$ is the disjoint union of two paths of length $n$.
  \varqed
\end{example}

The following lemma shows that $\LREC_=$ is closed under $\LREC_=$-reductions.
Precisely, this means that:

\begin{lemma}
  \label{lem:transduction-lemma}
  Let $\tau_1,\tau_2$ be vocabularies, let $\ell \geq 1$,
  let
  \[
    \Theta \,=\,
    \Bigl(
      \theta_V(\tup{u}),
      \theta_\approx(\tup{u},\tup{v}),
      \bigl(
        \theta_R(\tup{u}_{R,1},\dotsc,\tup{u}_{R,\ar(R)})
      \bigr)_{R \in \tau_2}
    \Bigr)
  \]
  be an $\ell$-ary $\LREC_=[\tau_1,\tau_2]$-transduction,
  and let $\phi(x_1,\dotsc,x_\kappa,p_1,\dotsc,p_\lambda)$
  be an $\LREC_=[\tau_2]$-formula with $x_1,\dotsc,x_\kappa$ structure variables
  and $p_1,\dotsc,p_\lambda$ number variables.

  Then there is an $\LREC_=[\tau_1]$-formula
  $
    \phi^{-\Theta}(\tup{u}_1,\dotsc,\tup{u}_\kappa,
      \tup{q}_1,\dotsc,\tup{q}_\lambda),
  $
  where $\tup{u}_1,\dotsc,\tup{u}_\kappa$ are compatible with $\tup{u}$
  and $\tup{q}_1,\dotsc,\tup{q}_\lambda$ are $\ell$-tuples of number variables,
  such that for all $\tau_1$-structures $A$ where $\Theta[A]$ is defined,
  all $\tup{a}_1,\dotsc,\tup{a}_\kappa \in A^{\tup{u}}$
  and all $\tup{n}_1,\dotsc,\tup{n}_\lambda \in N(A)^\ell$,
  \begin{align*}
    A \models
    \phi^{-\Theta}[\tup{a}_1,\dotsc,\tup{a}_\kappa,
      \tup{n}_1,\dotsc,\tup{n}_\lambda]
    \iff\
    & \eclass{\tup{a}_1}{\approx},\dotsc,\eclass{\tup{a}_\kappa}{\approx}
      \in V(\Theta[A]), \\
    & \num[A]{\tup{n}_1},\dotsc,\num[A]{\tup{n}_\lambda}
      \in N(\Theta[A]),\ \text{and} \\
    & \Theta[A] \models
      \phi\bigl[
        \eclass{\tup{a}_1}{\approx},\dotsc,\eclass{\tup{a}_\kappa}{\approx},
        \num[A]{\tup{n}_1},\dotsc,\num[A]{\tup{n}_\lambda}
      \bigr],
  \end{align*}
  where $\approx$ is the equivalence relation as defined
  in Definition~\ref{def:interpr}.
\end{lemma}

\begin{proof}
  The proof is by induction on the structure of $\phi$.
  Without loss of generality,
  we assume that $\phi$ neither contains implication ($\limplies$)
  nor biimplication ($\lequiv$).

  To simplify the presentation,
  we consider a fixed $\tau_1$-structure $A$ where $\Theta[A]$ is defined
  and let $\approx$ be the equivalence relation as defined
  in Definition~\ref{def:interpr}.
  We also consider fixed $\tup{a}_1,\dotsc,\tup{a}_\kappa \in A^{\tup{u}}$
  and $\tup{n}_1,\dotsc,\tup{n}_\lambda \in N(A)^\ell$.
  The reader should consider $A$ and these tuples to be universally quantified
  in the statements where they occur.

  From $\theta_\approx(\tup{u},\tup{v})$,
  it is easy to construct an $\LREC_=$-formula
  $\theta'_\approx(\tup{u},\tup{v})$
  such that $\theta'_\approx[A;\tup{u},\tup{v}]$ is the equivalence relation
  generated by $\theta_\approx[A;\tup{u},\tup{v}]$,
  that is, the reflexive, symmetric, and transitive closure
  of $\theta_\approx(\tup{u},\tup{v})$.
  Let
  $
    \chi_s(\tup{u}) \isdef
    \exists \tup{v}\, \bigl(
      \theta_V(\tup{v})
      \land
      \theta'_\approx(\tup{u},\tup{v})
    \bigr).
  $
  Then for all $\tup{a} \in A^{\tup{u}}$,
  \[
    A \models \chi_s[\tup{a}]
    \iff
    \eclass{\tup{a}}{\approx} \in V(\Theta[A]).
  \]
  Using the construction from the proof of \cite[Lemma~2.4.3]{laubner11diss},
  we can construct an $\LREC_=$-formula $\delta^\#_V(\tup{q})$
  such that for all $\tup{n} \in N(A)^\ell$ we have
  $A \models \delta^\#_V[\tup{n}]$ whenever
  $\num[A]{\tup{n}} = \card{V(\Theta[A])}$.
  Hence, for
  $
    \chi_n(\tup{q}) \isdef
    \exists \tup{q}' \bigl(
      \delta^\#_V(\tup{q}') \land \text{``$\tup{q} \leq \tup{q}'$''}
    \bigr)
  $
  and for all $\tup{n} \in N(A)^\ell$,
  \[
    A \models \chi_n[\tup{n}]
    \iff
    \num[A]{\tup{n}} \in N(\Theta[A]).
  \]
  Finally, let
  \[
    \chi \isdef
    \bigland_{i \in [\kappa]} \chi_s(\tup{u}_i) \land
    \bigland_{i \in [\lambda]} \chi_n(\tup{q}_i)
  \]
  Then,
  \begin{multline*}
    A \models
    \chi[\tup{a}_1,\dotsc,\tup{a}_\kappa,\tup{n}_1,\dotsc,\tup{n}_\lambda]
    \iff\ \\
    \eclass{\tup{a}_1}{\approx},\dotsc,\eclass{\tup{a}_\kappa}{\approx}
    \in V(\Theta[A])\ \ \text{and}\ \
    \num[A]{\tup{n}_1},\dotsc,\num[A]{\tup{n}_\lambda}
    \in N(\Theta[A]).
  \end{multline*}

  Given $\phi(x_1,\dotsc,x_\kappa,p_1,\dotsc,p_\lambda)$
  we now construct
  $
    \phi^{-\Theta}(
      \tup{u}_1,\dotsc,\tup{u}_\kappa,\tup{q}_1,\dotsc,\tup{q}_\lambda
    )
  $
  inductively as follows:
  \begin{enumerate}[leftmargin=*]
  \item
    Suppose that $\phi = R(x_{i_1},\dotsc,x_{i_k})$,
    where $i_1,\dotsc,i_k \in [\kappa]$.
    Let $\Ic \isdef \set{i_1,\dotsc,i_k}$.
    Then,
    \[
      \phi^{-\Theta} \,\isdef\,
      \chi \land
      \left(\exists \tup{v}_i\right)_{i \in \Ic}
      \left(
        \bigland_{i \in \Ic} \theta'_\approx(\tup{u}_i,\tup{v}_i)
        \land
        \bigland_{i \in \Ic} \theta_V(\tup{v}_i)
        \land
        \theta_R(\tup{v}_{i_1},\dotsc,\tup{v}_{i_k})
      \right).
    \]
  \item
    If $\phi = x_i = x_j$, where $i,j \in [\kappa]$, then
    $
      \phi^{-\Theta} \,\isdef\,
      \chi \land \theta'_\approx(\tup{u}_i,\tup{u}_j).
    $
  \item
    If $\phi = p_i \star p_j$,
    where $\star \in \set{=,\leq}$ and $i,j \in [\lambda]$,
    then,
    $
      \phi^{-\Theta} \,\isdef\,
      \chi \land \text{``$\tup{q}_i \star \tup{q}_j$''}.
    $
  \item
    If $\phi = \lnot \psi$, then
    $
      \phi^{-\Theta} \,\isdef\,
      \chi \land \lnot \psi^{-\Theta}.
    $
  \item
    If $\phi = \psi_1 \star \psi_2$, where $\star \in \set{\land,\lor}$,
    then
    $
      \phi^{-\Theta} \,\isdef\,
      \psi_1^{-\Theta} \star \psi_2^{-\Theta}.
    $
  \item
    Suppose that $\phi = Q u\, \psi$ with $Q \in \set{\forall,\exists}$
    and $u \in \set{x_1,\dotsc,x_\kappa,p_1,\dotsc,p_\lambda}$.
    In case that $Q = \forall$ and $u = x_i$,
    we let
    $
      \phi^{-\Theta} \,\isdef\,
      \chi \land
      \forall \tup{u}_i \bigl(\chi_s(\tup{u}_i) \limplies \psi^{-\Theta}\bigr).
    $
    The other cases can be dealt with similarly.
  \item
    Suppose that
    $
      \phi =
      \# (x_{i_1},\dotsc,x_{i_k},p_{i_{k+1}},\dotsc,p_{i_{k+m}})\, \psi
      =
      (p_{j_1},\dotsc,p_{j_{k'}})
    $.
    Based on the construction from the proof
    of \cite[Lemma~2.4.3]{laubner11diss},
    it is possible to construct
    an $\LREC_=$-formula $\delta(\tup{r}_1,\dotsc,\tup{r}_{k'})$
    such that for all $\tup{m}_1,\dotsc,\tup{m}_{k'} \in N(A)^\ell$,
    \begin{align*}
      \hspace{4em} & A \models \delta[\tup{m}_1,\dotsc,\tup{m}_{k'}] \\
      & \iff
      \Bigl\lvert
        \Bigl\{
          \bigl(
            \eclass{\tup{a}_{i_1}}{\approx},\dotsc,
            \eclass{\tup{a}_{i_k}}{\approx},
            \num[A]{\tup{n}_{i_{k+1}}},\dotsc,\num[A]{\tup{n}_{i_{k+m}}}
          \bigr)
          \mid \\
      & \qquad\qquad\qquad\qquad\qquad\qquad
          A \models
          \psi^{-\Theta}[\tup{a}_1,\dotsc,\tup{a}_\kappa,
            \tup{n}_1,\dotsc,\tup{n}_\lambda]
        \Bigr\}
      \Bigr\rvert
      =
      \num[A]{\tup{m}_1,\dotsc,\tup{m}_{k'}},
    \end{align*}
    where
    \[
      \num[A]{\tup{m}_1,\dotsc,\tup{m}_{k'}}
      \,\isdef\,
      \sum_{s=1}^{k'} \num[A]{\tup{m}_s} \cdot \card{N(\Theta[A])}^{s-1}.
    \]
    We then let
    $
      \phi^{-\Theta} \,\isdef\,
      \chi \land \delta(\tup{q}_{j_1},\dotsc,\tup{q}_{j_{k'}}).
    $
  \item
    Suppose that
    $
      \phi =
      \lreceq{\tup{u}'}{\tup{v}'}{\tup{p}'}{\phi_=}{\phi_{\graphE}}{\phi_{\graphC}}
      (\tup{w}',\tup{r}').
    $
    Then,
    \[
      \phi^{-\Theta} \,\isdef\,
      \chi \land \exists \tup{r} \bigl(
        \beta(\tup{r}'',\tup{r}) \land
        \lreceq
          {\tup{u}''}
          {\tup{v}''}
          {\tup{p}''}
          {\phi_=^{-\Theta}}
          {\phi_{\graphE}^{-\Theta}}
          {\phi_{\graphC}^{-\Theta}}
          (\tup{w}'',\tup{r})
        \bigr),
    \]
    where $\tup{u}'',\tup{v}'',\tup{w}'',\tup{p}'',\tup{r}''$
    are obtained from $\tup{u}',\tup{v}',\tup{w}',\tup{p}',\tup{r}'$
    by replacing, for each $i \in [\kappa]$, the variable $x_i$ by $\tup{u}_i$,
    and for each $i \in [\lambda]$, the variable $p_i$ by $\tup{q}_i$;
    $\tup{r}$ is a tuple of number variables
    of length $\ell \cdot \len{\tup{r}'}$;
    and $\beta$ is defined as follows.
    For simplicity, assume that $\tup{r}' = (p_1,\dotsc,p_k)$.
    Hence, $\tup{r}'' = (\tup{q}_1,\dotsc,\tup{q}_k)$.
    The formula $\beta(\tup{r}'',\tup{r})$ has the property
    that for all $\tup{m} \in N(A)^{\ell \cdot k}$,
    \[
      A \models \beta[\tup{n}_1,\dotsc,\tup{n}_k,\tup{m}]
      \iff
      \num[A]{\tup{m}}
      =
      \sum_{s=1}^k \num[A]{\tup{n}_s} \cdot \card{N(\Theta[A])}^{s-1}.
    \]
    Note that, since $\card{N(\Theta[A])} \leq \card{N(A)}^\ell$,
    the tuple $\tup{m}$ is long enough to hold the sum on the right hand side.
    Constructing $\beta$ as desired is a not too difficult exercise.
  \end{enumerate}
  It is straightforward, though tedious,
  to verify that $\phi^{-\Theta}$ is as desired.
\end{proof}

\section{Capturing Logspace on Interval Graphs}\label{sec:interval}
\label{sec:interval-eq}

With the added expressive power of $\LREC_=$,
it is not only possible to capture $\LOGSPACE$ on the class of all trees,
but also on the class of all interval graphs,
as we shall show in this section.
Basically, interval graphs are graphs whose vertices are closed intervals,
and whose edges join any two distinct intervals with a non-empty intersection.
They form a well-established and widely investigated class of graphs,
and it was recently shown \cite{koebler10interval}
(see also \cite{laubner11diss})
that interval graph canonisation is in $\LOGSPACE$.

To prove that $\LREC_=$ captures $\LOGSPACE$ on interval graphs,
we proceed as in the case of directed trees.
First, we describe an $\LREC_=$-definable canonisation procedure
for interval graphs,
and then we use the fact that $\DTC$ (and hence $\LREC_=$)
captures $\LOGSPACE$ on ordered structures.
Our canonisation procedure combines algorithmic techniques
from \cite{koebler10interval}
with the logical definability framework in \cite{lau10}.
Parts of this section can be found in more detail in \cite{laubner11diss}.

\subsection{Background on Interval Graphs}\label{subsec:defIntG}

In this section,
we define interval graphs and state some basic properties.
For a more detailed exposition,
we refer the reader to \cite{laubner11diss}.

\begin{definition}[Interval graph, interval representation]
  Given a finite collection $\mathcal I$
  of closed intervals $I_i = [a_i,b_i] \subset \nat$,
  let $G_{\mathcal I} = (V,E)$ be the graph with vertex set $V = \mathcal I$,
  joining two distinct intervals $I_i,I_j \in V$ by an edge whenever
  $I_i \cap I_j \neq \emptyset$.
  We call $\mathcal I$ an \emph{interval representation} of a graph $G$
  if $G \isomorphic G_{\mathcal I}$.
  A graph $G$ is an \emph{interval graph}
  if there is an interval representation of $G$.
\end{definition}

\noindent
Figure~\ref{fig:interval-graph} shows an interval graph $G$
together with an interval representation of $G$.
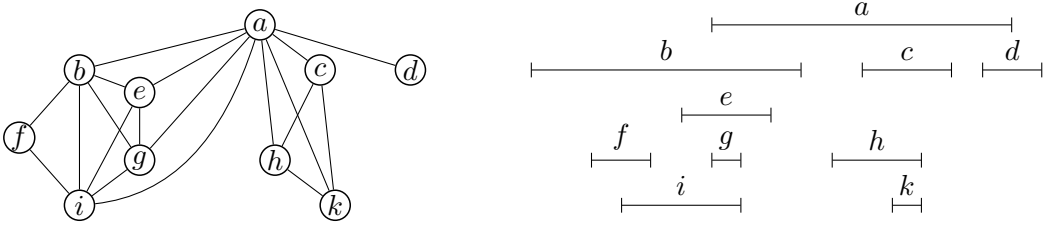
\begin{figure}
  \centering
  \begin{tikzpicture}
    \begin{scope}[xscale=0.4,yscale=0.6]
      \tikzstyle{interval}=[draw]
      \tikzstyle{ilabel}=[midway,above]

      \path[interval,|-|] (7,5) -- (17,5) node[ilabel] (a) {$a$};
      \path[interval,|-|] (1,4) -- (10,4) node[ilabel] (b) {$b$};
      \path[interval,|-|] (12,4) -- (15,4) node[ilabel] (c) {$c$};
      \path[interval,|-|] (16,4) -- (18,4) node[ilabel] (d) {$d$};
      \path[interval,|-|] (6,3) -- (9,3) node[ilabel] (e) {$e$};
      \path[interval,|-|] (3,2) -- (5,2) node[ilabel] (f) {$f$};
      \path[interval,|-|] (7,2) -- (8,2) node[ilabel] (g) {$g$};
      \path[interval,|-|] (11,2) -- (14,2) node[ilabel] (h) {$h$};
      \path[interval,|-|] (4,1) -- (8,1) node[ilabel] (i) {$i$};
      \path[interval,|-|] (13,1) -- (14,1) node[ilabel] (k) {$k$};
    \end{scope}
    \begin{scope}[xshift=-8cm,xscale=0.4,yscale=0.6]
      \node[vertex4] (a) at (7+5,5) {$a$};
      \node[vertex4] (b) at (1+4.5+0.5,4) {$b$};
      \node[vertex4] (c) at (12+1.5+0.5,4) {$c$};
      \node[vertex4] (d) at (16+1,4) {$d$};
      \node[vertex4] (e) at (6+1.5+0.5,3.5) {$e$};
      \node[vertex4] (f) at (3+1,2.5) {$f$};
      \node[vertex4] (g) at (7+0.5+0.5,2) {$g$};
      \node[vertex4] (h) at (11+1.5,2) {$h$};
      \node[vertex4] (i) at (4+2,1) {$i$};
      \node[vertex4] (k) at (13+0.5+1,1) {$k$};

      \path (a) edge (b) edge (c) edge (d) edge (e) edge (g) edge (h)
                edge[bend left=30] (i) edge (k);
      \path (b) edge (e) edge (f) edge (g) edge (i);
      \path (c) edge (h) edge (k);
      \path (e) edge (g) edge (i);
      \path (f) edge (i);
      \path (g) edge (i);
      \path (h) edge (k);
    \end{scope}
  \end{tikzpicture}
  \caption{An interval graph $G$ and an interval representation of $G$.}
  \label{fig:interval-graph}
\end{figure}

An interval representation $\mathcal I$ of a graph $G$ is called \emph{minimal}
if the set $\bigcup \mathcal I \subset \nat$ is of minimum size
among all interval representations of $G$.
Clearly, for any interval representation $\mathcal I$
there exists a minimal interval representation $\mathcal I_{\min}$
such that $G_{\mathcal I} \isomorphic G_{\mathcal I_{\min}}$.

Recall that a \emph{clique} of a graph $G = (V, E)$
is a set $C \subseteq V$ such that the subgraph of $G$ induced by $C$
is complete.
A maximal clique, or \emph{max clique}, of $G$
is a clique of $G$ that is not properly contained in another clique of $G$.
We denote the set of all max cliques of $G$ by $\mathcal M_G$.
Let $\mathcal{I}$ be a minimal interval representation of $G$
and $I_v$ denote the interval in $\mathcal{I}$ that corresponds
to vertex $v \in V$.
Then $M(k) = \{ v \bigmid k \in I_v\}$ is a max clique of $G$ for every $k$  for which $M(k)$  is non-empty.
Furthermore, for any max clique $M$ of $G$,
there is a $k \in \nat$ with $M = M(k)$.
Thus, any minimal interval representation of $G$
induces a linear order on $\mathcal{M}_G$
which has the property that each vertex is contained
in \emph{consecutive} max cliques.
It is known \cite{gilmore64characterization,moehring84algorithmic}
that a graph $G$ is an interval graph if and only if
its max cliques can be brought into a linear order,
so that each vertex of $G$ is contained in consecutive max cliques.

Thus, max cliques play an important role for the structure of interval graphs.
Our canonisation procedure essentially relies on bringing the max cliques
of an interval graph into a suitable order.

The maximal cliques of an interval graph $G = (V,E)$ can be handled
rather easily in our logic.
Let $\ngh(v)$ denote the closed neighbourhood of a vertex $v$ in $G$,
that is, the set containing $v$ and all vertices adjacent to $v$.
As shown in \cite{lau10},
the max cliques of $G$ can be identified by the vertex pairs $(u,v) \in V^2$
with the property that $\ngh(u) \cap \ngh(v)$ is a clique in $G$,
and for no other pair $(u',v') \in V^2$
where $\ngh(u') \cap \ngh(v')$ is a clique in $G$
it holds that $\ngh(u) \cap \ngh(v) \subsetneq \ngh(u') \cap \ngh(v')$:

\begin{lemma}[\cite{lau10}, Lemma IV.1]\label{lem:maxCliquesDefable}
Let $G$ be an interval graph and let $M$ be a max clique of $G$. Then there are vertices $u,v \in M$, not necessarily distinct, such that $M = \ngh(u)\cap \ngh(v)$.
\end{lemma}

\noindent
In particular, the max cliques of $G$ as well as the equivalence relation
on vertex pairs defining the same max clique are first-order definable.

\subsection{Modular Decompositions}
\label{modular decomposition tree}

Our canonisation procedure relies on a specific decomposition of graphs,
known as \emph{modular decomposition},
which was first introduced by Gallai \cite{gallai67transitiv}.
The basic building blocks of modular decompositions are \emph{modules}.
Given a graph $G=(V,E)$,
a set $W \subseteq V$ is a \emph{module} of $G$
if for all vertices $v \in V \setminus W$
either $\{v\}\times W \subseteq E$ or $(\{v\}\times W) \cap E = \emptyset$.
Note that $V$ and all singleton vertex sets are modules of $G$,
called \emph{trivial modules}.
We call a module $W$ proper if $W\subsetneq V$.

Gallai's modular decomposition is based on the following:
If $G$ is not connected, then its connected components $W_1,\dotsc,W_k$
are clearly proper modules.
Similarly, if the complement graph $G^c$ of $G$ is not connected,
then the connected components $W_1,\dotsc,W_k$ of $G^c$ are
proper modules of $G$.
For graphs $G$ with more than one vertex where both $G$ and $G^c$ are connected,
Gallai  shows in \cite{gallai67transitiv} that the set of maximal proper modules of $G$ is a partition of $G$'s vertex set.
We base our modular decomposition on the same properties, only for the last one we use a slightly different partition
into modules  $W_1,\dotsc,W_k$, which we define in  Section~\ref{sec:DecompModulesLG}.%
\footnote{The main difference
between our decomposition and Gallai’s is that we do not bother to
create extra modules for sets of pairwise connected twins
since we can handle them perfectly well with our methods.}
Let  $\mathcal W_G$ be the set of modules $W_1,\dotsc,W_k$
and let $\sim_G$ be the equivalence relation on $V$ corresponding to
the partition $\mathcal W_G$
(i.e., $v \sim_G w$ whenever $v,w \in W_i$ for some $i \in [k]$).
Let us consider the graph
\[
  L_G := (V/_{\sim_{G}},E_{L_G}),
  \quad \text{where} \quad
  E_{L_G} := \set{(u/_{\sim_G},v/_{\sim_G}) \mid (u,v) \in E}.
\]
Intuitively, $L_G$ is the graph obtained from $G$
by collapsing all the modules in $\mathcal W_G$ into single vertices.
Since each pair of modules $W_i,W_j \in \mathcal W_G$, $i \neq j$,
is either completely connected or completely disconnected,
$G$ is completely determined by $L_G$ and the graphs $G[W_i]$, for $i \in [k]$,
where $G[W_i]$ denotes the subgraph of $G$ induced by the vertices in $W_i$.
By decomposing the $G[W_i]$, $i\in[k]$, inductively until we arrive at singleton sets everywhere,
we obtain $G$'s  modular decomposition.

We  define the \emph{modular decomposition tree} $T(G)$ of a graph $G$ recursively.
If $|V|=1$, then $T(G)$ is the rooted tree that consists of only one vertex, vertex $V\!$,
which is the root of $T(G)$.
Let $|V|>1$. Then, the modular decomposition tree
$T(G)$ is a rooted tree which consists of a vertex $V\!$, which is the root of $T(G)$,
and of subtrees $T(G[W])$ for all $W\in \mathcal W_G$.
We obtain $T(G)$ by adding an edge
from $V$ to the root of  $T(G[W])$ for all $W\in \mathcal W_G$.
This modular decomposition tree is uniquely determined for every graph $G$
\cite{gallai67transitiv}.

Notice that for an interval graph $G$ where $G^c$ is not connected,
all except one connected component of $G^c$ must contain only a single vertex.
Each of these single vertices is adjacent to all other vertices in $G$.
We call a vertex with that property an \emph{apex}.
Thus, if $G$ is an interval graph with $G^c$ disconnected,
then $\mathcal W_G= \bigcup_{a\in A} \{\{a\}\} \cup \{V\setminus A\}$ where $A$ is the set of apices, and
the graph $L_G$ is isomorphic to a clique.
Also, if $G$ contains an apex, then either $|V|=1$ or   $G^c$ is not connected.

\medskip
The following three sections
are about defining and canonising the graph $L_G$ for an interval graph $G$.
This is easy for unconnected graphs $G$ or graphs that have at least one apex.
Thus, we will consider connected graphs without any apices.

\subsection{Extracting Information About the Order of Maximal Cliques}
\label{subsec:orderingMaxCliques}

Throughout this section let $G$ be a connected interval graph without any apices.

We call a max clique $C$ a possible \emph{end} of $G$
if there is a minimal interval representation $\mathcal I$ of $G$ so that $C$
is minimal with respect to the order induced by $\mathcal I$.

Now we pick a max clique $M$ of $G$. We assume it to be a possible end of $G$, and give a
recursive procedure that turns out to recover all the information about the order of the max cliques
induced by choosing $M$ as an end of $G$.

Let $M\in \mathcal M_G$.
The binary relation $\prec_M$\index{ @$\prec_M$} is defined recursively on the elements of $\mathcal M_G$ as follows:
\begin{align*}
\mbox{Initialisation:}& \quad M \prec_M C \mbox{ for all } C \in \mathcal M_G \setminus \{M\}\\
  C\prec_M D & \quad\mbox{if }
  \begin{cases} \exists E \in \mathcal M_G \mbox{ with } E \prec_M D \mbox{ and }
  (E\cap C) \setminus D \neq \emptyset \quad\mbox{or}\\
   \exists E \in \mathcal M_G \mbox{ with } C \prec_M E \mbox{ and }
   (E\cap D) \setminus C \neq \emptyset. \tag{$\bigstar$} \end{cases}\label{eqn:indDef}
\end{align*}

By exploiting the definition's symmetry, $\prec_M$ can be defined through a
reachability query in the undirected graph $O_M$, which has pairs of max
cliques from $\mathcal M_G$ as its vertices, and in which two vertices $(A,B)$
and $(C,D)$ are connected by an edge whenever $A \prec_M B$ implies $C \prec_M
D$ with one application of (\ref{eqn:indDef}).
Hence:

\begin{lemma}\label{lem:maxcliqueOrderSTCdefinable}
There exists an \STC-formula that for any interval graph $G$ and for any max clique $M$ of $G$ defines the relation $\prec_M$.
\end{lemma}

We now state a few important properties of $\prec_M$.
Recall that a binary relation $R$ on a set $A$ is \emph{asymmetric} if $ab \in R$ implies $ba \not\in R$ for all $a,b\in A$.
In particular, asymmetric relations are irreflexive.

\begin{lemma}[\cite{lau10}, Lemma~IV.3, Corollary~IV.6, Lemma~IV.7] \label{lem:partialOrderEqEnd}
Let $M$ be a max clique of an interval graph $G$.
Then the following properties are equivalent:
\begin{itemize}[leftmargin=*]
	\item $\prec_M$ is asymmetric,
	\item $\prec_M$ is a strict weak order (that is, $\prec_M$ is irreflexive, transitive, and incomparability is an equivalence relation),
	\item $M$ is a possible end of $G$.
\end{itemize}
\end{lemma}

Since $\prec_M$ is \STC-definable  and
asymmetry of $\prec_M$ is \FO-definable, the preceding lemma gives us a way to define possible ends of interval graphs in \STCC.

\begin{lemma}\label{lem:moduleImpIncomparable}
Let $\mathcal C \subset \mathcal M_G$ be a set of max cliques with $M \not\in \mathcal C$.
Suppose that for all $B\in \mathcal M_G\setminus\mathcal C$ and any $C,C'\in\mathcal C$ it holds that $B\cap C = B\cap C'$.
Then the max cliques in $\mathcal C$ are mutually incomparable with respect to $\prec_M$.
\end{lemma}

\proof
By a derivation chain of length $k$ we mean a finite sequence $X_0 \prec_M Y_0$, $X_1 \prec_M Y_1$, $\ldots$, $X_k \prec_M Y_k$ such that $X_0 = M$ and for each $i \in [k]$, the relation $X_i \prec_M Y_i$ follows from $X_{i-1} \prec_M Y_{i-1}$ by one application of (\ref{eqn:indDef}). Clearly, whenever it holds that $X \prec_M Y$ there is a derivation chain that has $X \prec_M Y$ as its last element.

Suppose for contradiction that there are $C,C'\in\mathcal C$ with $C\prec_M C'$. Let $M \prec_M Y_0$, $X_1 \prec_M Y_1$, $\ldots$, $X_k \prec_M Y_k$ be a derivation chain for $C\prec_M C'$. Since $X_k = C$, $Y_k = C'$, and $M\not\in \mathcal C$, there is a largest index $i$ so that either $X_i$ or $Y_i$ is not contained in $\mathcal C$.

If $X_i \not\in \mathcal C$, then $X_{i+1}\in\mathcal C$ and $Y_i=Y_{i+1} \in\mathcal C$ and it holds that $X_i \cap X_{i+1}\setminus Y_{i+1} \neq \emptyset$. Consequently, $X_i \cap X_{i+1} \neq X_i \cap Y_{i+1}$, contradicting the assumption of the lemma. Similarly, if $Y_i \not\in \mathcal C$, then $Y_{i+1}\in\mathcal C$ and $X_i = X_{i+1}\in\mathcal C$ and it holds that $Y_i\cap Y_{i+1}\setminus X_{i+1} \neq \emptyset$. Thus, $Y_i \cap Y_{i+1} \neq Y_i \cap X_{i+1}$, again a contradiction.\qed

The \emph{span} of a vertex $v\in V$ in $G$, denoted $\spn(v)$,
is the number of max cliques of $G$ that $v$ is contained in.
Recall from Section~\ref{subsec:defIntG}
that the equivalence relation on vertex pairs defining the same max clique
is first-order definable.
Note that, since equivalence classes can be counted in $\STCC$ \cite[Lemma~II.7]{lau10},
the span of a vertex is \STCC-definable on the class of all interval graphs.

\begin{lemma}[\cite{lau10}, Lemma~IV.4, Corollary~IV.5]\label{cor:incompImpModule}
Suppose $M$ is a possible end of $G$
and $\mathcal C$ is a maximal set of $\prec_M$-incomparable max cliques. Then
\begin{itemize}[leftmargin=*]
\item $B\cap C= B\cap C'$ for all $C,C'\in  \mathcal C$,  $B\in  \mathcal M_G\setminus \mathcal C$,
 \item $S_{\mathcal C} := \bigcup_{C \in \mathcal C} C \setminus \bigcup_{B \in \mathcal M_G \setminus \mathcal C}B$ is a module of $G$, and
 \item $S_{\mathcal C} = \left\{ v \in \bigcup\mathcal C \;\big|\; \spn(v) \leq |\mathcal C| \right\}$.\qed
\end{itemize}
\end{lemma}

Finally, let $\sim^G_M$ be the equivalence relation on $V$ for which $x
\sim^G_M y$ if and only if $x=y$, or there is a maximal set $\mathcal C$ of
incomparable max cliques with respect to $\prec_M$ with $|\mathcal C|>1$
so that $x,y \in
S_{\mathcal C}$.
Let $G_M = G/_{\sim^G_M} := (V/_{\sim^G_M}, E_M)$,
where $E_M := \set{(u/_{\sim^G_M},v/_{\sim^G_M}) \mid (u,v) \in E}$.
It is easy to check that $\sim^G_M$ and the graph $G_M$ are $\STCC$-definable.

If $\mathcal C$ is a maximal set of $\prec_M$-incomparables in $G$ with $|\mathcal C|>1$,
then there is precisely one max clique $M_{\mathcal C}$ in $G_M$
which contains all the equivalence classes associated with $\mathcal C$, i.e.,
$M_{\mathcal C}=\{v/_{\sim^G_M} \mid v\in \bigcup\mathcal C\}$. We conclude:
\begin{lemma}\label{lem:GmodoutIntG}
$\prec_M$ induces a linear order on $G_M$'s max cliques.
In particular, $G_M$ is an interval graph. \qed
\end{lemma}

\subsection{Modules \texorpdfstring{$\mathcal W_G$}{W\_G} and the Graph \texorpdfstring{$L_G$}{L\_G}}\label{sec:DecompModulesLG}

We are now ready to give the definition of the set $\mathcal W_G$, which we
mentioned in Section~\ref{modular decomposition tree},
for connected interval graphs $G$ without an apex.
Furthermore, we show how to define graphs
that are isomorphic to the graph $L_G$ from Section~\ref{modular decomposition tree}
in $\STCC$.
In particular, this will enable us to prove, in Section~\ref{sec:L_G-can},
that an isomorphic copy of $L_G$ on the number sort is $\STCC$-definable.

Let $G = (V,E)$ be a connected interval graph without an apex.
Then $G$ contains more than one max clique.
Let $\mathfrak P_G$ be the set of all maximal proper subsets $\mathcal C$ of $\mathcal M_G$
with the property that for any $B\in\mathcal M_G \setminus \mathcal C$
we have $B\cap C = B\cap C'$ for all $C,C' \in \mathcal C$.
We must have $|\mathfrak P_G| \geq 3$ since $G$ is connected and no vertex may be
included in all max cliques of $G$.
Furthermore, if $\mathcal C, \mathcal C' \in \mathfrak P_G$ and $\mathcal C \neq \mathcal C'$,
then $\mathcal C \cap \mathcal C' = \emptyset$.
To see this, suppose that $D \in \mathcal C \cap \mathcal C'$.
Then $B\cap A = B\cap D = B\cap C$ for all $A,C \in \mathcal C\cup\mathcal C'$ and $B \not\in \mathcal C \cup \mathcal C'$.
So as $|{\mathfrak P_G}| \geq 3$, $\mathcal C \cup \mathcal C'$ is a proper subset of $\mathcal M_G$ satisfying the above property, which contradicts the maximality of $\mathcal C$ and $\mathcal C'$.
We conclude that ${\mathfrak P_G}$ is a partition of $\mathcal M_G$.

For each $\mathcal C\in {\mathfrak P_G}$ with $|\mathcal C|\geq2$ we define
$S_{\mathcal C} = \bigcup \mathcal C \,\setminus\, \bigcup (\mathcal M_G \setminus \mathcal C)$.
The correspondence in names to the modules $S_{ \mathcal C}$ as defined in
Lemma~\ref{cor:incompImpModule} is intended,
of course, and makes sense since the sets $ \mathcal C\in {\mathfrak P_G}$ enjoy the same interaction
properties with the rest of the graph as maximal sets of $\prec_M$-incomparable max cliques (cf.\
Lemma~\ref{cor:incompImpModule}).

We can now define the modules $\mathcal W_G$
mentioned in Section~\ref{modular decomposition tree}
for connected interval graphs $G$ without an apex.
We let $\mathcal S:=\{S_{ \mathcal C}\mid  \mathcal C\in {\mathfrak P_G}\text{ with }|\mathcal C|\geq2\}$, and define
\[
  \mathcal W_G:=\mathcal S\cup \bigcup_{v\in V\setminus \bigcup \mathcal S}\{\{v\}\}.
\]
From the fact that ${\mathfrak P_G}$ is a partition of $\mathcal M_G$,
we conclude that $\mathcal W_G$ forms a partition of $V$,
whereby inducing the equivalence relation $\sim_{G}$ on $V$.
In the following,
we call this equivalence relation alternatively $\sim_{{\mathfrak P_G}}$.

We are going to construct $\STCC$-definable graphs isomorphic to  $L_G$.
Let $Z_M$ be the max clique which is $\prec_M$-maximal in $G_M$. Now we forget
about $\prec_M$ and consider $\prec_{Z_M}$ on $G_M$. We write
\[
  L_M :=
  G_M/_{\sim_{Z_M}^{G_M}} =
  (V(G_M)/_{\sim_{Z_M}^{G_M}},E(G_M)/_{\sim_{Z_M}^{G_M}})
\]
with $E(G_M)/_{\sim_{Z_M}^{G_M}} =
\set{(u/_{\sim_{Z_M}^{G_M}},v/_{\sim_{Z_M}^{G_M}}) \mid (u,v) \in E(G_M)}$.
Lemma \ref{lem:GmodoutIntG} implies again that $\prec_{Z_M}$ induces a linear order on the max cliques of $L_M$.

\begin{lemma}\label{lem:doubleOrderingAbstraction}
Let $G$ be a connected interval graph that does not contain an apex,
and let $M_1,\ldots, M_k$ be its possible ends.
Then all of the graphs $L_{M_l}$, $l\in[k]$, are isomorphic to $L_G$ and we may partition $[k]$ into at most two sets $Q,Q'$ so that $(L_{M_i},\prec_{Z_{M_i}})$ and $(L_{M_j},\prec_{Z_{M_j}})$ are order isomorphic whenever $i,j \in Q$ or $i,j\in Q'$.
\end{lemma}

\proof

Equivalence relation $\sim_{ {\mathfrak P_G}}$ does the same as $\sim^G_M$, only that it is based on ${\mathfrak P_G}$ instead of the (finer) partition of max cliques induced by a strict weak ordering $\prec_M$.

Our goal is to show that each $L_{M}$ with $M\in \{M_1,\ldots,M_k\}$ is isomorphic to $G/_{\sim_{{\mathfrak P_G}}}$. For this it is enough to show that the concatenation of equivalence relation $\sim^G_{M}$ with $\sim^{G_{M}}_{Z_M}$ is equal to $\sim_{ {\mathfrak P_G}}$. Whenever $\mathcal C \in {\mathfrak P_G}$ and $M \not\in \mathcal C$, Lemma~\ref{lem:moduleImpIncomparable} implies that the max cliques in $\mathcal C$ are $\prec_M$-incomparable. As the sets in ${\mathfrak P_G}$ were chosen to be maximal, $\mathcal C$ is also a maximal set of $\prec_M$-incomparables (Lemma~\ref{cor:incompImpModule}).
It follows that $\sim_{ {\mathfrak P_G}}$ is equal to $\sim^G_M$ on $\bigcup_{M \not\in \mathcal C\in {\mathfrak P_G}} \mathcal C$.

When forming $G_M = G/_{\sim^G_M}$, each maximal set of $\prec_M$-incomparable max cliques $\mathcal C$ is replaced by the max clique $M_{\mathcal C} = \{v/_{\sim^G_M} \mid v \in \bigcup \mathcal C \}$. Note that this is also true when $\mathcal C$ consists of just one max clique. As a result, ${\mathfrak P_G}$ induces a partition $\mathfrak P_M$ of the max cliques of $G_M$. Also, if $\mathcal C_M$ is the cell of $\mathfrak P_M$ which contains $M$, then $\mathcal C_M$ is the only cell of $\mathfrak P_M$ which is possibly not a singleton. As $|\mathfrak P_M|\geq 3$, $Z_M \not\in \mathcal C_M$.

The final step is to show that $\sim_{ \mathfrak P_M}$ equals $\sim^{G_M}_{Z_M}$ on $G_M$. If $v/_{\sim^G_M}$ is a vertex of $G_M$ and $v/_{\sim^G_M}$ is an equivalence class of $\sim^G_{M}$ with  $|v/_{\sim^G_M}|>1$, then $v/_{\sim^G_M}$ is only contained in one max clique of $G_M$. Hence, $\mathfrak P_M$ inherits from ${\mathfrak P_G}$ the property that it partitions the max cliques $\mathcal M_{G_M}$ of $G_M$ into maximal sets $\mathcal C$ so that for any $B\in\mathcal M_{G_M} \setminus \mathcal C$ we have $B\cap C = B\cap C'$ for all $C,C' \in \mathcal C$. Arguing analogously as above, it follows that $\sim_{ \mathfrak P_M}$ equals $\sim^{G_M}_{Z_M}$. Therefore, $\sim_{ {\mathfrak P_G}}$ is equal to the concatenation of $\sim^G_{M}$ with $\sim^{G_{M}}_{Z_M}$ and $L_M$ is isomorphic to $L_G$. This proves the first part of the lemma.

To see the second part, observe that $\prec_{Z_M}$ induces a linear order on $L_M$'s max cliques. This is true for all $M\in \{M_1,\ldots, M_k\}$, so whenever $N$ is a possible end of $L_M$, then $\prec_N$ linearly orders the max cliques of $L_M$. Thus, $L_M$ has two possible ends which correspondingly induce two orders on the max cliques and vertices of $G/_{\sim_{ {\mathfrak P_G}}}$.\qed

\subsection{Canonising \texorpdfstring{$L_G$}{L\_G}}\label{sec:L_G-can}

Before showing how to use the
modular decomposition tree for canonising interval graphs $G = (V,E)$ in our logic,
let us take a look at how to define a canonical copy of $L_G$ in $\STCC$.

From the fact that $G$ is an interval graph,
it is not hard to see that $L_G$ is an interval graph, too.
Furthermore, notice that, if $A$ is a max clique of $G$,
then
\[
  A_{L_G} := \set{v/_{\sim_G} \mid v \in A}
\]
is a max clique of $L_G$,
and that all max cliques of $L_G$ are of this form.

\begin{lemma}\label{lem:LG-interval}\mbox{}
\begin{enumerate}[label=\emph{(\arabic*)},leftmargin=*]
\item\label{item1-lem:LG-interval}
  There are $\STCC$-formulae $\phi_{\sim}$, $\phi_{L}$
  such that for all interval graphs $G$,
  $\phi_{\sim}$ defines the equivalence relation $\sim_{G}$,
  and $\phi_{L}$ the edge relation of the graph $L_G$.
\item\label{item2-lem:LG-interval}
Let $G$ be a connected graph without any apices.
  If $L_G$ has $m>1$ max cliques,
  then there exist exactly two linear orderings of $L_G$'s max cliques,
  each the reverse of the other.
  There is an $\STCC$-formula
  that defines all pairs of tuples $(u,v),(u',v') \in V^2$
  such that $(u,v),(u',v')$ represent max cliques $A,M$ of $G$,
  $M$ is a possible end of $G$,
  and $A_{L_G}$ appears within the first $\lfloor \frac{m}{2}\rfloor$
  max cliques of $L_G$ with respect to $\prec_{Z_M}$.
\item\label{item3-lem:LG-interval}
  There is an  $\STCC$-formula that, for all interval graphs $G$,
  defines an isomorphic copy of $L_G$ on the number sort.
\end{enumerate}
\end{lemma}

\begin{proof}
Let us start by showing property~\ref{item1-lem:LG-interval}.
If $G$ is not connected or $G$
contains an apex, then $\sim_{G}$ is $\STC$-definable.
If $G$ is connected and does not contain an apex, then  for each possible end $M$ of $G$ the concatenation of equivalence relation $\sim^G_{M}$ with
$\sim^{G_{M}}_{Z_M}$ is equal to $\sim_{G}$ (Lemma~\ref{lem:doubleOrderingAbstraction}).
The $\STCC$-definability of equivalence relation  $\sim_{G}$ is a direct consequence of the \STCC-definability of the possible ends $M$ and
the equivalence relation $\sim^G_{M}$, Lemma~\ref{lem:GmodoutIntG}, which allows us to define max clique $Z_M$,
and the \STCC-definability of  $\sim^{G_M}_{Z_M}$.

We do not define the graph $L_G$ explicitly,
but rather implicitly within $G$.
That is, we do not single out a representative of each equivalence class
$v/_{\sim_G}$ of $\sim_G$,
but treat all vertices in $v/_{\sim_G}$ as representatives of $v/_{\sim_G}$.
Notice that, since all equivalence classes of $\sim_G$ are modules of $G$,
the edge relation of $L_G$ can be defined as the set of all edges of $G$
between vertices in different equivalence classes.

To show Property~\ref{item2-lem:LG-interval},
recall that by Lemma~\ref{lem:doubleOrderingAbstraction}
there are exactly two linear orderings of $L_G$'s max cliques,
each the reverse of the other.
By Property~\ref{item1-lem:LG-interval}, we can define $L_G$,
and for a possible end $M$ of $L_G$ we can define the linear order $\prec_{Z_M}$
(Lemma~\ref{lem:maxcliqueOrderSTCdefinable}).
Hence, given max clique $A$,
we can define $A_{L_G}$ and associate the linear order with $A$
where $A_{L_G}$ appears within the first $\lfloor \frac{m}{2}\rfloor$
max cliques of $L_G$.

Property~\ref{item3-lem:LG-interval}
is easy to see for graphs that are not connected or contain an apex.
For connected interval graphs $G$ that do not have any apices, Property~\ref{item3-lem:LG-interval}
follows directly
from Section~IV.B in \cite{lau10},
where the author shows that there is an \STCC-formula that
defines an ordered copy of $G$ on the number sort
if there is a max clique $M$ of $G$
such that $\prec_M$ is a linear order on $G$'s max cliques.
\end{proof}

According to the preceding lemma we can define an isomorphic copy of $L_G$
on the number sort.
In the following, we denote this copy by $\mathcal K(L_G)$.

\subsection{The Coloured Modular Decomposition Tree}
\label{sec:col-mod-dec-tree}

To obtain a complete invariant of an interval graph $G = (V,E)$,
we construct a refinement of the modular decomposition tree,
the coloured modular decomposition tree, in this section.

Let us consider the modular decomposition tree $T(G)$ of an interval graph $G$.
We call a module $W\in V(T(G))$ a \emph{decomposition module}
if $W=V$, or $|W|>1$ and $G[W^*]$ is a connected  graph, where $W^*$ is the parent of $W$ in $T(G)$.
All modules $W$ where $G[W^*]$ is not connected are called \emph{component modules}.
We let $\mathcal W^{dec}_G$ be the set of all decomposition modules
and $\mathcal W^{con}_G$ be the set of all component modules
occurring in the modular decomposition tree of $G$.

Let $P' := \set{(M,n) \mid M \in \mathcal M_G,\, n \in [\card{V}]}$.
Recall the definition of the \emph{span} of a vertex
from Section~\ref{subsec:orderingMaxCliques},
and that it is $\STCC$-definable.
For each $(M,n) \in P'$,
define $V_{M,n}$ as the set of vertices of the connected component of
$G[\set{v \in V \mid \spn(v) \leq n}]$ which intersects with $M$ (if non-empty),
and let $G_{M,n} := G[V_{M,n}]$.
Now let $P$ be the set of those $(M,n)\in P'$ for which the following properties are satisfied:
\begin{enumerate}[leftmargin=*]
\item\label{prop:Mn1} The number $n$ is maximal among those $n'$ with the property that $V_{M,n'}=V_{M,n}$.
\item\label{prop:Mn3} For all $m'>n$ where $V_{M,m'}$ is a module,
$V_{M,n}$ is a subset of an equivalence class of $\sim_{G_{M,m'}}$
with more than one vertex, or
  there exists a vertex $a \in V_{M,m'}\setminus V_{M,n}$ that is an apex of $G_{M,m'}$.
\end{enumerate}

\begin{lemma}\label{lem:ModDecompSTCC}
  If $(M,n) \in P\!$,
  then $V_{M,n}$ is a connected component of a decomposition module in $\mathcal W_G^{dec}$.
  Moreover, if $D$ is a connected component of a decomposition module
  in the modular decomposition tree of $G$,
  then there is an $(M,n) \in P$ with $V_{M,n} = D$.
\end{lemma}

\begin{proof}
Notice that for all modules $W$ of $G$ and all max cliques $C$ of $G$ with $C\cap W\not=\emptyset$
the set $W\cap C$ is a max clique of $G[W]$, and every max clique of $G[W]$ is of that form.
Further, an easy induction shows
	that for all modules $W\in \mathcal W^{dec}_G\cup \mathcal W^{con}_G$ the following properties are satisfied:
	\begin{enumerate}[label=(\Alph*),ref=\Alph*,leftmargin=*]
		\item\label{enum:modulepropA}
			Let $C,C'\in \mathcal M_G$ be max cliques of $G$ with  $C\not = C'$ where
			$C\cap W\not=\emptyset$ and $C'\cap W\not=\emptyset$.
			Then for max cliques  $C\cap W\!$, $C'\cap W$ of $G[W]$ we have $C\cap W\not=C'\cap W\!$.
		\item\label{enum:modulepropB}
			Let $\mathcal C:=\{C\in \mathcal M_G\mid C\cap W\not=\emptyset\}$.
			Then for all $B\in \mathcal M_G\setminus \mathcal C$ and all $C,C'\in \mathcal C$
			we have $B\cap C= B\cap C'\!$.
		\item\label{enum:modulepropC}
			For the set $\mathcal C$ from (\ref{enum:modulepropB}),
			$W=\bigcup_{C\in \mathcal C} V_{C,c}$ where $c:=|\mathcal C|$ if $W$ contains an apex
			and $c:=|\mathcal C|-1$ if $W$ has no apices,
			and for each $C\in \mathcal C$ the set $V_{C,c}$ is a connected component of $W\!$.
	\end{enumerate}

  In order to show Lemma~\ref{lem:ModDecompSTCC},
  we also need the following properties:

  \begin{sclaim}\label{clm:ModDecompSTCC}
    If $V_{M,k}$ is a connected component of a decomposition module $W$ of $G$,
    and $V_{M,k}\subsetneq V_{M,l}$ for an $l>k$,
    then $W \subsetneq V_{M,l}$.
  \end{sclaim}

  \begin{proofofclaim}
  Let $k'$ be the maximum span of a vertex in $W\!$.
  Since $V_{M,k}\in \mathcal W^{dec}_G\cup \mathcal W^{con}_G$, we have
	$V_{M,k} = V_{M,k'}$ as a direct consequence of Property~\ref{enum:modulepropC}.
	Thus, 	we can assume that $k\geq k'\!$.
    Further,
    we have $V_{M,l}\not\subseteq W$ as  $V_{M,l}\subseteq W$ leads to a contradiction, because
	$V_{M,l}$ is connected and $V_{M,k}\subsetneq V_{M,l}$ is a connected component of $W\!$.
	Thus, $V_{M,l}\setminus W$ is non-empty.
	Since $V_{M,l}$ is connected and $V_{M,k}\subseteq V_{M,l}\cap W\!$,
	there must exist a vertex $v\in V_{M,l}\setminus W$ that is adjacent to a vertex in the non-empty set $V_{M,l}\cap W\!$.
	As $W$ is a module, $v$ is adjacent to all vertices in $W\!$.
	Therefore, $W\subseteq V_{M,l}$, because $\spn(w)\leq k'\leq l$ for all vertices $w\in W\!$.
    \varqed
  \end{proofofclaim}
  \begin{sclaim}\label{clm:ModDecompCliq}
    Let $(M,d)\in P'$ and $V_{M,d}$ be a module in $\mathcal W^{dec}_G\cup \mathcal W^{con}_G$.
    If $V_{M,d}$ is a clique, then there exists only one max clique $C\in \mathcal M_G$
    with $C\cap V_{M,d}\not=\emptyset$.
  \end{sclaim}

  \begin{proofofclaim}
	Since $V_{M,d}$ is a clique, there must exist a max clique $B\in \mathcal M_G$ with $V_{M,d}\subseteq B$.
	Let us assume, there exists a max clique $B'\in \mathcal M_G$ different from $B$ with $B'\cap V_{M,d}\not=\emptyset$.
	According to Property~\ref{enum:modulepropA} we have  $B\cap V_{M,d}\not= B'\cap V_{M,d}$ and therefore
	$B'\cap V_{M,d}\subsetneq V_{M,d}$.
	Since  $V_{M,d}$ is a module, $B'\cup V_{M,d}$ is a clique, a contradiction to $B'$ being a max clique.
  	 \varqed
  \end{proofofclaim}
\begin{sclaim}\label{clm:ModDecompApices}
    Let $(M,d)\in P'$ and $V_{M,d}$ be a module in $\mathcal W^{dec}_G\cup \mathcal W^{con}_G$.
    Further, let $0<d'<d$ be such that $V_{M,d'}\subsetneq V_{M,d}$, and let $A\not=\emptyset$
    be the set of apices of $G_{M,d}$.
    Then $V_{M,d'}\subseteq V_{M,d}\setminus A$.
\end{sclaim}
  \begin{proofofclaim}
	If $V_{M,d}$ is a clique, then according to Claim~\ref{clm:ModDecompCliq} max clique $M$ is the only
	max clique in $\mathcal M_G$ with $M\cap V_{M,d}\not=\emptyset$.
	Thus, $V_{M,d}=V_{M,1}$ and there does not exist a $d'$ with $0<d'<d$ such that $V_{M,d'}\subsetneq V_{M,d}$.

Now let $V_{M,d}$ be not a clique.
Further, let $\mathcal C$ be the set of max cliques $C\in \mathcal M_G$ with $C\cap V_{M,d}\not=\emptyset$
and $c:=|\mathcal C|$.
In the following we show that $a\in V_{M,d}$ is an apex of $G_{M,d}$ if and only if $\spn(a)=c$.
If  $a\in V_{M,d}$ and $\spn(a)=c$, then $a$ is contained in every max clique of $G$
that has a non-empty intersection with  $V_{M,d}$.
As every vertex in $V_{M,d}$ is contained in at least one max clique of $G$,
which of course has a non-empty intersection with  $V_{M,d}$, $a$ is an apex of $G_{M,d}$.
Now let $a$ be an apex of $G_{M,d}$ and let us assume that there exists a max clique $C\in \mathcal M_G$
with $C\cap V_{M,d}\not=\emptyset$ and $a\not\in C$.
Apex $a$ is adjacent to all vertices in $C\cap V_{M,d}$, and since $V_{M,d}$ is a module, $a$ is also
adjacent to all vertices in $C\setminus V_{M,d}$. Therefore. $C\cup \{a\}$ is a clique, which is a contradiction to
$C$ being a maximal clique of $G$.

From $\spn(v)=c$ for all vertices $v \in A$, $\spn(v)<c$ for all $v \in  V_{M,d}\setminus A$  and $V_{M,d'}\subsetneq V_{M,d}$ it follows that $d'<c$.
Consequently, $V_{M,d'}\subseteq V_{M,d}\setminus A$.
  	 \varqed
  \end{proofofclaim}

To proceed with the proof of Lemma~\ref{lem:ModDecompSTCC},
we first show that if $D$ is a connected component of a decomposition module $W\in \mathcal W_G^{dec}$
and $M\in \mathcal M_G$ with $M\cap D\not=\emptyset$,
then there is an $n\in \nat$ such that $(M,n) \in P$ and $V_{M,n} = D$.

We proof this by induction on the depth of the modular decomposition tree:
Clearly, if $D$ is a connected component of decomposition module $V$
(i.e., a connected component of $G$),
then $D=V_{M,|V|}$
for a max clique $M$ with $M\cap D\not= \emptyset$, and $(M,|V|)\in P\!$.

Now, let $D$ be a component of module $W\in \mathcal W_G^{dec}$  with $W\not= V\!$.
Let $c$ be the number $c'$ of max cliques of $G$ intersecting with $W$
if $W$ contains an apex and $c'-1$ if $W$ has no apices.
According to Property~\ref{enum:modulepropC}, $V_{M,c}=D$.
Let $n$ be maximal with $V_{M,n}=V_{M,c}$.
Then $(M,n)\in P'$ and $D= V_{M,n}$.
Choosing $(M,n)$ like that ensures that Property~\ref{prop:Mn1} is satisfied for $(M,n)$.

It remains to show Property~\ref{prop:Mn3}.
Let $m'> n$ and let $V_{M,m'}$ be a module.
According to Property~\ref{prop:Mn1} we have   $V_{M,n}\subsetneq V_{M,m'}$.
Thus, Claim~\ref{clm:ModDecompSTCC} implies $W\subsetneq V_{M,m'}$.

First, let us assume there exists an apex $a$ of $G_{M,m'}$.
If there exists an apex of $G_{M,m'}$ in $V_{M,m'}\setminus W\!$, we have shown Property~\ref{prop:Mn3}.
Thus, let us assume all apices of $G_{M,m'}$, in particular $a$, are in $W\!$.
Since $W$ is a module and $a\in W$, the vertex sets $V_{M,m'}\setminus W$ and $W$ must be completely connected.
If $W$ contains two vertices $w,w'$ that are not adjacent,
then in the minimal interval representation the interval of each vertex in $V_{M,m'}\setminus W$
has to intersect the intervals of both $w$ and $w'$.
Thus, the intervals of all vertices in $V_{M,m'}\setminus W$ intersect with each other
and each vertex in $V_{M,m'}\setminus W$ is an apex, a contradiction.
Let us assume $W$ is a clique.
Let $W^*$ be the parent module of $W$ in the modular decomposition tree of $G$.
Since $W$ is a decomposition module, $|W|>1$ and
$W^*$ contains either an apex, or is connected and contains no apices.
$W^*$ cannot contain an apex, because then all vertices in $W^*$ form a clique and $W$ is not in $\mathcal W_{G[W^*]}$.
If $W^*$ is connected and contains no apices,
then $W=S_{\mathcal C}$ for $\mathcal C\in \mathfrak P_{G[W^*]}$
where $\mathcal C$ is a set of
max cliques of $G[W^*]$
with  $|\mathcal C|\geq 2$ (see Section \ref{sec:DecompModulesLG}).
As $W$ is connected, $W=V_{M,n}$.
According to Claim~\ref{clm:ModDecompCliq} there exists only one max clique $C$ of $G$ with $C\cap W\not=\emptyset$.
Consequently, $C':=C\cap W^*$ is the only max clique in $G[W^*]$ with $C'\cap W\not=\emptyset$, a contradiction.
Hence, $W$ cannot be a clique.

Now let us assume that there does not exist an apex of $G_{M,m'}$.
Thus, $\sim_{G_{M,m'}}$ is constructed as described in Section \ref{sec:DecompModulesLG}.
Let $W'$ be the parent module $W^*$ of $W$ in the modular decomposition tree of $G$
if $W^*$ is a decomposition module,
or if $W^*$ is a component module, let $W'$ be the parent of module $W^*\!$.
Then $W'$ is a decomposition module.
Further, let $D'$ be the component of $W'$ that contains $W\!$.
Notice that no matter what set we chose for $W'\!$, we have $D'=W^*\!$.
According to Property~\ref{enum:modulepropC}, there exists an $n'\in [|V|]$ such that $D'=V_{M,n'}$.
Let $n'$ be maximal with that property.
Therefore, $W^*=V_{M,n'}$ and $W^*$ is a component of a decomposition module.
By inductive assumption we have $(M,n')\in P\!$.
If $V_{M,n'}=V_{M,m'}$, then $V_{M,n}$ is a subset of equivalence class $W$ of $\sim_{G_{M,m'}}$
with more than one vertex and we are done.
Therefore, let us assume $V_{M,n'}\not =V_{M,m'}$.
If $n'<m'\!$, then $V_{M,n}\subseteq W\subsetneq W^*=V_{M,n'}\subsetneq V_{M,m'}$.
As $(M,n')$ satisfies Property~\ref{prop:Mn3} and there does not exist an apex in $G_{M,m'}$,
the set $V_{M,n'}$, and therefore also the set $V_{M,n}\subsetneq V_{M,n'}$, is a subset of an equivalence class of
$\sim_{G_{M,m'}}$ with more than one vertex.

It remains to consider $m'<n'$ where $V_{M,n'}\not =V_{M,m'}$.
Then $V_{M,n}\subseteq W\subsetneq V_{M,m'}\subsetneq V_{M,n'}=W^*\!$.
If $W^*=V_{M,n'}$ contains an apex,
then $W=V_{M,n'}\setminus A$ where $A$ is the set of apices of $G_{M,n'}$.
According to Claim~\ref{clm:ModDecompApices},
$V_{M,m'}\subseteq V_{M,n'}\setminus A$.
But this implies $V_{M,m'} \subseteq W$,
a contradiction.

Finally, let us assume $W^*=V_{M,n'}$ is connected and does not contain an apex.
Then $W=S_{\mathcal C}$ for a $\mathcal C\in \mathfrak P_{G[W^*]}$ with $|\mathcal C|\geq 2$ where
$\mathfrak P_{G[W^*]}$ is the set of all maximal proper subsets $\mathcal C'$ of $\mathcal M_{G[W^*]}$,
the set of max cliques of $G[W^*]$, with the property that
for any $B\in \mathcal M_{G[W^*]}\setminus \mathcal C'$ we have $C\cap B=C'\cap B$ for all $C,C'\in  \mathcal C'$.
For all $C \in \mathcal M_{G[W^*]}$ with $C\cap V_{M,m'}\not=\emptyset$
let $f(C)$ be the set $C\cap V_{M,m'}$.
As $V_{M,m'}$ is a module, the set $\{f(C)\mid C\in \mathcal M_{G[W^*]}, C\cap V_{M,m'}\not=\emptyset\}$
is the set $\mathcal M_{G_{M,m'}}$ of max cliques of $G_{M,m'}$.
Let $f(\mathcal C)$ be the set $\{f(C)\mid C\in \mathcal C\}$.
Then $f(\mathcal C)$ is exactly the set of max cliques of  $G_{M,m'}$ that have a non-empty intersection with $W$.
Let $f(C),f(C')\in f(\mathcal C)$ and $f(B)\in \mathcal M_{G_{M,m'}}\setminus f(\mathcal C)$.
Then $f(C)\cap f(B)=f(C')\cap f(B)$, because $C\cap B=C'\cap B$ and therefore
$(C\cap V_{M,m'})\cap (B\cap V_{M,m'})=(C'\cap V_{M,m'})\cap (B\cap V_{M,m'})$.
Further, $|f(\mathcal C)|>1$, since $|\mathcal C|>1$ and for $C,C'\in \mathcal C\subseteq \mathcal M_{G[W^*]}$
with $C\not= C'$ we have $C\cap W\not= C'\cap W$ according to Property~\ref{enum:modulepropA}.
Consequently, $(C\cap V_{M,m'})\cap W\not= (C'\cap V_{M,m'})\cap W$ and
$f(C)\not= f(C')$ for max cliques  $f(C),f(C')\in f(\mathcal C)$.
We obtain that there exists a subset $f(\mathcal C')$ of  $\mathcal M_{G_{M,m'}}$ with
$f(\mathcal C)\subseteq f(\mathcal C')$ such that $f(\mathcal C')\in \mathfrak P_{G_{M,m'}}$.
As there exists no max clique $f(B)\in \mathcal M_{G_{M,m'}}\setminus f(\mathcal C')$ with $f(B)\cap W\not=\emptyset$,
$W\subseteq S_{f(\mathcal C')}$ and we have shown that $V_{M,n}$ is a subset of equivalence class
$S_{f(\mathcal C')}$ of $\sim_{G_{M,m'}}$ with more than one vertex.

\bigskip

For the other direction,
let $(M,n)\in P$, we need to show that
$V_{M,n}$ is a component of a decomposition module.
We prove this by induction on $n$.
Clearly, this holds for $n=|V(G)|$, so let $n<|V(G)|$.
Let $p$ be minimal such that $p>n$ and  $(M,p)\in P$.
Since $(M,|V|)\in P$ such a number exists.
By inductive assumption we know that  $V_{M,p}$ is a component of a decomposition module.
Thus,  $V_{M,p}$ is a module occurring in  $V(T(G))$, the vertices of the modular decomposition tree of $G$.

Since $(M,n)\in P$, $(M,n)$ satisfies Property~\ref{prop:Mn3}.
Thus, $V_{M,n}$ is a subset of an equivalence class of $\sim_{G_{M,p}}$ with more than one vertex
or there exists an apex of $G_{M,p}$ in $V_{M,p}\setminus V_{M,n}$.

Let $V_{M,n}$ be a subset of an equivalence class $W$ of $\sim_{G_{M,p}}$ with more than one vertex.
As $V_{M,p}$ is connected, the equivalence class $W$ is a decomposition module.
Let $D$ be the connected component of $W$ that contains $V_{M,n}$.
If $V_{M,n}=D$, then $V_{M,n}$ is a component of a decomposition module and we are done.
If $V_{M,n}$ is a proper subset of $D$, we obtain a contradiction to the choice of $p$,
since we have already shown that for component $D$ of decomposition module $W$
there must exist an $m \in [|V|]$ such that $(M,m)\in P$ and $V_{M,m}=D$, and $n<m<p$.

Now let there be a vertex $a\in V_{M,p}\setminus V_{M,n}$ that is an apex of $G_{M,p}$.
Let $A$ be the set of apices of $G_{M,p}$.
According to Claim~\ref{clm:ModDecompApices} we have $V_{M,n}\subseteq V_{M,p}\setminus A$.
Further, $|V_{M,p}\setminus A|=1$ implies that $v\in V_{M,p}\setminus A$ is also an apex. Consequently,
 $|V_{M,p}\setminus A|>1$.
Therefore, we have either shown that $V_{M,n}$ is a component of equivalence class $V_{M,p}\setminus A$ of
$\sim_{G_{M,p}}$ with more than one vertex or obtain a contradiction to the choice of $p$.
\end{proof}

\begin{corollary}
  There is an $\STCC$-formula $\phi(x,y,z)$
  such that
  for all interval graphs $G = (V,E)$, all $v,w \in V$,
  and all $n \in [\card{V}]$,
  we have $G \models \phi[v,w,n]$ iff
  $M = \ngh(v) \cap \ngh(w)$ is a max clique of $G$ and $(M,n) \in P$.
\end{corollary}

We are now ready to define the coloured modular decomposition tree.
An illustration of the tree can be found in Figure~\ref{fig:modularTree}.

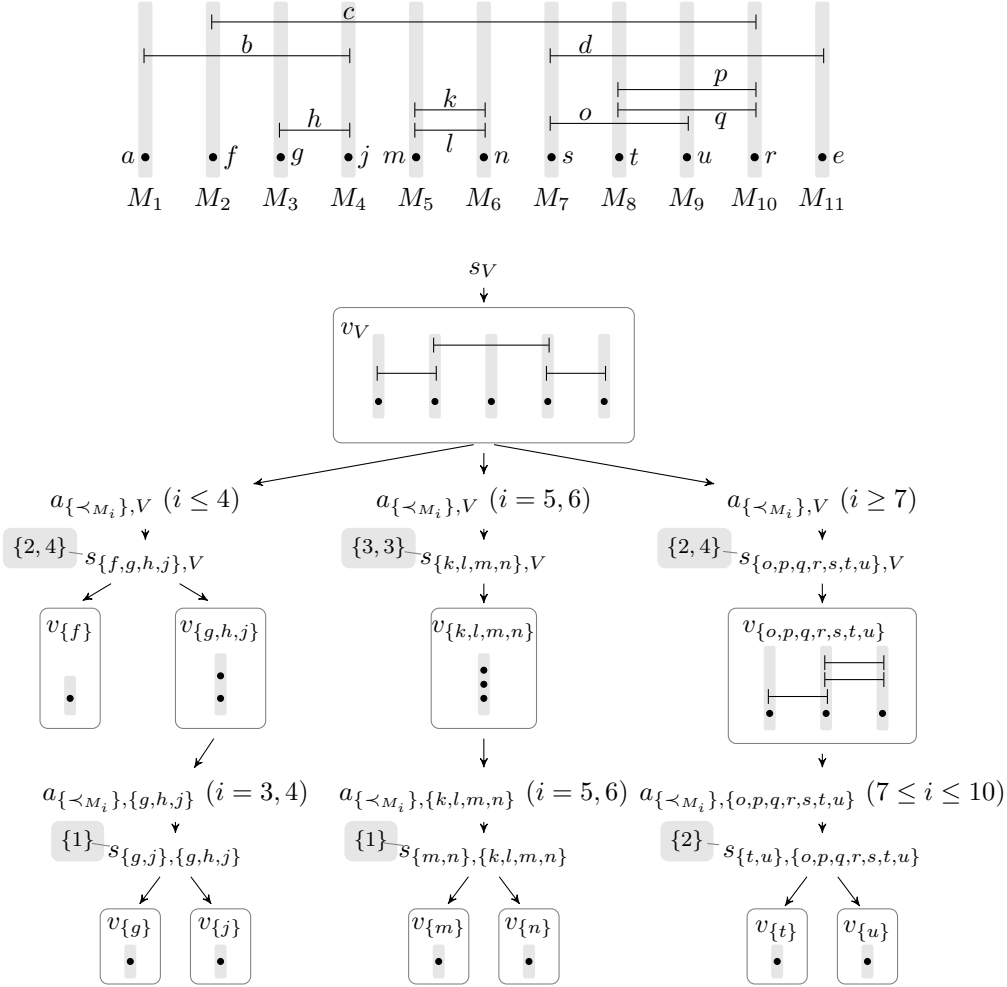
\begin{figure}
\centering
\begin{tikzpicture}[intvl/.style={|-|,shorten >=-.7pt,shorten <=-.7pt},
                    v/.style={circle,inner sep=0pt,minimum size=4pt},
        			e/.style={->, shorten >=1pt}]

\begin{scope}[xshift=-4.5cm,scale=.9]

  \foreach \x in {1,...,11}
 {
  \fill[xshift=\x cm-1cm,rounded corners=1pt,fill=gray!20!white] (-3pt,-.3cm) rectangle (3pt,2.3cm);
  \draw[xshift=\x cm-1cm] (0,-.3cm) node[anchor=north] {\small $M_{\x}$};
 }

\filldraw (0,0) node[anchor=east] {\small $a$} circle (1.5pt)
    (10,0) node[anchor=west] {\small $e$} circle (1.5pt)
    (1,0) node[anchor=west] {\small $f$} circle (1.5pt)
    (2,0) node[anchor=west] {\small $g$} circle (1.5pt)
    (3,0) node[anchor=west] {\small $j$} circle (1.5pt)
    (4,0) node[anchor=east] {\small $m$} circle (1.5pt)
    (5,0) node[anchor=west] {\small $n$} circle (1.5pt)
    (6,0) node[anchor=west] {\small $s$} circle (1.5pt)
    (7,0) node[anchor=west] {\small $t$} circle (1.5pt)
    (8,0) node[anchor=west] {\small $u$} circle (1.5pt)
    (9,0) node[anchor=west] {\small $r$} circle (1.5pt);

\draw[intvl] (0,1.5) -- (3,1.5) node[midway,above,inner sep=1pt] {\small $b$};
\draw[intvl] (1,2) -- (9,2) node[near start,above,inner sep=1pt] {\small $c$};
\draw[intvl] (6,1.5) -- (10,1.5) node[very near start,above,inner sep=1pt] {\small $d$};
\draw[intvl] (2,.4) -- (3,.4) node[midway,above,inner sep=1pt] {\small $h$};
\draw[intvl] (4,.4) -- (5,.4) node[midway,below,inner sep=1pt] {\small $l$};
\draw[intvl] (4,.7) -- (5,.7) node[midway,above,inner sep=1pt] {\small $k$};
\draw[intvl] (6,.5) -- (8,.5) node[near start,above,inner sep=1pt] {\small $o$};
\draw[intvl] (7,.7) -- (9,.7) node[near end,below,inner sep=1pt] {\small $q$};
\draw[intvl] (7,1) -- (9,1) node[near end,above,inner sep=1pt] {\small $p$};

\end{scope}

\begin{scope}[yshift=-3cm]

\node (root) at (0,1.5) {\small $s_V$};
\node (vVentry) at (0,1) {};
\node (vV) at (0,-.8) {};
\node (vVlabel) at (-1.7,.7) {\small $v_V$};
\draw[draw=gray,rounded corners = 3pt] (-2,-.8) rectangle (2,1);
\node (lM16) at (-4.5,-1.6) {\small $a_{{\{\prec_{M_i}\},V}}\ {\footnotesize (i\leq 4)}$};
\node (lM56) at (0,-1.6) {\small $a_{{\{\prec_{M_i}\},V}}\ (i=5,6)$};
\node (lM66) at (4.5,-1.6) {\small $a_{{\{\prec_{M_i}\},V}}\ (i\geq 7)$};

\draw[->,shorten >=-2pt] (root) -- (vVentry);
\draw[e] (vV) -- (lM16);
\draw[e] (vV) -- (lM56);
\draw[e] (vV) -- (lM66);

\node (sfghi) at (-4.5,-2.4) {\small $s_{\{f,g,h,j\},V}$};
\node (vf) at (-5.5,-3) {}; 
\node (vghi) at (-3.5,-3) {}; 
\node (vfLabel) at (-5.5,-3.3) {\small $v_{\{f\}}$};
\node (vghiLabel) at (-3.5,-3.3) {\small $v_{\{g,h,j\}}$};
\draw[rounded corners=3pt,draw=gray] (-5.9,-4.6) rectangle (-5.1,-3);
\draw[rounded corners=3pt,draw=gray] (-4.1,-4.6) rectangle (-2.9,-3);

\draw[e] (lM16) -- (sfghi);
\draw[e] (sfghi) -- (vf);
\draw[e] (sfghi) -- (vghi);

\node (vghiOut) at (-3.5,-4.6) {};
\node (aghi) at (-4.1,-5.5) {\small $a_{\{\prec_{M_i  }\},\{g,h,j \} }\ (i=3,4)$};
\node (sgi) at (-4.1,-6.3) {\small $s_{\{g,j\},\{g,h,j \}}$};
\node (vg) at (-4.7,-7.1) {};
\node (vi) at (-3.5,-7.1) {};
\draw[rounded corners=3pt,draw=gray] (-5.1,-8) rectangle (-4.3,-7);
\draw[rounded corners=3pt,draw=gray] (-3.9,-8) rectangle (-3.1,-7);
\draw[e] (vghiOut) -- (aghi);
\draw[e] (aghi) -- (sgi);
\draw[e] (sgi) -- (vg);
\draw[e] (sgi) -- (vi);

\node (sjkl) at (0,-2.4) {\small $s_{\{k,l,m,n\},V}$};
\node (vjkl) at (0,-3) {}; 
\node (vjklLabel) at (0,-3.3) {\small $v_{\{k,l,m,n\}}$};
\draw[rounded corners=3pt,draw=gray] (-.7,-4.6) rectangle (.7,-3);
\draw[e] (lM56) -- (sjkl);
\draw[->,shorten >=-2pt] (sjkl) -- (vjkl);

\node (vjklOut) at (0,-4.6) {};
\node (ajklm) at (0,-5.5) {\small $a_{\{\prec_{M_i}\}, \{k,l,m,n \}  }\ (i=5,6)$};
\node (sjm) at (0,-6.3) {\small $s_{\{m,n\},\{k,l,m,n\}}$};
\node (vj) at (-.6,-7.1) {};
\node (vm) at (.6,-7.1) {};
\draw[rounded corners=3pt,draw=gray] (-1,-8) rectangle (-.2,-7);
\draw[rounded corners=3pt,draw=gray] (.2,-8) rectangle (1,-7);
\draw[e] (vjklOut) -- (ajklm);
\draw[e] (ajklm) -- (sjm);
\draw[e] (sjm) -- (vj);
\draw[e] (sjm) -- (vm);

\node (smnopqrs) at (4.5,-2.4) {\small $s_{\{o,p,q,r,s,t,u\},V}$};
\node (vmnopqrs) at (4.5,-3) {};
\node (vmnopqrsExit) at (4.5,-4.8) {};
\node at (4.4,-3.3) {\small $v_{\{o,p,q,r,s,t,u\}}$};
\draw[rounded corners=3pt,draw=gray] (3.25,-4.8) rectangle (5.75,-3);
\node (amnopqrs) at (4.5,-5.5) {\small $a_{\{\prec_{M_i}\},\{o,p,q,r,s,t,u\}}\ (7\leq i\leq 10)$};
\node (srs) at (4.5,-6.3) {\small $s_{\{t,u\},\{o,p,q,r,s,t,u\}}$};
\node (vr) at (3.9,-7.1) {}; 
\node (vs) at (5.1,-7.1) {}; 
\node (vrLabel) at (3.9,-7.3) {\small $v_{\{t\}}$};
\node (vsLabel) at (5.1,-7.3) {\small $v_{\{u\}}$};
\draw[rounded corners=3pt,draw=gray] (3.5,-8) rectangle (4.3,-7);
\draw[rounded corners=3pt,draw=gray] (4.7,-8) rectangle (5.5,-7);

\draw[e] (lM66) -- (smnopqrs);
\draw[->,shorten >=-2pt] (smnopqrs) -- (vmnopqrs);
\draw[e] (vmnopqrsExit) -- (amnopqrs);
\draw[e] (amnopqrs) -- (srs);
\draw[->] (srs) -- (vr);
\draw[->] (srs) -- (vs);

\node[fill=gray!20!white,rounded corners=3pt] (sfghiColor) at (-5.9,-2.2) {\scriptsize $\{2,4\}$};
\draw[gray,shorten >=-2pt,shorten <=-2pt] (sfghi) -- (sfghiColor);
\node[fill=gray!20!white,rounded corners=3pt] (sjklColor) at (-1.4,-2.2) {\scriptsize $\{3,3\}$};
\draw[gray,shorten >=-2pt,shorten <=-2pt] (sjkl) -- (sjklColor);
\node[fill=gray!20!white,rounded corners=3pt] (smnopqrsColor) at (2.8,-2.2) {\scriptsize $\{2,4\}$};
\draw[gray,shorten >=-2pt,shorten <=-2pt] (smnopqrs) -- (smnopqrsColor);
\node[fill=gray!20!white,rounded corners=3pt] (srsColor) at (2.7,-6.1) {\scriptsize $\{2\}$};
\draw[gray,shorten >=-2pt,shorten <=-2pt] (srs) -- (srsColor);

\node[fill=gray!20!white,rounded corners=3pt] (sgiColor) at (-5.4,-6.1) {\scriptsize $\{1\}$};
\draw[gray,shorten >=-2pt,shorten <=-2pt] (sgi) -- (sgiColor);
\node[fill=gray!20!white,rounded corners=3pt] (sjmColor) at (-1.5,-6.1) {\scriptsize $\{1\}$};
\draw[gray,shorten >=-2pt,shorten <=-2pt] (sjm) -- (sjmColor);

\begin{scope}[xshift=3.8cm,yshift=-4.4cm,scale=.75]


\foreach \x in {1,...,3}
 {
  \fill[xshift=\x cm-1cm,rounded corners=1pt,fill=gray!20!white] (-3pt,-.3cm) rectangle (3pt,1.2cm);
 }

\filldraw (0,0) node[anchor=east] {} circle (1.5pt)
            (1,0) node[anchor=east] {} circle (1.5pt)
            (2,0) node[anchor=east] {} circle (1.5pt);

\draw[intvl] (0,.3) -- (1,.3) node[midway,above,inner sep=1pt] {};
\draw[intvl] (1,.6) -- (2,.6) node[midway,above,inner sep=1pt] {};
\draw[intvl] (1,.9) -- (2,.9) node[midway,above,inner sep=1pt] {};

\end{scope}

\begin{scope}[xshift=0cm,yshift=-4.2cm,scale=.75]


\foreach \x in {1,...,1}
 {
  \fill[xshift=\x cm-1cm,rounded corners=1pt,fill=gray!20!white] (-3pt,-.3cm) rectangle (3pt,.8cm);
 }

\filldraw (0,0) node[anchor=east] {} circle (1.5pt)
            (0,.25) node[anchor=east] {} circle (1.5pt)
		(0,.5) node[anchor=east] {} circle (1.5pt);

\end{scope}

\begin{scope}[xshift=-3.5cm,yshift=-4.2cm,scale=.75]


\foreach \x in {1,...,1}
 {
  \fill[xshift=\x cm-1cm,rounded corners=1pt,fill=gray!20!white] (-3pt,-.3cm) rectangle (3pt,.8cm);
 }

\filldraw (0,0) node[anchor=east] {} circle (1.5pt)
            (0,.4) node[anchor=east] {} circle (1.5pt);

\end{scope}

\begin{scope}[xshift=-5.5cm,yshift=-4.2cm,scale=.75]


\foreach \x in {1,...,1}
 {
  \fill[xshift=\x cm-1cm,rounded corners=1pt,fill=gray!20!white] (-3pt,-.3cm) rectangle (3pt,.4cm);
 }
\filldraw (0,0) node[anchor=east] {} circle (1.5pt);

\end{scope}

\begin{scope}[xshift=3.9cm,yshift=-7.7cm,scale=.75]


\foreach \x in {1,...,1}
 {
  \fill[xshift=\x cm-1cm,rounded corners=1pt,fill=gray!20!white] (-3pt,-.3cm) rectangle (3pt,.3cm);
 }
\filldraw (0,0) node[anchor=east] {} circle (1.5pt);

\end{scope}

\begin{scope}[xshift=5.1cm,yshift=-7.7cm,scale=.75]


\foreach \x in {1,...,1}
 {
  \fill[xshift=\x cm-1cm,rounded corners=1pt,fill=gray!20!white] (-3pt,-.3cm) rectangle (3pt,.3cm);
 }
\filldraw (0,0) node[anchor=east] {} circle (1.5pt);

\end{scope}

\begin{scope}[xshift=-4.7cm,yshift=-7.7cm,scale=.75]


\node at (0,.6) {\small $v_{\{g\}}$};

\foreach \x in {1,...,1}
 {
  \fill[xshift=\x cm-1cm,rounded corners=1pt,fill=gray!20!white] (-3pt,-.3cm) rectangle (3pt,.3cm);
 }
\filldraw (0,0) node[anchor=east] {} circle (1.5pt);

\end{scope}

\begin{scope}[xshift=-3.5cm,yshift=-7.7cm,scale=.75]


\node at (0,.6) {\small $v_{\{j\}}$};

\foreach \x in {1,...,1}
 {
  \fill[xshift=\x cm-1cm,rounded corners=1pt,fill=gray!20!white] (-3pt,-.3cm) rectangle (3pt,.3cm);
 }
\filldraw (0,0) node[anchor=east] {} circle (1.5pt);

\end{scope}

\begin{scope}[xshift=-.6cm,yshift=-7.7cm,scale=.75]


\node at (0,.6) {\small $v_{\{m\}}$};

\foreach \x in {1,...,1}
 {
  \fill[xshift=\x cm-1cm,rounded corners=1pt,fill=gray!20!white] (-3pt,-.3cm) rectangle (3pt,.3cm);
 }
\filldraw (0,0) node[anchor=east] {} circle (1.5pt);

\end{scope}

\begin{scope}[xshift=.6cm,yshift=-7.7cm,scale=.75]


\node at (0,.6) {\small $v_{\{n\}}$};

\foreach \x in {1,...,1}
 {
  \fill[xshift=\x cm-1cm,rounded corners=1pt,fill=gray!20!white] (-3pt,-.3cm) rectangle (3pt,.3cm);
 }
\filldraw (0,0) node[anchor=east] {} circle (1.5pt);

\end{scope}

\begin{scope}[xshift=-1.4cm,yshift=-.25cm,scale=.75]


\foreach \x in {1,...,5}
 {
  \fill[xshift=\x cm-1cm,rounded corners=1pt,fill=gray!20!white] (-3pt,-.3cm) rectangle (3pt,1.2cm);
 }

\filldraw (0,0) node[anchor=east] {} circle (1.5pt)
            (1,0) node[anchor=east] {} circle (1.5pt)
            (2,0) node[anchor=east] {} circle (1.5pt)
            (3,0) node[anchor=east] {} circle (1.5pt)
            (4,0) node[anchor=east] {} circle (1.5pt);

\draw[intvl] (0,.5) -- (1,.5) node[midway,above,inner sep=1pt] {};
\draw[intvl] (1,1) -- (3,1) node[midway,above,inner sep=1pt] {};
\draw[intvl] (3,.5) -- (4,.5) node[midway,above,inner sep=1pt] {};

\end{scope}

\end{scope}

\end{tikzpicture}
\caption{An interval graph and its coloured modular decomposition tree. Component vertices $v_U$ are represented together with the interval graph $L_U$ labeling them. The colours of module vertices are indicated in the gray fields next to them.}\label{fig:modularTree}
\end{figure}
Formally, the coloured modular decomposition tree is defined as
$\mathcal T = \mathcal T_G = (V_{\mathcal T},E_{\mathcal T})$,
where the set $V_{\mathcal T}$ of nodes and the set $E_{\mathcal T}$ of edges
of $\mathcal T$ is defined as follows.
$V_{\mathcal T}$ is the union of the following sets:
\begin{itemize}[leftmargin=*]
\item the set $\mathcal V$ of \emph{component vertices} $v_{V_{M,n}}$,
		one for each set $V_{M,n}$ with $(M,n) \in P$,
\item the set $\mathcal A$ of \emph{arrangement vertices} $a_{\{\prec_Q\} ,V_{M,n}}$
       where $\{\prec_Q\}$ is the singleton set of the distinguished minimal order on $L_{G_{M,n}}$'s max cliques
       if $\mathcal K(L_{G_{M,n}})$ is not order isomorphic under its two
       linear orderings
       (recall the definition of $\mathcal K(L_{G_{M,n}})$ from Section~\ref{sec:L_G-can}).
      If $\mathcal K(L_{G_{M,n}})$ is order isomorphic under its two linear orderings, then max clique $Q$
		identifies an order $\prec_Q$, namely, the order where
                $Q_{L_{G_{M,n}}}$ occurs first (see Section~\ref{sec:L_G-can}
                for the definition of $Q_{L_{G_{M,n}}}$).
		$Q$ defines both orders if $Q_{L_{G_{M,n}}}$ is located in the middle.
	Thus, for each $Q$ the set $\{\prec_Q\}$ is the set of orders containing either only one of the isomorphic orders
	or both. Consequently, for each set $V_{M,n}$ there are at most three arrangement vertices of the form $a_{\{\prec_Q\} ,V_{M,n}}$.
 \item the set $\mathcal S$ of \emph{module vertices} $s_{W_A,V_{M,n}}$
			for which $A$ is a max clique of $G$,
and $W_A$ is the vertex in $L_{G_{M,n}}$ ($W_A$ is a module of $V_{M,n}$ with more than one vertex)
that contains vertices of $A$,
and
 \item $\{ s_{V} \}$, where $s_V$ is a special vertex acting as the root of $\mathcal T$.
\end{itemize}
We colour the vertices in $\mathcal V$ by assigning to each $v_{V_{M,n}}\in\mathcal V$ the ordered graph $\mathcal K (L_{G_{M,n}})$.
The vertices in $\mathcal A$ remain uncoloured and may
therefore be exchanged by an automorphism of $\mathcal T$
whenever their subtrees are isomorphic.
Each $s_{W_A,V_{M,n}}\in\mathcal S$ is coloured with the multiset of integers corresponding to the positions
that the max clique $A_{L_{G_{M,n}}}$ takes in the orders of $L_{G_{M,n}}$.
The edge relation $E_{\mathcal T}$ of $\mathcal T$ is now defined in a straight-forward manner, with all edges directed away from the root $s_V$.
\begin{itemize}[leftmargin=*]
 \item $s_V$ is connected to all $v_{V_{M,n}} \in \mathcal V$ with  $n=|V|$.
\item Each $v_{V_{M,n}} \in \mathcal V$ is connected to all vertices in $\mathcal A$ of the form $a_{\{\prec_Q\},V_{M,n}}$ with
		$Q\cap V_{M,n}\not=\emptyset$. Therefore,  $v_{V_{M,n}}$ is connected to at most three vertices.
\item Each $a_{\{\prec_Q\},V_{M,n}} \in\mathcal A$ is connected to all those $s_{W_A,V_{M,n}} \in \mathcal S$
	so that $\{\prec_Q\}$ is the set of orders of $L_{V_{M,n}}$ under which module
	$W_A\in V(L_{G_{M,n}})$ attains its minimal position, that is,
		for every max clique $Q$ that intersects with a module $W$ of $V_{M,n}$ with $|W|>1$,
		vertex~$a_{\{\prec_Q\},V_{M,n}} \in\mathcal A$ is connected to $s_{W_Q,V_{M,n}} \in \mathcal S$.
\item Every $s_{W_A,V_{M,n}} \in \mathcal S$ is connected to those $v_{V_{M'\!,n'}} \in \mathcal V$
	for which $V_{M'\!,n'}$ is a connected component of the module $W_A$, that is, for each max clique $A$ the vertex
	$s_{W_A,V_{M,n}} \in \mathcal S$ is connected to $v_{V_{A,n'}} \in \mathcal V$ with $n'= \max\{ m<n\mid (V_{A,m})\in P\}$.
\end{itemize}

The point of the arrangement vertices $\mathcal A$ is to ensure that the order of submodules is properly accounted for.
If our modular tree did not have such a safeguard, exchanging modules in symmetric positions might give rise to a non-isomorphic graph,
but it would not change the tree, so $\mathcal T$ would be useless for the task of distinguishing between these~two~graphs.

We will later need \STCC-definability of this coloured tree.
Thus, notice that the tree's vertices are equivalence classes, which are $\STCC$ definable. Also the edge relation
and the colours are \STCC-definable (Lemma~\ref{lem:LG-interval}).

Lemma~\ref{lem:modularTreeIsomImpIntGisom} below shows that our modular trees are a complete invariant of interval graphs,
so modular trees can be used to tell whether two interval graphs are isomorphic.

\begin{lemma}[\cite{koebler10interval},\cite{laubner11diss}]\label{lem:modularTreeIsomImpIntGisom}
Let $G$ and $H$ be interval graphs. If their modular trees are isomorphic, then so are $G$ and $H$.\qed
\end{lemma}

The graphs $L_{G_{M,n}}$ resemble the concept of \emph{overlap components} used in \cite{koebler10interval}
for the definition of a similar kind of modular tree. Overlap components are connected components of the subgraph of $G$
in which only those edges are present for which the neighbourhood of neither endpoint is contained in the neighbourhood of the other
(intuitively, their intervals \textit{overlap}).
It can be checked that overlap components and graphs $L_{G_{M,n}}$ only differ in the way they treat vertices
that are contained in just one max clique: overlap components treat them as further modules (which they trivially are),
the $L_{G_{M,n}}$ graphs directly put them into their unambiguous places.
In \cite{koebler10interval} the authors show Lemma~\ref{lem:modularTreeIsomImpIntGisom} for this similar kind of modular tree.
A detailed proof of Lemma~\ref{lem:modularTreeIsomImpIntGisom} can be found in \cite{laubner11diss}.

\subsection{Total Preorder on Coloured Directed Trees}
\label{sec:preorder-on-coloured-trees}

We can make use of the $\STCC$-definable modular decomposition tree,
and define a \emph{total preorder} on the vertices of  $\mathcal T_G$, that is,
	a linear order on the isomorphism classes of the (coloured) subtrees of $\mathcal T_G$ identified by its root vertices.

For our purposes,
we define a \emph{coloured directed tree} as a tuple $T=(V,E,L)$,
where $(V,E)$ is a directed tree and $L\subseteq V \times N(V)^2$
is a relation that assigns to each vertex $a \in V$
a colour $L_a := \{(m,n)\mid (a,m,n)\in L\}$.
It is easy to bring the coloured modular decomposition tree into this form.
For example, if $a$ is a component vertex, say $v_{V_{M,n}}$,
then $L_a$ consists of all tuples $(m,n)$,
where $(m,n)$ is an edge in the colour of $a$
(i.e., an edge in the canon of $L_{V_{M,n}}$ by which $a$ is coloured
 in $\mathcal T_G$).
Furthermore, if $a$ is a module vertex, say $s_{W_A,V_{M,n}}$,
then $L_a$ consists of all tuples $(m,n)$,
where $m$ occurs $n$ times in the colour of $a$.
In all other cases, we simply leave $L_a$ empty.

We let $\phi_\trianglelefteq(x,y)$ be the formula such that for all coloured directed trees $T$, assignments $\alpha$ and $a,b\in V(T)$:
 \[
(T,\alpha)\models \phi_\trianglelefteq[a,b]\ \iff\
\text{$L_a $ is lexicographically less than or equal to $L_b$.
}
\]
Then $\phi_\trianglelefteq$ defines a total preorder $\trianglelefteq$ on the vertices of any coloured directed tree.

Let  $\phi_{\prec}(x,y)$ and $\phi_{\cong}(x,y)$ be as defined in Section~\ref{sec:tree-order} and Section~\ref{sec:tree-iso}, respectively.
If we identify each subtree of a directed tree with its root vertex,
then the $\LREC_{=}$-formula $\phi_{\preceq}(x,y):= \phi_{\prec}(x,y) \lor \phi_{\cong}(x,y)$
defines a linear order $\preceq$ on the isomorphism classes of the subtrees of a directed tree.

We define a refinement $\preceq'$ of $\preceq$
  by letting $v \prec' w$ whenever $v \vartriangleleft w$,
  or: $v \trianglelefteq w$ and $w \trianglelefteq v$ and $v \prec w$.
  It should be obvious how to modify the formula $\psi_{\preceq}(x,y)$
  to an $\LREC_{=}$-formula $\psi_{\preceq'}$
  defining $\preceq'$.

\subsection{Canonisation}
This section deals with the canonisation of interval graphs, that is,
how to construct an $\LREC_=$-formula $\kappa'(p,q)$ such that for each interval graph $G$
we have $G\cong([|V(G)|],\kappa'[G;p,q])$. As a result we obtain the following:

\begin{theorem}
	$\LREC_=$ captures $\LOGSPACE$ on the class of all interval graphs.
\end{theorem}

We use the modular decomposition tree and the total preorder on its vertices for canonisation.
We apply l-recursion on the modular decomposition tree,
and as we have done for canonising trees we build the canon from the leaves to the root of the tree.
Recursively, we construct the canon by first building the disjoint union of the canons of the components of submodules,
then use the arrangement vertices to insert all submodules at the
correct side and build the canon of the corresponding component of a module.

In the following we explain the canonisation procedure in more detail.
The following lemma shows that it suffices to give
an $\LREC_=$-formula $\kappa(p,q)$ such that for every interval graph $G$
we have $G\cong([|V(G)|],\kappa[\mathcal T_G;p,q])$.
It follows from Lemma~\ref{lem:transduction-lemma}
and the fact that the coloured modular decomposition tree of an interval graph
is $\STCC$-definable.

\begin{lemma}\label{lem:modtree-total-order}
  If there exists an $\LREC_{=}$-formula $\kappa(p,q)$
  such that for all interval graphs $G$
  we have $G\cong([|V(G)|],\kappa[\mathcal T_G;p,q])$
  and $\kappa[\mathcal T_G;p,q]\subseteq [|V(G)|]^2$,
  then there also exists an $\LREC_{=}$-formula $\kappa'(p',q')$
  such that for all interval graphs $G$, $G\cong([|V(G)|],\kappa'[G;p',q'])$.
\end{lemma}

\begin{proof}
  As pointed out at the end of Section~\ref{sec:col-mod-dec-tree},
  the coloured modular decomposition tree of an interval graph $G$
  is definable in $\STCC$, and thus in $\LREC_=$.
  That is, there are $\LREC_=$-formulae $\theta_{V}(\tup{u})$,
  $\theta_{\approx}(\tup{u},\tup{v})$, $\theta_{E}(\tup{u},\tup{v})$
  and $\theta_{L}(\tup{u},\tup{q})$,
  where $\tup{u},\tup{v}$ are compatible tuples and
  $\tup{q}$ is a tuple of number variables,
  such that for all interval graphs $G$ and all assignments $\alpha$,
  \begin{itemize}[leftmargin=*]
  \item
    $\theta_{\approx}[G,\alpha;\tup{u},\tup{v}]$ is an equivalence
    relation $\approx$,
  \item
    $\theta_{V}[G,\alpha;\tup{u}]/_{\approx}$ is the set of vertices
    of $\mathcal T_G$,
  \item
    $\theta_{E}[G,\alpha;\tup{u},\tup{v}]/_{\approx} :=
     \set{(\tup{a}/_\approx,\tup{b}/_\approx) \mid
       (\tup{a},\tup{b}) \in \theta_{E}[G,\alpha;\tup{u},\tup{v}]}$
    is the edge relation of $\mathcal T_G$,
  \item
    and $\theta_{L}[G,\alpha;\tup{u},\tup{q}]/_{\approx}$ is the
    colour-relation of the modular decomposition tree $\mathcal T_G$.
  \end{itemize}
  We now apply Lemma~\ref{lem:transduction-lemma} with the transduction
  $
    \Theta =
    (\theta_{V}(\tup{u}), \theta_{\approx}(\tup{u},\tup{v}),
     \theta_{E}(\tup{u},\tup{v}),\theta_{L}(\tup{u},\tup{q}))
  $
  to obtain an $\LREC_=$-formula $\kappa^{-\Theta}(\tup{p}',\tup{q}')$
  such that for all $\tup{m},\tup{n} \in N(G)^{\len{\tup{u}}}$,
  $G \models \kappa^{-\Theta}[\tup{m},\tup{n}]$ iff
  $\num[G]{\tup{m}},\num[G]{\tup{n}} \in N(\Theta[G])$ and
  $\Theta[G] \models \kappa[\num[G]{\tup{m}},\num[G]{\tup{n}}]$.
  Note that $\Theta[G] = \mathcal T_G$.
  As $\kappa[\mathcal T_G;p,q]\subseteq [|V(G)|]^2$,
  the condition $\num[G]{\tup{m}},\num[G]{\tup{n}} \in N(\Theta[G])$
  can be replaced by $\num[G]{\tup{m}},\num[G]{\tup{n}} \in N(G)$.
  Hence, the tuples $\tup{p}',\tup{q}'$ of number variables
  in $\kappa^{-\Theta}$ can be identified with single number variables $p',q'$,
  which yields the desired formula $\kappa'(p',q')$.
\end{proof}

In general, the canonisation procedure is similar to the one of directed trees.
To apply l-recursion we use a graph $\graphG=(\graphV,\graphE)$ with
labels $\graphC(v)\subseteq \nat$ for all $v\in \graphV$.
We let $\graphV:=V(\mathcal T_G)\times N(\mathcal T_G)^2$ be the vertices of $\graphG$
and
for all component vertices $v_{V_{M,n}}\in\mathcal V$,
$(v_{V_{M,n}},p,q)\in V$ stands for  ``$(p,q)\in X_{v_{V_{M,n}}}\text{?}$'', where $X_{v_{V_{M,n}}}$
is the edge relation of an isomorphic copy $([|V_{M,n}|],X_{v_{V_{M,n}}})$ of $G_{M,n}$.

In the following we explain the edge relation $\graphE$ and labels $\graphC$ of graph $\graphG$.

\bigskip\noindent\textit{Edges introduced by module vertices.}\\
In $\mathcal T_G$, each vertex  $s_{W_A,V_{M,n}}\in \mathcal S$ is connected to those $v_{V_{M'\!,n'}}\in \mathcal V$ for which $V_{M'\!,n'}$
	is a connected component of the module $W_A$. Thus,
	we can use the available total preorder $\prec'$ on the children of
        $s_{W_A,V_{M,n}}$ (cf.\ Section~\ref{sec:preorder-on-coloured-trees})
	to construct the canon of the disjoint union of the children's canons from the canons of the children.
	For a vertex $s\in \mathcal S$ and a child $v:=v_{V_{M,n}}\in \mathcal V$ of $s$, let $D_{v}$ be the set of all children $v'$
	of $s$ with $v'\prec' v$, and $e_{v}$ be the number of children $v'$ of $s$
        defining modules isomorphic to $V_{M,n}$ (i.e., $v' \preceq' v$ and $v
        \preceq' v'$).
	For all $p,q\in N(\mathcal T_G)^2$ and all $i\in[0,e_{v}-1]$, we let $\tup{a}:=(s,p_{v,i}+p,p_{v,i}+q)$ have an edge to $(v,p,q)$
	where $p_{v,i}:=\sum_{v_{V_{M'\!,n'}}\in D_{v}}|V_{M'\!,n'}|+i\cdot|V_{M,n}|$ and define $\graphC(\tup{a})=\{e_{v}\}$.
	Notice that here we can have an in-degree greater than 1.

\bigskip\noindent\textit{Edges introduced by arrangement vertices.}\\
Let us consider a vertex $a_{\{\prec_Q\},V_{M,n}} \in\mathcal A$.
	Its children in $\mathcal T_G$ are vertices $s_{W_A,V_{M,n}}$ for specific submodules of the module $V_{M,n}$,
	and we need to integrate the canons of them into the canon $\mathcal K(L_{V_{M,n}})$ of $L_{V_{M,n}}$.
	The canon $\mathcal K(L_{V_{M,n}})$ is $\STCC$-definable (Lemma~\ref{lem:LG-interval}) and we assume it to be
	assigned to the first part $[1,|V(L_{V_{M,n}})|]$ of the number sort.
	Notice that on the number sort we have a distinguished ordering $<_N$ of the max cliques.

	If $a_{\{\prec_Q\},V_{M,n}} \in\mathcal A$ has no sibling, then
	we have a distinguished order of the max cliques of $L_{V_{M,n}}$,
	and we can integrate each canon of a submodule into $\mathcal K(L_{V_{M,n}})$ according to the colour of its vertex $s$.
	By integrating a submodule,
	we mean the following:
	We first sum up the size of $\mathcal K(L_{V_{M,n}})$ and the sizes of all submodules defined by children of  $a_{\{\prec_Q\},V_{M,n}}$
	with smaller colours,
	and increase each vertex of the canon of the submodule by this number.
	Further, in the canon  $\mathcal K(L_{V_{M,n}})$ we want to replace the smallest vertex $z$ that
	lies in the max clique that is at the position defined by the colour of $s$ and in no other max clique
	by the modified canon of the submodule.
	In order to do that, we add an edge between all vertices that are adjacent to $z$ and all vertices of the modified canon of the submodule.
	For $a_{\{\prec_Q\},V_{M,n}}$,
        we define the out-going edges of $a_{\{\prec_Q\},V_{M,n}}$ in $\graphG$ such that,
        if $X$ denotes the relation defined by the final $\LREC_=$-formula,
        the graph with edge relation $\{(p,q)\in N(T_G)^2\mid
        ((a_{\{\prec_Q\},V_{M,n}},p,q),\ell)\in X\ \text{for large enough $\ell$}\}$ consists of
	the modified canons of the submodules and all new edges.
	Note that we have not yet removed the replaced vertices.

	If $a_{\{\prec_Q\},V_{M,n}}$ has siblings, a single child, and the colour of the single child contains two equal positions,
	we know we have to insert the canon of its child in the middle
	(regarding the ordering of the max cliques) of  $\mathcal K(L_{V_{M,n}})$.
	For such a  vertex $a_{\{\prec_Q\},V_{M,n}}$  we construct the edges of $\graphG$ so that we obtain
	the following graph on the number sort:
	We add the size of $\mathcal K(L_{V_{M,n}})$ to each vertex of the canon of the submodule,
	and add all edges that would be generated if we inserted the modified canon
	into the canon  $\mathcal K(L_{V_{M,n}})$ replacing the smallest vertex in the middle max clique.

	Now, let $a_{\{\prec_{Q_1}\},V_{M,n}}$ and $a_{\{\prec_{Q_2}\},V_{M,n}}$
	be siblings where the colour of at least one child contains different positions.
	We determine their order with respect to the total preordering.
	Say, $a_{\{\prec_{Q_1}\},V_{M,n}}\prec' a_{\{\prec_{Q_2}\},V_{M,n}}$.
	Then we want to integrate the submodules of $a_{\{\prec_{Q_1}\},V_{M,n}}$ all
	into the first half (regarding $<_N$) of the max cliques of canon $\mathcal K(L_{V_{M,n}})$, and
	the submodules of $a_{\{\prec_{Q_2}\},V_{M,n}}$ into the second half.
	Therefore, we create the edges of $\graphG$ in such a way that the graph on the number sort at vertex $a_{\{\prec_{Q_1}\},V_{M,n}}$
	is as follows:
	Each child of vertex $a_{\{\prec_{Q_1}\},V_{M,n}}$ represents a certain submodule of $V_{M,n}$.
	We
	sum up the size of $\mathcal K(L_{V_{M,n}})$, the size of the submodule in the middle if it exists, and
	the sizes of all submodules defined by children of  $a_{\{\prec_{Q_1}\},V_{M,n}}$
	with smaller colours, and add this value to
	each vertex of the canon of this certain submodule.
	Finally, we insert each of these modified canons into the max clique specified
	by the smaller value contained in the colour of
	the corresponding vertex, in the same way we did before,
	and add all newly created edges to the modified canons of the submodules.
	For  vertex $a_{\{\prec_{Q_2}\},V_{M,n}}$ we construct the graph on the number sort equivalently,
	only that we additionally add the sum of the sizes of all submodules defined by children of  $a_{\{\prec_{Q_1}\},V_{M,n}}$
	to the vertices of the canons of the submodules.

	If $a_{\{\prec_{Q_1}\},V_{M,n}}$ and $a_{\{\prec_{Q_2}\},V_{M,n}}$ are equivalent with respect to the total preorder, we
	insert the submodules of $a_{\{\prec_{Q_i}\},V_{M,n}}$  for $i=1,2$ at both sides.
	We position the submodules according to both values that are contained in their colours.
	Thus, if there is no submodule that belongs in the middle at vertex $a_{\{\prec_{Q_i}\},V_{M,n}}$, for $i=1,2$,
	the edge relation of $\graphG$ almost enables us to define
	the canon of module $V_{M,n}$,
	except that we still need to remove the vertices that were replaced.

	For each  $a\in \mathcal A$ that fits in the last case, we let $\graphC(a,p,q)=\set{2}$,
	otherwise $\graphC(a,p,q)=\set{1}$, for all  $p,q\in N(T_G)$.
	Note, that only in the last case, we obtain in-degrees larger than 1, that is, there the in-degree is 2.

\bigskip\noindent\textit{Edges introduced by component vertices.}\\
  Let $v=v_{V_{M,n}} \in \mathcal V$.
  In the preceding step, we introduced edges for arrangement vertices
  $a_{\set{\prec_Q},V_{M,n}}$ so that,
  if $X$ denotes the relation defined by the final $\LREC_=$-formula
  in an interval graph whose coloured modular decomposition tree is $\mathcal T_G$,
  the graph with edge relation
  $\{(p,q)\in N(T_G)^2\mid ((a_{\{\prec_Q\},V_{M,n}},p,q),\ell)\in X\ \text{for
    large enough $\ell$}\}$
  is almost a canon of $V_{M,n}$;
  we still need to insert $\mathcal K(L_{V_{M,n}})$,
  and remove the vertices of $\mathcal K(L_{V_{M,n}})$
  that correspond to the submodules of $V_{M,n}$.

  Recall from Lemma~\ref{lem:LG-interval} that
  the canon $\mathcal K(L_{V_{M,n}})$ is $\STCC$-definable.
  The set of vertices of $\mathcal K(L_{V_{M,n}})$ is $[1,|V(L_{V_{M,n}})|]$.
  Let $R$ be the set of vertices that have to be removed
  from $\mathcal{K}(V_{M,n})$,
  so that the resulting graph plus the edges from $\{(p,q)\in N(T_G)^2\mid ((a_{\{\prec_Q\},V_{M,n}},p,q),\ell)\in X\ \text{for large enough $\ell$}\}$
  is isomorphic to $V_{M,n}$.
  It is easy to define $R$ by considering the different cases as we did above.

  Let $f(r):=r-d_r$, where $d_r=|\{s\in R\mid s<r\}|$.
  Then, the contracted canon $\mathcal Q:=\{(f(p),f(q))\mid (p,q)\in\mathcal K(L_{V_{M,n}})\}$
	is $\STCC$-definable,
        and we assign it to the first part $[1,|V(L_{V_{M,n}})|-|R|]$ of the number sort.
	Thus, we set $\graphC(v,f(p),f(q))=N(\mathcal T_G)$ for all $(p,q)\in\mathcal K(L_{V_{M,n}})$.
	Furthermore, for each child $a\in \mathcal A$ of $v$, we include all edges from $(v,f(p),f(q))$ to $(a,p,q)$ for all $p,q\in N(T_G)\setminus R$.
	Finally, for all $(p,q)\not \in\mathcal K(L_{V_{M,n}})$ we set $\graphC(v,f(p),f(q))=\set{1}$.

\bigskip\noindent\textit{Finishing the construction.}\\
In order to actually  perform l-recursion we need sufficient ``resources''.
Taking a look at the in-degrees, we notice that they are only larger than one
when we treat isomorphic connected components while building the disjoint union, or when the graph $V_{M,n}$ is symmetric and we insert the submodules twice at both sides.
Either way, an incoming degree of $d$ means that we insert $d$ disjoint isomorphic copies into the graph on the number sort.
Hence, it suffices to use a binary resource term.

\medskip

\begin{remark}
\label{rem:interval-canon-and-LREC}
  It is possible to show that there is no $\LREC{+}\TC[\set{E}]$-sentence $\phi$
  such that for all connected interval graphs $G_1,G_2$
  we have $G_1 \disjunion G_2 \models \phi$
  if and only if $G_1 \isomorphic G_2$.
  The proof is based on similar ideas as the proof of Theorem~\ref{thm:reach-nondef}.
\end{remark}

\section{Conclusion}
\label{sec:concl}

We introduce the new logics $\LREC$ and $\LREC_=$, extending first-order logic
with counting by a recursion operator that can be evaluated in logarithmic
space. By capturing \LOGSPACE\ on trees and interval graphs, we obtain the
first nontrivial descriptive characterisations of \LOGSPACE\ on natural classes of unordered
structures. It would be interesting to extend our results to further classes
of structures such as the class of planar graphs or classes of graphs of
bounded tree width.

The expressive power of $\LREC_=$ is not yet well-understood. For example, it
is an open question whether directed graph reachability is expressible in
$\LREC_=$, and even whether $\LREC_=$ has the same expressive power as
$\FPC$. (Of course assumptions from complexity theory indicate that the answer
to both questions is negative.) It is also an open question whether
reachability on undirected trees is expressible in plain \LREC.

It is obvious that our capturing results can be transferred to
nondeterministic logarithmic space \NL\ by adding a transitive closure
operator to the logic. However, it would be much nicer to have a natural
``nondeterministic'' variant of our limited recursion operator that allows it
to express directed graph reachability and thus yields a logic that contains
\TC. We leave it as an open problem to find such an operator.

\bibliographystyle{plain}
\bibliography{bibliography}

\end{document}